\documentclass[runningheads,envcountsect,envcountsame]{llncs}

\makeatletter
\renewcommand*\l@author[2]{}
\renewcommand*\l@title[2]{}
\makeatletter

\usepackage[english]{babel}
\usepackage[ascii]{inputenc}
\usepackage[T1]{fontenc}
\usepackage{xspace}
\usepackage{graphicx}

\usepackage{verbatim}
\usepackage{url}
\usepackage[usenames]{color}
\usepackage{picinpar}
\usepackage{times}
\usepackage{cancel}
\usepackage{bm}
\usepackage{latexsym}
\usepackage[backgroundcolor=green!20, textsize=scriptsize]{todonotes}
\usepackage[verbose]{wrapfig}
\def\mwrapfigR{\wrapfigure[4]{R}{3cm}\ifdim\pagetotal=0pt \else\vspace{-\intextsep}\fi}

\usepackage{hyperref}

\usepackage{amsmath}
\usepackage{amssymb}
\usepackage{stmaryrd}
\usepackage{amsfonts}
\usepackage{mathtools}

\usepackage[capitalize]{cleveref}

\usepackage{tikz}
\usetikzlibrary{positioning, arrows}
\usetikzlibrary{patterns}
\usetikzlibrary{shapes}
\usetikzlibrary{decorations.pathreplacing}
\usetikzlibrary{fadings}
\usetikzlibrary{intersections}

\makeatletter
\newcommand{\leqnomode}{\tagsleft@true}
\newcommand{\reqnomode}{\tagsleft@false}
\makeatother

\newcommand\GaD{\!{\implies}\!}
\newcommand\DaG{\!{\impliedby}\!}

\newcommand{\eqE}{\mathrel{E}}

\newcommand\tta{{\mathtt{a}}}
\newcommand\ttb{{\mathtt{b}}}
\newcommand\ttc{{\mathtt{c}}}
\newcommand\Pcal{\mathcal P}
\newcommand\Jcal{\mathcal J}

\newcommand\nat{\mathbb{N}}
\newcommand\Nat{\nat}
\newcommand\natom{\nat_{\omega}}

\newcommand\Pf{\Pcal_{\!\mathit{f}}}
\renewcommand{\setminus}{\smallsetminus}
\newcommand\upc{\mathop{\uparrow}\nolimits}
\newcommand\dc{\mathop{\downarrow}\nolimits}
\newcommand\dwc{\mathop{\downarrow}\nolimits}
\newcommand{\subword}{\preccurlyeq}

\newcommand\ibar{{\bar{\imath}}}
\newcommand\exi{{\textrm{ex}_1}}
\newcommand\exd{{\textrm{ex}_2}}
\newcommand\ext{{\textrm{ex}_3}}
\newcommand{\egdef}{\stackrel{\mbox{\begin{scriptsize}def\end{scriptsize}}}{=}}

\newcommand{\equivdef}{\stackrel{\mbox{\begin{scriptsize}def\end{scriptsize}}}{\Leftrightarrow}}
\newcommand{\lequiv}{\iff}

\newcommand\Idl{\mathop{\mathit{Idl}}\nolimits}
\newcommand\Up{\mathop{\mathit{Up}}\nolimits}
\newcommand\Down{\mathop{\mathit{Down}}\nolimits}
\newcommand\Fil{\mathop{\mathit{Fil}}\nolimits}
\newcommand\tup[1]{\langle #1 \rangle}
\newcommand\tupsum[1]{\langle #1 \rangle}
\DeclarePairedDelimiter\size{\lvert}{\rvert}

\newcommand{\qo}{QO\xspace}
\newcommand{\qos}{QOs\xspace}
\newcommand{\wqo}{WQO\xspace}
\newcommand{\wqos}{WQOs\xspace}

\newcommand{\PP}{\text{PP}}
\newcommand{\OD}{\eqref{OD}\xspace} 
\newcommand{\ID}{\eqref{ID}\xspace} 
\newcommand{\IDsanspar}{\ref{ID}\xspace} 
\newcommand{\XF}{\eqref{XF}\xspace} 
\newcommand{\XI}{\eqref{XI}\xspace} 
\newcommand{\CF}{\eqref{CF}\xspace} 
\newcommand{\CFsanspar}{\ref{CF}\xspace} 
\newcommand{\CFbare}{\text{CF}\xspace} 
\newcommand{\CI}{\eqref{CI}\xspace} 
\newcommand{\CIsanspar}{\ref{CI}\xspace} 
\newcommand{\IF}{\eqref{IF}\xspace} 
\newcommand{\IFsanspar}{\ref{IF}\xspace} 
\newcommand{\II}{\eqref{II}\xspace} 
\newcommand{\IIsanspar}{\ref{II}\xspace} 
\newcommand{\PI}{\eqref{PI}\xspace} 
\newcommand{\CD}{(\textrm{CD})\xspace}

\newcommand\bu{\bm{u}}
\newcommand\bv{\bm{v}}
\newcommand\bw{\bm{w}}

\newcommand\bt{\bm{t}}

\newcommand\bP{\bm{P}}
\newcommand\bQ{\bm{Q}}
\newcommand\bU{\bm{U}}
\newcommand\bV{\bm{V}}
\newcommand\bD{\bm{D}}
\newcommand\bT{\bm{T}}
\newcommand\bx{\bm{x}}
\newcommand\by{\bm{y}}
\newcommand\ba{\bm{a}}

\newcommand\bB{\bm{B}}

\newcommand\bA{\bm{A}}

\newcommand\bS{\bm{S}}

\newcommand\bepsilon{\bm{\epsilon}}

\newcommand\dwcp{\dwc_{\le'}}
\newcommand\ci{{\mathcal Cl}_{\mathrm{I}}}
\newcommand\cf{{\mathcal Cl}_{\mathrm{F}}}

\newcommand{\leqst}{\leq_{\mathrm{st}}}
\newcommand{\gest}{\ge_{\mathrm{st}}}
\newcommand{\upcst}{\upc_{\mathrm{st}}}
\newcommand{\simst}{\sim_{\mathrm{st}}}
\newcommand{\equivst}{\equiv_{\mathrm{st}}}

\newcommand{\leqcj}{\leq_{\mathrm{cj}}}
\newcommand{\simcj}{\sim_{\mathrm{cj}}}

\newcommand{\si}{\mathcal{S}_{\mathrm{I}}}
\newcommand{\sff}{\mathcal{S}_{\mathrm{F}}}

\newcommand\mult[1]{#1^{\circledast}}
\newcommand\multX{\mult{X}}

\newcommand\lemb{\le_{\mathrm{emb}}}

\newcommand\lex{\le_{\mathrm{lex}}}

\newcommand\hoare{\sqsubseteq_H}
\newcommand\invhoare{\sqsupseteq_H}

\begin{document}

\title{
The Ideal Approach to Computing\\
Closed Subsets in Well-Quasi-Orderings\thanks{Supported by ANR project BRAVAS (grant ANR-17-CE40-0028).}
}

\author{
  J.~Goubault-Larrecq \inst{1}
\and
  S.~Halfon\inst{1}
\and
  P.~Karandikar \inst{2,3}
\and
  K.~Narayan Kumar \inst{2}
\and
  Ph.~Schnoebelen \inst{3,1}
}
\authorrunning{$\!\!\!\!\!\!\!\!\!\!\!\!\!\!\!\!\!\!$J.\ Goubault-Larrecq, S.\ Halfon, P.\ Karandikar, K.\ Narayan Kumar, and Ph.\ Schnoebelen}
\institute{
  ENS Paris-Saclay, France
  \and
  CMI, Chennai, India
  \and
  CNRS, France
}

 \maketitle

\begin{abstract}
Elegant and general algorithms for handling upwards-closed and downwards-closed
subsets of \wqos can be developed using the filter-based and ideal-based
representation for these sets.
These algorithms can be built in a generic or parameterized way, in parallel
with the way complex \wqos are obtained by combining or modifying simpler
\wqos.
\end{abstract}

\section{Introduction}
\label{sec-intro}

The theory of well-quasi-orderings (\wqos for short) has proved useful
in many areas of mathematics, logic, combinatorics, and computer
science. In computer science, it appears prominently in termination
proofs~\cite{dershowitz87}, in formal
languages~\cite{dalessandro2008}, in graph algorithms (e.g., via the
Graph Minor Theorem~\cite{lovasz2006}), in program verification (e.g., with
well-structured systems~\cite{abdulla2000c,finkel98b,SS-concur13}), automated deduction,
distributed computing, but also in machine learning~\cite{akama2011},
program transformation~\cite{leuschel2002}, etc. We refer to~\cite{kriz90} for
``four [main] reasons to be interested in \wqo theory''.

In computer science, tools from WQO theory were commonly seen as
lacking algorithmic contents. This situation is changing. For example,
tight complexity bounds for \wqo-based algorithms have recently been established
and are now used when
comparing logics or computational models~\cite{HSS-lics2012,schmitz2014,schmitz-toct2016,SS-icalp11}. As
another example, the field of well-structured
systems  grows not just by the identification of new
families of models, but also by the
development of new generic algorithms based on \wqo theory, see,
e.g., \cite{BS-fmsd2013,blondin2018}.

In this chapter we are concerned with the issue of reasoning about, and
computing with,
downwards-closed and upwards-closed subsets of a \wqo. These sets appear
in program verification (prominently in model-checking of well structured
systems~\cite{BS-fmsd2013}, in verification of Petri
nets~\cite{habermehl2010}, in separability
problems~\cite{goubault2016,zetzsche2018},
but also as an effective abstraction
tool~\cite{bachmeier2015,zetzsche2015}).  The question of how to
handle downwards-closed subsets of \wqos \emph{in a generic way} was first
raised by Geeraerts et al.: in~\cite{geeraerts06} the authors
postulated the existence of an \emph{adequate domain of limits}
satisfying some representation conditions. It turns out that the \emph{ideals} of
\wqos always satisfy these conditions, and usually enjoy further
algorithmic properties.

\paragraph{Outline of this chapter.}
We start by recalling, as a motivating example,
the algorithmic techniques that have been successfully used to handle
upwards-closed and downwards-closed subsets in two different \wqos: the
tuples of natural numbers with component-wise ordering, and the set of
finite words with subword ordering.
We then describe the fundamental structures that underlie these
algorithms and propose in \cref{sec-effective} a generic set of
effectiveness assumptions on which the algorithms can be based.

The second part of the chapter, \cref{sec-basic-eff,sec-more-constr},
shows how many examples of \wqos used in applications fulfill the
required effectiveness assumptions. Since in practice complex \wqos
are most often obtained by composing or modifying simpler \wqos, our
strategy for showing their effectiveness involves proving that \wqo
constructors preserve effectiveness.

A final section discusses our choices ---of effectiveness assumptions
and of algorithms--- and lists some of the first questions raised by
our approach.

\paragraph{Genesis of this chapter.}
This text grew from~\cite{VJGL09} (unpublished) where
Goubault-Larrecq proposed a notion of effective \wqos, and where
\cref{thm-vj} was first proven. There, Goubault-Larrecq also shows
that products, sequence extensions, and tree extensions of effective
\wqos are effective.  Then, in 2016 and 2017,  Karandikar,
Narayan Kumar and Schnoebelen developed the framework and handled \wqos
obtained by extensions, by quotients, and by substructures. Finally,
in 2017 and 2018, Halfon joined the project and contributed most
of the results on powersets and multisets. He also studied variant
sets of axioms for effective \wqos as reported in
\cref{sec-axioms}. In the meantime, the constructions initiated
by~\cite{VJGL09} have been used in several papers, starting
from~\cite{FWD-WSTS-1,FWD-WSTS-2}, and
including~\cite{blondin2017b,blondin2018,finkel2016,goubault2016,lazic2015,lazic2016b,leroux2015,leroux2016}.

\section{Well-quasi-orderings, ideals, and some motivations}
\label{sec-basics}

A \emph{quasi-ordering} (a \qo) $(X,{\leq})$ is a set $X$ equipped with a
reflexive and transitive relation. We write $x<y$ when $x\leq y$ and
$y\not\leq x$, and $x\equiv y$ when $x\leq y$ and $y\leq x$.
For $S\subseteq X$, we let $\upc S$ and $\dc S$ denote
the upward and downward closures, respectively, of $S$ in $X$.
Formally, $\upc S \egdef \{x\in X\mid\exists y\in S: y\leq x\}$ and
$\dc S \egdef \{x\in X\mid\exists y\in S: x\leq y\}$.
We will also use $\dwc_< S$ and $\upc_< S$ to collect elements that
are strictly above, or below, elements of $S$, i.e., $\dwc_< S \egdef
\{x\in X\mid\exists y\in S: x < y\}$ and similarly for $\upc_< S$.

When $S=\{x\}$ is a singleton, we may simply write $\upc x$ or
$\dc_< x$. A subset of $X$ of the form $\upc x$ is called a
\emph{principal filter} while a subset of the form $\dc x$ is a
\emph{principal ideal}.  A subset $S\subseteq X$ is
\emph{upwards-closed} when $S=\upc S$, and \emph{downwards-closed}
when $S=\dc S$. Note that arbitrary unions and intersections of
upwards-closed (resp.\ downwards-closed) sets are upwards-closed
(resp.\ downwards-closed).  Observe also that the complement of an
upwards-closed set is downwards-closed, and conversely.  We write
$\Up(X)$ for the set of upwards-closed subsets of $X$, with typical
elements $U$, $U'$, $V$,~\@\ldots{} Similarly, $\Down(X)$ denotes the set of
its downwards-closed subsets, with typical elements $D$, $D'$, $E$,~\@\ldots{}

\subsection{Two motivating examples}
\label{sec-motivations}

Consider the set $X=\nat^2$ of pairs of natural numbers. These are the
points with integral coordinates in the upper-right quadrant. We order
these points with the coordinate-wise ordering, also called
\emph{product ordering}:
\[
\tup{a,b}\leq\tup{a',b}
\equivdef a\leq a' \land b\leq b'
\:.
\]
Note that this is only a partial ordering: $\tup{1,2}$ and $\tup{3,0}$
are incomparable.

\subsubsection{$\pmb{\nat}^2$ and its upwards-closed subsets.}
In many applications, we need to consider upwards-closed subsets $U$,
$U'$,~\ldots, of $\nat^2$. These may be defined by simple, or not so
simple, constraints such as $U_{\exi}$ and $V_{\exi}$ in
\cref{fig-exmp1}.

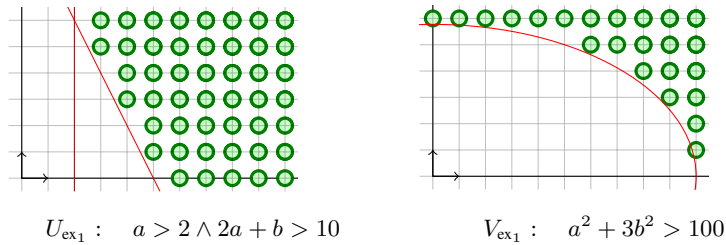
\begin{figure}[htbp]
\begin{center}
\tikzfading[name=fade right, left color=transparent!100, right color=transparent!0]
\tikzfading[name=fade up, bottom color=transparent!100, top color=transparent!0]
\begin{tikzpicture}[auto,scale=0.35]
  \def\blot{[very thick, fill=green!20,draw=green!50!black] circle (8pt)}
  \draw[-] (0,9) -- (0,0) -- (13,0);
  \draw[<->] (0,1) -- (0,0) -- (1,0);
  \foreach \i in {3,...,13} \foreach \j in {5,...,9} \filldraw (\i,\j) \blot;
  \foreach \i in {4,...,13} \foreach \j in {3,...,9} \filldraw (\i,\j) \blot;
  \foreach \i in {5,...,13} \foreach \j in {1,...,9} \filldraw (\i,\j) \blot;
  \foreach \i in {6,...,13} \foreach \j in {0,...,9} \filldraw (\i,\j) \blot;

  \draw[-,red,thin] (2,9.5) -- (2,-0.5);
  \draw[-,red,thin] (0.25,9.5) -- (5.25,-0.5);

  \draw[nearly transparent, step=1,black] (-0.5,-0.5) grid (13.5,9.5);
  \fill [path fading=fade right,white] (10.5,-0.5) rectangle (13.5,9.5);
  \fill [path fading=fade up,white] (-0.5,6.5) rectangle (13.5,9.5);

\node at (6.5,-2){$U_{\exi}:\quad a>2 \land 2a+b > 10$};
\end{tikzpicture}
 \quad
\tikzfading[name=fade right, left color=transparent!100, right color=transparent!0]
\tikzfading[name=fade up, bottom color=transparent!100, top color=transparent!0]
\begin{tikzpicture}[auto,scale=0.35]
  \def\blot{[very thick, fill=green!20,draw=green!50!black] circle (8pt)}
  \draw[-] (0,9) -- (0,0) -- (13,0);
  \draw[<->] (0,1) -- (0,0) -- (1,0);
  \foreach \i in {0,...,13} \foreach \j in {6,...,9} \filldraw (\i,\j) \blot;
  \foreach \i in {6,...,13} \foreach \j in {5,...,9} \filldraw (\i,\j) \blot;
  \foreach \i in {8,...,13} \foreach \j in {4,...,9} \filldraw (\i,\j) \blot;
  \foreach \i in {9,...,13} \foreach \j in {3,...,9} \filldraw (\i,\j) \blot;
  \foreach \i in {10,...,13} \foreach \j in {1,...,9} \filldraw (\i,\j) \blot;
  \foreach \i in {11,...,13} \foreach \j in {0,...,9} \filldraw (\i,\j) \blot;

  \begin{scope}
  \clip (-0.5,-0.5) -- (-0.5,9) -- (13,9) -- (13,-0.5) -- cycle;
  \draw[name path=ellipse,red,thin]
	  (0,0) circle[x radius = 10, y radius = 5.7735];
  \end{scope}

  \draw[nearly transparent, step=1,black] (-0.5,-0.5) grid (13.5,9.5);
  \fill [path fading=fade right,white] (10.5,-0.5) rectangle (13.5,9.5);
  \fill [path fading=fade up,white] (-0.5,6.5) rectangle (13.5,9.5);

\node at (6.5,-2){$V_{\exi}:\quad a^2+3b^2> 100$};
\end{tikzpicture}
 \end{center}
\caption{Two upwards-closed subsets of $\nat^2$}
\label{fig-exmp1}
\end{figure}

A striking aspect of these depictions of $U_{\exi}$ and $V_{\exi}$ ---see also
\cref{fig-exmp2}--- is that both can be seen  as unions
of a few principal filters:
\begin{xalignat*}{1}
U_{\exi}&= \upc\tup{3,5} \cup  \upc\tup{4,3} \cup  \upc\tup{5,1} \cup
\upc\tup{6,0}
\:,
\\
V_{\exi}&=\upc\tup{0,6}\cup \upc\tup{6,5} \cup  \upc\tup{8,4} \cup
\upc\tup{9,3} \cup  \upc\tup{10,1} \cup  \upc\tup{11,0}
\:.
\end{xalignat*}
We write $U=\bigcup_{i<n}\upc x_i$ to say that the upwards-closed
subset $U$ of $X$ is the union of $\upc x_0$,~\ldots, $\upc x_{n-1}$.
The elements $x_i$ are the \emph{generators}, and the finite set
$\{x_0, \cdots, x_{n-1}\}$ is a \emph{finite basis} of $U$.  We also say
that $\bigcup_{i<n} x_i$ is a \emph{finite basis representation} of
$U$.  By removing elements that are not minimal, we obtain a
\emph{minimal finite basis} of $U$.

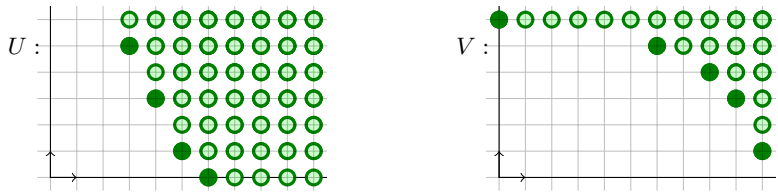
\begin{figure}[htbp]
\begin{center}
\tikzfading[name=fade right, left color=transparent!100, right color=transparent!0]
\tikzfading[name=fade up, bottom color=transparent!100, top color=transparent!0]
\begin{tikzpicture}[auto,scale=0.35]
  \def\blot{[very thick, fill=green!20,draw=green!50!black] circle (8pt)}
  \def\minblot{[very thick, fill=green!50!black,draw=green!50!black] circle (8pt)}
  \draw[-] (0,9) -- (0,0) -- (13,0);
  \draw[<->] (0,1) -- (0,0) -- (1,0);
  \foreach \i in {3,...,13} \foreach \j in {5,...,9} \filldraw (\i,\j) \blot;
  \filldraw (3,5) \minblot;
  \foreach \i in {4,...,13} \foreach \j in {3,...,9} \filldraw (\i,\j) \blot;
  \filldraw (4,3) \minblot;
  \foreach \i in {5,...,13} \foreach \j in {1,...,9} \filldraw (\i,\j) \blot;
  \filldraw (5,1) \minblot;
  \foreach \i in {6,...,13} \foreach \j in {0,...,9} \filldraw (\i,\j) \blot;
  \filldraw (6,0) \minblot;

  \draw[nearly transparent, step=1,black] (-0.5,-0.5) grid (13.5,9.5);
  \fill [path fading=fade right,white] (10.5,-0.5) rectangle (13.5,9.5);
  \fill [path fading=fade up,white] (-0.5,6.5) rectangle (13.5,9.5);

\node at (-1,5){$U:$};
\end{tikzpicture}
 \quad
\tikzfading[name=fade right, left color=transparent!100, right color=transparent!0]
\tikzfading[name=fade up, bottom color=transparent!100, top color=transparent!0]
\begin{tikzpicture}[auto,scale=0.35]
  \def\blot{[very thick, fill=green!20,draw=green!50!black] circle (8pt)}
  \def\minblot{[very thick, fill=green!50!black,draw=green!50!black] circle (8pt)}
  \draw[-] (0,9) -- (0,0) -- (13,0);
  \draw[<->] (0,1) -- (0,0) -- (1,0);
  \foreach \i in {0,...,13} \foreach \j in {6,...,9} \filldraw (\i,\j) \blot;
  \filldraw (0,6) \minblot;
  \foreach \i in {6,...,13} \foreach \j in {5,...,9} \filldraw (\i,\j) \blot;
  \filldraw (6,5) \minblot;
  \foreach \i in {8,...,13} \foreach \j in {4,...,9} \filldraw (\i,\j) \blot;
  \filldraw (8,4) \minblot;
  \foreach \i in {9,...,13} \foreach \j in {3,...,9} \filldraw (\i,\j) \blot;
  \filldraw (9,3) \minblot;
  \foreach \i in {10,...,13} \foreach \j in {1,...,9} \filldraw (\i,\j) \blot;
  \filldraw (10,1) \minblot;
  \foreach \i in {11,...,13} \foreach \j in {0,...,9} \filldraw (\i,\j) \blot;
  \filldraw (11,0) \minblot;

  \draw[nearly transparent, step=1,black] (-0.5,-0.5) grid (13.5,9.5);
  \fill [path fading=fade right,white] (10.5,-0.5) rectangle (13.5,9.5);
  \fill [path fading=fade up,white] (-0.5,6.5) rectangle (13.5,9.5);

\node at (-1,5){$V:$};
\end{tikzpicture}
 \caption{Finite bases for $U_{\exi}$ and $V_{\exi}$}
\label{fig-exmp2}
\end{center}
\end{figure}

We shall see later that all upwards-closed subsets of
$(\nat^2,{\leq})$ admit such a representation.  For the time being we
want to stress how this representation of upwards-closed subsets is
convenient \emph{from an algorithmic viewpoint}.  To begin with, it
provides us with a finite data structure for subsets that are infinite
and thus cannot be represented in extension on a
computer. Interestingly, some important set-theoretical operations are
very easy to perform on this representation: testing whether some
point $\tup{a,b}$ is in $U$ or $V$ just amounts to comparing
$\tup{a,b}$ with the points forming the basis of $U$ or $V$
respectively. Testing whether $U\subseteq V$ reduces to checking
whether all points in the basis of $U$ belong to $V$. We see that
$U_{\exi}\not\subseteq V_{\exi}$ since there is a point in
$U_{\exi}$'s base that is not in $V_{\exi}$, i.e., not larger than (or
equal to) any of the points in $V_{\exi}$'s basis: for instance
$\tup{3,5}\not\in V_{\exi}$.
Similarly,
$\tup{0,6}\not\in U_{\exi}$ hence $V_{\exi}\not\subseteq U_{\exi}$.

Two further operations that are easily performed are computing
$W=U\cup V$ and $W'=U\cap V$ for upwards-closed $U$ and $V$ (recall
that such unions and intersections are upwards-closed as observed
earlier). For $U\cup V$, we just join the two finite bases and
(optionally) remove any element that is not minimal. For example
\begin{align*}
U_{\exi}\cup V_{\exi} =& \: \bigl(\upc\tup{3,5} \cup  \upc\tup{4,3} \cup  \upc\tup{5,1} \cup \upc\tup{6,0}\bigr)
\\
& \cup   \bigl(\upc\tup{0,6} \cup \cancel{\upc\tup{6,5}} \cup  \cancel{\upc\tup{8,4}} \cup
\cancel{\upc\tup{9,3}} \cup  \cancel{\upc\tup{10,1}} \cup  \cancel{\upc\tup{11,0}}\bigr)
\\
= &  \upc\tup{0,6}  \cup \upc\tup{3,5} \cup  \upc\tup{4,3} \cup  \upc\tup{5,1}  \cup
\upc\tup{6,0}
\:.
\end{align*}
For $U\cap V$, we first observe that principal filters
can be intersected with
\begin{gather}
\label{eq-intersection-of-filters-in-nat2}
\upc\tup{a,b}\cap\upc\tup{a',b'} =\upc\tup{\max(a,a'),\max(b,b')}
\end{gather}
and then use the distributivity law
$(\bigcup_{i<n}\upc x_i)\cap (\bigcup_{j<m}\upc y_j)
=\bigcup_{i,j}(\upc x_i\cap \upc y_j)$
to handle the general case.
This gives, for example,
\begin{align*}
U_{\exi} \cap V_{\exi} =&
\bigl[\upc\tup{3,5}\cap\upc\tup{0,6}\bigr]
\:\cup\:
\bigl[\upc\tup{3,5}\cap\upc\tup{6,5}\bigr]
\:\cup\:
\bigl[\upc\tup{3,5}\cap\upc\tup{8,4}\bigr]
\\
&\!\cup\:
\bigl[\upc\tup{4,3}\cap\upc\tup{9,3}\bigr]
\:\cup\:
\bigl[\upc\tup{5,1}\cap\upc\tup{10,1}\bigr]
\:\cup\:
\bigl[\upc\tup{6,0}\cap\upc\tup{11,0}\bigr]
\\
&\!\cup\:
\cdots \text{ \emph{more filters on elements that are not minimal} } \cdots
\\
=&\upc\tup{3,6}
\cup
\upc\tup{6,5}
\cup
\upc\tup{8,4}
\cup
\upc\tup{9,3}
\cup
\upc\tup{10,1}
\cup
\upc\tup{11,0}
\:.
\end{align*}

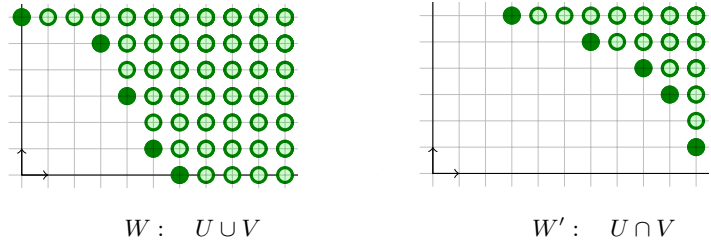
\begin{figure}[htbp]
\begin{center}
\tikzfading[name=fade right, left color=transparent!100, right color=transparent!0]
\tikzfading[name=fade up, bottom color=transparent!100, top color=transparent!0]
\begin{tikzpicture}[auto,scale=0.35]
\node at (6.5,-2){$W:\quad U\cup V$};
  \def\blot{[very thick, fill=green!20,draw=green!50!black] circle (8pt)}
  \def\minblot{[very thick, fill=green!50!black,draw=green!50!black] circle (8pt)}  
  \draw[-] (0,9) -- (0,0) -- (13,0);
  \draw[<->] (0,1) -- (0,0) -- (1,0);
  \foreach \i in {0,...,13} \foreach \j in {6,...,9} \filldraw (\i,\j) \blot;
  \filldraw (0,6) \minblot;
  \foreach \i in {3,...,13} \foreach \j in {5,...,9} \filldraw (\i,\j) \blot;
  \filldraw (3,5) \minblot;
  \foreach \i in {4,...,13} \foreach \j in {3,...,9} \filldraw (\i,\j) \blot;
  \filldraw (4,3) \minblot;
  \foreach \i in {5,...,13} \foreach \j in {1,...,9} \filldraw (\i,\j) \blot;
  \filldraw (5,1) \minblot;
  \foreach \i in {6,...,13} \foreach \j in {0,...,9} \filldraw (\i,\j) \blot;
  \filldraw (6,0) \minblot;

  \draw[nearly transparent, step=1,black] (-0.5,-0.5) grid (13.5,9.5);
  \fill [path fading=fade right,white] (10.5,-0.5) rectangle (13.5,9.5);
  \fill [path fading=fade up,white] (-0.5,6.5) rectangle (13.5,9.5);
\end{tikzpicture}
 \quad
\tikzfading[name=fade right, left color=transparent!100, right color=transparent!0]
\tikzfading[name=fade up, bottom color=transparent!100, top color=transparent!0]
\begin{tikzpicture}[auto,scale=0.35]
\node at (6.5,-2){$W':\quad U\cap V$};
  \def\blot{[very thick, fill=green!20,draw=green!50!black] circle (8pt)}
  \def\minblot{[very thick, fill=green!50!black,draw=green!50!black] circle (8pt)}  
  \draw[-] (0,9) -- (0,0) -- (13,0);
  \draw[<->] (0,1) -- (0,0) -- (1,0);
  \foreach \i in {3,...,13} \foreach \j in {6,...,9} \filldraw (\i,\j) \blot;
  \filldraw (3,6) \minblot;
  \foreach \i in {6,...,13} \foreach \j in {5,...,9} \filldraw (\i,\j) \blot;
  \filldraw (6,5) \minblot;
  \foreach \i in {8,...,13} \foreach \j in {4,...,9} \filldraw (\i,\j) \blot;
  \filldraw (8,4) \minblot;
  \foreach \i in {9,...,13} \foreach \j in {3,...,9} \filldraw (\i,\j) \blot;
  \filldraw (9,3) \minblot;
  \foreach \i in {10,...,13} \foreach \j in {1,...,9} \filldraw (\i,\j) \blot;
  \filldraw (10,1) \minblot;
  \foreach \i in {11,...,13} \foreach \j in {0,...,9} \filldraw (\i,\j) \blot;
  \filldraw (11,0) \minblot;

  \draw[nearly transparent, step=1,black] (-0.5,-0.5) grid (13.5,9.5);
  \fill [path fading=fade right,white] (10.5,-0.5) rectangle (13.5,9.5);
  \fill [path fading=fade up,white] (-0.5,6.5) rectangle (13.5,9.5);
\end{tikzpicture}
 \caption{Computing intersections and unions via finite bases}
\label{fig-exmp2bis}
\end{center}
\end{figure}

Finally, a last feature of the finite basis representation for
upwards-closed subsets of $\nat^2$ is that, if we only consider
minimal bases, namely bases of incomparable elements ---in essence, if
we systematically remove unnecessary generators that are subsumed by
smaller generators,--- then the representation is \emph{canonical}:
there is a unique way of representing any $U\in
\Up(\nat^2)$. Algorithmically, this allows one to implement
the required structures using \emph{hash-consing}
\cite{Deutsch:ipv,Goto:HLISP}, where structures with the same contents
are allocated at the same address, with the help of auxiliary
hash-tables.  In particular, finite \emph{sets} can be implemented
this way, efficiently \cite{JG:ML}.  Equality tests can then be
performed in constant time, notably.

\subsubsection{Words and their subwords.}

Our second example comes from formal languages
and combinatorics~\cite{sakarovitch83}. 
\begin{mwrapfigR}
\fbox{
\begin{tikzpicture}[->,auto,inner sep=0pt]

\node (w1) [label=above:$\ttb$] at (0,0) {};
\node (w2) [right=1.1ex of w1,label=above:$\tta$] {};
\node (w3) [right=1.1ex of w2,label=above:$\ttc$] {};
\node (w4) [right=1.1ex of w3,label=above:$\tta$] {};
\node (w5) [right=1.1ex of w4,label=above:$\ttb$] {};
\node (w6) [right=1.1ex of w5,label=above:$\tta$] {};
\node (w7) [right=1.1ex of w6,label=above:$\ttb$] {};
\node (w8) [right=1.1ex of w7,label=above:$\quad$] {};
\node (qm) [below=.1em of w8,inner sep=1pt] {?};

\node (u1) [below=5ex of w1,label=above:$\tta$] at (0,0) {};
\node (u2) [right=1.1ex of u1,label=above:$\ttb$] {};
\node (u3) [right=1.1ex of u2,label=above:$\ttb$] {};
\node (u4) [right=1.1ex of u3,label=above:$\tta$] {};

\node (wlab) [left=2em of w1,label=above:$v:$] {};
\node (ulab) [below=5ex of wlab,label=above:$u:$] {};

\path (u1) ++ (.3ex,1.8ex) edge (w2);
\path (u2) ++ (.3ex,1.8ex) edge (w5);
\path (u3) ++ (.3ex,1.8ex) edge (w7);
\path (u4) ++ (.5ex,1.5ex) edge[dashed] (qm);

\end{tikzpicture}
}
 \end{mwrapfigR}
Let
us fix a three-letter alphabet $A=\{\tta,\ttb,\ttc\}$ and write
$A^*=\{u,v,\cdots\}$ for the set of all finite words over
$A$.
Standardly, the empty word is denoted by $\epsilon$, concatenation is denoted
multiplicatively, and $\size{u}$ is the length of $u$. We write $u\subword v$
when $u$ is a \emph{subword} of $v$, i.e., a subsequence: $u$ can be obtained
from $v$ by erasing some (occurrences of) letters. It is easy to check whether
$u\subword v$ by attempting to construct a leftmost embedding of $u$ into $v$:
this only requires at most one traversal of $u$ and $v$ and takes time linear
in $\size{u}+\size{v}$.  For example, the box to the right shows that
$u=\tta\ttb\ttb\tta$ is not a subword of $v=\ttb\tta\ttc\tta\ttb\tta\ttb$.

With the subword ordering comes the notion of upwards-closed and
downwards-closed languages (i.e., sets of words). For example the
language $U_{\exd}\subseteq A^*$ of words with at least one $\tta$ and
at least two $\ttb$s is upwards-closed, as is $V_{\exd}$, the language
of words with length at least $2$.  These upwards-closed languages
occur in many applications and one would like to know good data
structures and algorithms for manipulating them. It turns out that any
such upwards-closed language can be represented as a finite union of
principal filters.\footnote{This result is known as Haines'
  Theorem~\cite{haines69}, and is also a consequence of Higman's
  Lemma: see \cref{sec-sequences}.} For example, $U_{\exd}$ and
$V_{\exd}$ can be written
\begin{xalignat*}{2}
U_{\exd} &=\upc \tta\ttb\ttb\cup \upc \ttb\tta\ttb\cup \upc \ttb\ttb\tta
\:,
&
V_{\exd} &=\upc \tta\tta\cup \upc \tta\ttb \cup \cdots \cup \upc \ttc\ttc = \bigcup_{\size{u}=2}\upc u
\:.
\end{xalignat*}
In the subword setting, a principal filter is always a regular
language. Indeed, for any $u\in A^*$, of the form
$u=a_1a_2\cdots a_{\ell}$, one has
$\upc u= A^* a_1 A^* a_2 A^* \cdots A^* a_{\ell} A^*$, which is a
language at level $\frac{1}{2}$ in the Straubing-Th\'erien
hierarchy~\cite{pin96}. Being simple star-free regular languages, the
upwards-closed subsets can be handled with well-known
automata-theoretic techniques. However, one can also use the same
simple ideas we used for $\Nat^2$: testing $U\subseteq V$ reduces to
comparing the generators, computing unions is trivial, and bases made
of incomparable words provide a canonical representation.
Finally, computing
intersections reduces to intersecting principal filters, exactly as in
$\nat^2$. For this, we
observe that $\upc u\cap \upc v$ is generated by the minimal words
that contain both $u$ and $v$ as subwords. This set of minimal words, written $u\sqcap v$, is
called the \emph{infiltration product} of $u$ and
$v$~\cite{chen58}. For example $\tta\ttb\sqcap
\ttc\tta=\{\tta\ttb\ttc\tta,\tta\ttc\ttb\tta,\tta\ttc\tta\ttb,\ttc\tta\ttb\}$. Infiltrations are a generalization of
shuffles and we shall describe a simple algorithm for a generalized
infiltration product in \cref{sec-sequences}.

\subsection{What about downwards-closed subsets?}
\label{sec-dwc-nat}

With the previous two examples, we showed how it is natural and easy to
work with upwards-closed subsets of a quasi-ordered set when these
subsets are represented as a finite union $\bigcup_{i<n} \upc x_i$
of principal filters.

Let us now return to our previous setting, $X=\Nat^2$, and look at the
downwards-closed subsets $D,E,\ldots\in\Down(\nat^2)$.  As an example,
consider $D_{\exi} \egdef \nat^2\setminus U_{\exi}$ and
$E_{\exi}\egdef \nat^2\setminus V_{\exi}$.  We shall sometimes write
$D_{\exi}=\neg U_{\exi}$ and $E_{\exi}=\neg V_{\exi}$.

\begin{center}
\tikzfading[name=fade right, left color=transparent!100, right color=transparent!0]
\tikzfading[name=fade up, bottom color=transparent!100, top color=transparent!0]
\begin{tikzpicture}[auto,scale=0.35]
\node at (6.5,-2){$D_{\exi}:\quad a\leq 2\lor 2a+b\leq 10$};
  \def\dblot{[very thick, fill=blue!20,draw=blue!50!black] circle (8pt)}
  \def\maxdblot{[very thick, fill=blue!50!black,draw=blue!50!black] circle (8pt)}
  \draw[-] (0,9) -- (0,0) -- (13,0);
  \draw[<->] (0,1) -- (0,0) -- (1,0);
  \foreach \i in {0,...,2} \foreach \j in {0,...,9} \filldraw (\i,\j) \dblot;
  \foreach \i in {3} \foreach \j in {0,...,4} \filldraw (\i,\j) \dblot;
  \filldraw (3,4) \maxdblot;
  \foreach \i in {4} \foreach \j in {0,...,2} \filldraw (\i,\j) \dblot;
  \filldraw (4,2) \maxdblot;
  \foreach \i in {5} \foreach \j in {0} \filldraw (\i,\j) \dblot;
  \filldraw (5,0) \maxdblot;

  \draw[-,red,thin] (2,9.5) -- (2,-0.5);
  \draw[-,red,thin] (0.25,9.5) -- (5.25,-0.5);

  \draw[nearly transparent, step=1,black] (-0.5,-0.5) grid (13.5,9.5);
  \fill [path fading=fade right,white] (10.5,-0.5) rectangle (13.5,9.5);
  \fill [path fading=fade up,white] (-0.5,6.5) rectangle (13.5,9.5);
\end{tikzpicture}
 \quad
\tikzfading[name=fade right, left color=transparent!100, right color=transparent!0]
\tikzfading[name=fade up, bottom color=transparent!100, top color=transparent!0]
\begin{tikzpicture}[auto,scale=0.35]
\node at (6.5,-2){$E_{\exi}:\quad a^2+3b^2\leq 100$};
  \def\dblot{[very thick, fill=blue!20,draw=blue!50!black] circle (8pt)}
  \def\maxdblot{[very thick, fill=blue!50!black,draw=blue!50!black] circle (8pt)}
  \draw[-] (0,9) -- (0,0) -- (13,0);
  \draw[<->] (0,1) -- (0,0) -- (1,0);
  \foreach \i in {0,...,5} \foreach \j in {0,...,5} \filldraw (\i,\j) \dblot;
  \filldraw (5,5) \maxdblot;
  \foreach \i in {0,...,7} \foreach \j in {0,...,4} \filldraw (\i,\j) \dblot;
  \filldraw (7,4) \maxdblot;
  \foreach \i in {0,...,8} \foreach \j in {0,...,3} \filldraw (\i,\j) \dblot;
  \filldraw (8,3) \maxdblot;
  \foreach \i in {0,...,9} \foreach \j in {0,...,2} \filldraw (\i,\j) \dblot;
  \filldraw (9,2) \maxdblot;
  \foreach \i in {0,...,10} \foreach \j in {0,...,0} \filldraw (\i,\j) \dblot;  
  \filldraw (10,0) \maxdblot;

  \begin{scope}
  \clip (-0.5,-0.5) -- (-0.5,9) -- (13,9) -- (13,-0.5) -- cycle;
  \draw[name path=ellipse,red,thin]
	  (0,0) circle[x radius = 10, y radius = 5.7735];
  \end{scope}

  \draw[nearly transparent, step=1,black] (-0.5,-0.5) grid (13.5,9.5);
  \fill [path fading=fade right,white] (10.5,-0.5) rectangle (13.5,9.5);
  \fill [path fading=fade up,white] (-0.5,6.5) rectangle (13.5,9.5);
\end{tikzpicture}
 \end{center}

Here, $E_{\exi}$ can be represented using its maximal points as generators:
\begin{align*}
E_{\exi} &=
\dc\tup{5,5}\cup \dc\tup{7,4}\cup \dc\tup{8,3}\cup \dc\tup{9,2}\cup \dc\tup{10,0}
\:.
\end{align*}
Representing downwards-closed sets via a finite ``basis'', i.e., as a
finite union of principal ideals, of the form $\bigcup_{i<n} \dc x_i$,
allows for simple and efficient algorithms, exactly as for
upwards-closed subsets: one tests inclusion by comparing the generators
of the ideals, and computes unions by gathering all generators and
(optionally) removing non-maximal ones. For intersections one uses
\begin{gather}
\label{eq-intersection-of-dcs-in-nat2}
\dc\tup{a,b}\cap\dc\tup{a',b'} =\dc\tup{\min(a,a'),\min(b,b')}
\end{gather}
and the distribution law $(\bigcup_i \dc x_i)\bigcap(\bigcup_j \dc y_j)=
\bigcup_i\bigcup_j(\dc x_i\cap \dc y_j)$, valid in every \qo.

However, there is an important limitation here that we did not have with
upwards-closed subsets: not all downwards-closed subsets in $\nat^2$ can be
generated from finitely many elements. Indeed, for any $x\in\nat^2$, the
ideal $\dc x$ is finite and thus only the finite downwards-closed subsets of
$\nat^2$ can be represented via principal ideals.
Hence $D_{\exi}$ in the previous figure, or even $\nat^2$ itself, while perfectly
downwards-closed, cannot be represented in this way.
\\

A possible solution is to represent a downwards-closed subset $D\in\Down(X)$ via
the finite basis of its upwards-closed complement, writing
$D=X\setminus\bigcup_{i<n}\upc x_i$, or also $D=X(\setminus \upc
x_i)_{i<n}$. Continuing our example, $D_{\exi}=\neg U_{\exi}$ can be written $D_{\exi}=
\nat^2\setminus \upc\tup{3,5} \setminus \upc\tup{4,3} \setminus \upc\tup{5,1}
\setminus \upc\tup{6,0} $.
This representation \emph{by excluded minors} is contrapositive and thus
counter-intuitive. Computing intersections become easier while unions become
harder, which is usually not what we want in applications.
More annoyingly, constructing a representation of $\dc x$ from $x$
involves actually computing complements, a task that can be difficult
in general as we shall see later. Even in the easy $\nat^2$ case, it
is not transparent how from, e.g., $x=\tup{2,3}$, one gets to $\dc
x=\nat^2\setminus\upc \tup{0,4}\setminus\upc\tup{3,0}$.

\subsubsection{Downwards-closed subsets with $\omega$'s.}
In the case of $\nat^2$, there exists an elegant solution to the
representation problem for downwards-closed sets: one use pairs
$\tup{a,b}\in\natom^2$ where $\natom$ extends $\nat$ with an extra
value $\omega$ that is larger than all natural numbers. We can now
denote $D_{\exi}=\neg U_{\exi}$ (see last figure) as
$\dc\tup{2,\omega}\cup\dc\tup{3,4} \cup\dc\tup{4,2}
\cup\dc\tup{5,0}$. We note that $\dc\tup{2,\omega}$ should probably be
written more explicitly as $(\dc 2)\times\Nat$ since it denotes
$\{\tup{c,d}~|a\leq 2\land b\in\nat\}$, a subset of $\nat^2$, not of
$\natom^2$, however the $\omega$-notation inherited from vector
addition systems \cite{karp69} is now well-entrenched and we retain it
here.

The sets of the form $\dc\tup{a,b}$ where $a,b\in\natom$ are the
\emph{ideals}\footnote{We shall soon give the general definition. For
now, the reader has to accept the $\nat^2$ case.
} of $\nat^2$, and we see that they comprise the principal
ideals as a special case. They also comprise infinite
subsets and, for example, $\nat^2=\dc\tup{\omega,\omega}$ is one of
them.

Using such ideals, all the downwards-closed subsets of $\nat^2$ can be
represented, and the algorithms for membership, inclusions, unions and
intersections are just minor extensions of what we showed for finite downwards-closed sets, when all generators were
proper elements of $\nat^2$.  The only difference
is that we have to handle $\omega$'s in the obvious way when comparing
generators (e.g., in inclusion tests) and when computing $\min$'s,
e.g., in~\eqref{eq-intersection-of-dcs-in-nat2}.  Additionally,
and like for upwards-closed subsets of $\nat^2$, the representation of
downwards-closed sets by the downward closure of incomparable elements
is canonical, which here too brings in important algorithmic benefits.
\\

Now that we have finite representations for both upwards-closed
and downwards-closed subsets of $\nat^2$, it is natural to ask whether we can
compute complements.

It turns out that,  for $\nat^2$, this is an easy task. For
complementing filters, one uses
\begin{align}
\label{eq-compl-filter-nat2}
\neg\upc\tup{a,b} =
\left\{\begin{array}{cl}
\dc\tup{a-1,\omega} & \text{ if $a>0$}
\\
\emptyset & \text{otherwise}
\end{array}\right\}
\bigcup
\left\{\begin{array}{cl}
\dc\tup{\omega,b-1} & \text{ if $b>0$}
\\
\emptyset & \text{otherwise}
\end{array}\right\}
\:.
\end{align}
We see where the $\omega$'s are needed. In fact, only $\neg
\upc\tup{0,0}=\emptyset$ does not involve $\omega$'s. We note that
$\emptyset$, a downwards-closed subset, is indeed a finite union of
ideals: it is the empty union.

Complementing an ideal is also easy:
\begin{align}
\label{eq-compl-ideal-nat2}
\neg\dc\tup{a,b} =
\left\{\begin{array}{cl}
\upc\tup{a+1,0} & \text{ if $a<\omega$}
\\
\emptyset & \text{otherwise}
\end{array}\right\}
\bigcup
\left\{\begin{array}{cl}
\upc\tup{0,b+1} & \text{ if $b<\omega$}
\\
\emptyset & \text{otherwise}
\end{array}\right\}
\:.
\end{align}
We see here that complementing an ideal in $\nat^2$ always returns a
union of principal filters, with no $\omega$'s.

Complementing an arbitrary upwards-closed subset $U$ is easy
if $U=\bigcup_{i<n}\upc x_i$ is given as a finite union of filters: we
compute $\bigcap_{i<n}(X\setminus \upc x_i)$.  This needs complementing
filters and intersecting downwards-closed sets, two operations we know
how to perform on $\nat^2$.  Complementing an arbitrary
downwards-closed subset $D=\bigcup_{i<n} \dc x_i$ is done similarly,
even with $x_i\in\natom^2$: we complement each ideal and intersect
the resulting upwards-closed sets.

Finally, let us observe that, since
any upwards-closed set is a finite union of filters,
the proof that the complement
$\neg\upc\tup{a,b}$ of any filter, and the intersection of any two
ideals of $\nat^2$, can be expressed as a finite union of ideals,
entails that any downwards-closed $D\in\Down(\nat^2)$ is a finite
union of ideals, a result known as \emph{expressive completeness}.

\subsubsection{Downwards-closed sets of subwords.}

What about downwards-closed sets in $(A^*,{\subword})$?  As with
$\nat^2$, finite unions of principal ideals, of the form $\dc u_1\cup
\cdots \cup \dc u_{\ell}$, are easy to compare and combine but they
can only describe the finite downwards-closed languages.  The
contrapositive representation by excluded minors can describe any
downwards-closed set but here too it is cumbersome. For example, let us
fix $A=\{\tta,\ttb,\ttc\}$ and consider the language $D_{\ext}=\tta^*\ttb^*$,
i.e., the set of all words composed of any number of $\tta$'s followed
by any number of $\ttb$'s: it is clear that $D_{\ext}$ is closed by taking
subwords, hence $D_{\ext}\in\Down(A^*)$.
Its representation by excluded minors is
$ D_{\ext} = \neg (\upc \ttb\tta \cup \upc \ttc) $.  That is, ``a word
$w\in A^*$ is in $D_{\ext}$ iff it does not contain any $\ttc$, nor
some $\ttb$ before an $\tta$'': arguably, using $\tta^*\ttb^*$ to
denote $D_{\ext}$ is clearer.

We do not develop this example further, and just announce that indeed
the regular expression $\tta^*\ttb^*$ denotes an ideal of
$(A^*,{\subword})$, as we shall show in
\cref{sec-sequences}. Furthermore, and as with $\natom^2$, algorithms
for comparing ideals in $A^*$ are similar to algorithms that
compare elements of $A^*$. For example, testing whether (the
language denoted by) $\tta^*\ttb^*$ is a subset of
$\ttb^*\ttc^*\tta^*$ is essentially like testing whether
$\tta\ttb$ is a subword of $\ttb\ttc\tta$.

\subsection{Well-quasi-orders}

The previous section has made it clear that writing upwards-closed
sets as a finite union of principal filters, when possible, is
handy to compute with those sets. The quasi-orders for which it is possible to represent all
upwards-closed sets as such is known: it is the class of well-quasi
orders, which we introduce below.

A \qo $(X,{\leq})$ is \emph{well-founded} $\equivdef$ it does not
contain an infinite strictly decreasing sequence
$x_0>x_1>x_2>\cdots$. A subset $S \subseteq X$ is an \emph{antichain}
if for all distinct $x,y \in S$, neither of $x \leq y$ and $y \leq x$
holds. A \qo is \emph{well} (\wqo) $\equivdef$ it is well-founded and
does not contain an infinite antichain. Equivalently, $(X,{\leq})$ is
\wqo iff every infinite sequence $(x_i)_{i\in\nat}$ contains an
infinite monotonic subsequence
$x_{i_0}\leq x_{i_1}\leq x_{i_2}\leq \cdots$ with
$i_0<i_1<i_2<\cdots$.  See~\cite{kruskal72,SS-esslli2012} for proofs
and other equivalent characterizations.

\begin{example}[Some well-known \wqos]~\\
\begin{description}
\item[linear orderings:] $(\nat,{\leq})$ is a \wqo, as is every
  ordinal or every well-founded linear-ordering.
\item[words and sequences:] $(\Sigma^*,{\subword})$, the set of words
  over a finite alphabet with the (scattered) subword ordering is a
  \wqo.  Variants and extensions
  abound~\cite{dalessandro2008,HSS-lmcs,zetzsche2018}.  By Higman's
  Lemma, for any \wqo $(X,{\leq})$, its sequence extension ordered by
  embedding, $(X^*,{\leq_*})$, is a \wqo too.
\item[powersets:] $(\Pf(X),{\hoare})$, the set of all \emph{finite}
  subsets of $(X,{\leq})$ with Hoare's subset embedding is a \wqo when
  $X$ is.  The full powerset $\Pcal(X)$ is a \wqo if $X$ is an
  $\omega^2$-\wqo, a slightly stronger requirement than just being
  \wqo, see~\cite{marcone2001}.
\item[trees:] Labeled finite trees ordered by embedding form a \wqo
  (Kruskal's Tree Theorem \cite{Kruskal:tree}).
\item[graphs:] Finite graphs ordered by the minor relation constitute
  a \wqo (Robertson \& Seymour's Graph Minor Theorem \cite{RS:minor}).
\qed
\end{description}
\end{example}

Coming back to our motivation, here is the result claimed at the
beginning of this section:
\begin{lemma}[Finite basis property]
  \label{lem-FBP}
  If $(X,{\leq})$ is \wqo then every upwards-closed $U\in \Up(X)$
  contains a \emph{finite basis} $B\subseteq U$ such that
  $U=\bigcup_{x\in B}\upc x$.
\end{lemma}
It is easy to see that the converse holds: if every upwards-closed set
has a finite basis, then $(X,{\leq})$ is \wqo.

\Cref{lem-FBP} validates our choice of representing
sets via a finite set of generators, as we did
in our two motivating examples.
It also entails that, when $X$ is a countable \wqo, $\Up(X)$ is
countable too, as is  $\Down(X)$ since complementation  bijectively relates
upwards-closed and downwards-closed subsets
(see~\cite{bonnet75} for a more general statement).

We conclude this section by mentioning another useful characterization of \wqos, see~\cite{kruskal72}.
\begin{lemma}[Ascending/Descending chain condition]
\label{lem-no-chains}
\label{lem-ACC}
If $(X,{\leq})$ is \wqo then there exists no infinite strictly increasing
sequence $U_0\subsetneq U_1\subsetneq U_2\subsetneq \cdots$ of
upwards-closed subsets. Dually, there exists no infinite strictly decreasing sequence
$D_0\supsetneq D_1\supsetneq D_2\supsetneq \cdots$ of downwards-closed
subsets.

In other words, $(\Up(X),{\supseteq})$ and $(\Down(X),{\subseteq})$ are
well-founded posets.
\end{lemma}

\subsection{Canonical prime decompositions of closed subsets}

We now recall some basic facts about the canonical decompositions of
upwards-closed and downwards-closed subsets in prime components.

Let $(X,{\leq})$ be a \wqo. We use $\Up$ and $\Down$ as abbreviations
for $\Up(X)$ and $\Down(X)$.

\begin{definition}[Prime subsets]
  \label{def-prime}
  1.\ A non-empty $U\in \Up$ is \emph{(up) prime} if for any
  $U_1,U_2\in\Up$, $U\subseteq(U_1\cup U_2)$ implies $U\subseteq U_1$
  or $U\subseteq U_2$.

  \noindent
  2.\ Similarly, a non-empty $D\in \Down$ is \emph{(down) prime} if
  $D\subseteq(D_1\cup D_2)$ implies $D\subseteq D_1$ or
  $D\subseteq D_2$.
\end{definition}
Observe that all principal filters are up prime and all principal
ideals are down prime.  Note also that, by definition, the empty
subset is not prime.

\begin{lemma}[Irreducibility]
  \label{lem-irreduc}
  1.\ $U\in\Up$ is prime if, and only if, for all
  $U_1,\ldots,U_n\in\Up$, $U=U_1\cup \cdots\cup U_n$ implies $U=U_i$
  for some $i$.

  \noindent
  2.\ $D\in\Down$ is prime if, and only if, for all
  $D_1,\ldots,D_n\in\Down$, $D=D_1\cup \cdots\cup D_n$ implies $D=D_i$
  for some $i$.
\end{lemma}

The following lemma highlights the importance of prime subsets.
\begin{lemma}
\label{lem-fin-dec-exists}
1.\  Every upwards-closed set $U\in\Up$ is a finite union of up primes.

\noindent
2.\  Every downwards-closed set $D\in\Down$ is a finite union of down
primes.
\end{lemma}
\begin{proof}
1.\ is trivial:  the finite basis property
of \wqos (\cref{lem-FBP}) shows that any upwards-closed set is a
finite union of filters. 

2.\ is a classical result, going back to Noether, see \cite[Chapter
VIII, Corollary, p.181]{Birkhoff:latt}.  We include a proof for the
reader's convenience.  That proceeds by well-founded induction on $D$
in the well-founded poset $(\Down,{\subseteq})$ (\cref{lem-ACC}).
If $D$ is empty, then it is an empty (hence finite) union of
primes. If $D$ is prime, the claim holds trivially.  Finally, if
$D\not=\emptyset$ is not prime, then by \cref{lem-irreduc} it can
be written as $D=D_1\cup \cdots\cup D_n$ where each $D_i$ is properly
contained in $D$. By induction hypothesis each $D_i$ is a finite
union of primes. Hence $D$ is too.
\qed
\end{proof}
We say that a finite collection $\{P_1,\cdots,P_n\}$ of up (resp.\
down) primes is a \emph{decomposition} of $U\in\Up$ (resp., of
$D\in\Down$) if $U=P_1\cup\cdots\cup P_n$ (resp.,
$D=P_1\cup\cdots\cup P_n$). The decomposition is \emph{minimal} if
$P_i\subseteq P_j$ implies $i=j$.

\begin{theorem}[Canonical decomposition]
  Any upwards-closed $U$ (resp.\ downwards-closed $D$) has a finite
  minimal decomposition. Furthermore this minimal decomposition is
  unique. We call it the canonical decomposition of $U$ (resp.\ $D$).
\end{theorem}
\begin{proof}
  By \cref{lem-fin-dec-exists}, any $U$ (or $D$) has a finite
  decomposition: $U$ (or $D$) $=\bigcup_{i=1}^n P_i$. The
  decomposition can be made minimal by removing any $P_i$ that is
  strictly included in some $P_j$. To prove uniqueness we assume that
  $\bigcup_{i=1}^n P_i=\bigcup_{j=1}^m P'_j$ are two minimal
  decompositions. From $P_i\subseteq \bigcup_j P'_j$, and since $P_i$
  is prime, we deduce that $P_i\subseteq P'_{k_i}$ for some $k_i$.
  Similarly, for each $P'_j$ there is $\ell_j$ such that
  $P'_j\subseteq P_{\ell_j}$. The inclusions
  $P_i\subseteq P'_{k_i}\subseteq P_{\ell_{k_i}}$ require
  $i=\ell_{k_i}$ by minimality of the decomposition, hence are
  equalities $P_i=P'_{k_i}$. Similarly $j=k_{\ell_j}$ and
  $P'_j=P_{\ell_j}$ for any $j$. This one-to-one correspondence shows
  $\{P_1,\cdots,P_n\}=\{P'_1,\cdots,P'_m\}$.
  \qed
\end{proof}

\subsection{Filter decompositions and ideal decompositions}

\begin{definition}[Ideals]
  A subset $S$ of $X$ is an \emph{ideal} it if is non-empty,
  downwards-closed, and directed.  We write $\Idl(X)=\{I,J,\cdots\}$
  for the set of all ideals of $X$.

  Recall that $S$ is \emph{directed} if for all $x_1, x_2 \in S$,
  there exists $x \in S$ such that $x_1 \leq x$ and $x_2 \leq x$.
\end{definition}
A \emph{filter} is a non-empty, upwards-closed, and filtered set $S$,
where \emph{filtered} means that for all $x_1, x_2 \in S$, there
exists $x \in S$ such that $x \leq x_1, x_2$.  In a \wqo, the filters
are exactly the principal filters, hence there is no need to introduce
a new notion.  We write $\Fil (X)$ for the set of all (principal)
filters of $X$.

Every principal ideal $\dc x$ is directed hence is an ideal.  However
not all ideals are principal.  For example, in $(\nat,{\leq})$, the
set $\nat$ itself is an ideal (it is directed) and not of the form
$\dwc n$ for any $n \in \nat$.

\begin{remark}
  \label{rem-idl-iso-sobr}
  The above example illustrates a general phenomenon: the limit of an
  monotonic sequence of ideals (more generally, of a directed family
  of ideals) is an ideal.  In particular, if $x_0<x_1<x_2<\cdots$ is
  an infinite increasing sequence, $\bigcup_{i=0,1,2,\ldots}\dc x_i$
  is an ideal. It can be seen as the downward closure of a limit
  point, e.g. when one writes things like
  ``$\bigcup_{n\in\nat}\dc n=\dc\omega$''. It turns out that
  $(\Idl(X),{\subseteq})$, the domain-theoretical \emph{ideal
    completion} of $X$, is isomorphic to the \emph{sobrification}
  $(\widehat X,{\leq})$ ---a topological comple\-tion--- of
  $(X,{\leq})$, see~\cite{FWD-WSTS-1} for definitions and details.
\qed
\end{remark}

The following appears for example as Lemma~1.1 in \cite{kabil92}.
\begin{proposition}
  \label{prop:prime-idl}
  1.\  The up primes are exactly the filters. 

\noindent
2.\ The down primes are exactly the ideals.
\end{proposition}
\begin{proof}
1.\ is clear and we focus on 2.

\noindent
$(\GaD)$: We only have to check that a down prime $P$ is directed.
Assume it is not. Then it contains two elements $x_1,x_2$ such that
$\upc x_1\cap \upc x_2\cap P=\emptyset$. In other words, $P\subseteq
(P\setminus\upc x_1)\cup(P\setminus \upc x_2)$. But $P\setminus\upc
x_i$ is downwards-closed for both $i=1,2$, so $P$ being prime is
included in one $P\setminus \upc x_i$. This contradicts $x_i\in P$.

\noindent
$(\DaG)$: Consider an ideal $I\subseteq X$.
Aiming for a contradiction, we assume that $I \subseteq D_1 \cup D_2$
but $I \not\subseteq D_1$, $I \not\subseteq D_2$.  Pick a point $x_1$
from $I \setminus D_1$, and a point $x_2$ from $I \setminus D_2$.
Since $I$ is directed, there is a point $x \in I$ such that $x_1, x_2
\leq x$.  Since $D_1$ is downwards-closed, $x$ is not in $D_1$, and
similarly $x$ is not in $D_2$, so $x$ is not in $D_1 \cup D_2$, contradiction.
\qed
\end{proof}

\Cref{prop:prime-idl} explains why ideals appeared in our
representation of downwards-closed sets of $\Nat^2$ in
\cref{sec-dwc-nat}.  There is a general reason: ideals are the down
primes used in canonical decompositions, just like filters do for
upwards-closed sets.  Primality explains why the representation is
canonical, and why comparing downwards-closed sets reduces to comparing
generators. Meanwhile, the view of ideals as sets of the form $\dwc x$
where $x$ is either a normal point in $X$ or, possibly, a limit point
in $\widehat X$ ---recall \cref{rem-idl-iso-sobr}--- explains why
comparing ideals is often very similar to comparing points ---recall
testing whether $\dwc\tupsum{3,4}\subseteq\dwc\tupsum{\omega,1}$ or
whether $\tta^*\ttb^*\subseteq \ttb^*\ttc^*\tta^*$.

\section{Ideally effective \wqos}
\label{sec-effective}

When describing generic algorithms for \wqos, one needs to make
some basic computational assumptions on the \wqos at hand.  Such
assumptions are often summarized informally along the line of
``\emph{the \wqo $(X,{\le})$ is effective}'' and their precise meaning
is often defined at a later stage, when one gives sufficient
conditions based on the algorithm one is describing, a classic example
being~\cite{finkel98b}.
Sometimes the effectiveness assumptions are not even spelled out formally,
e.g., when one has in mind applications where the \wqo is
$(\nat^k,{\le_{\times}})$ or $(A^*,{\subword})$ which are obviously
``effective'' under all expected understandings.

The situation is different in this chapter since our goal is to
provide a formal notion of effectiveness that is \emph{preserved} by
the main \wqo constructions (and that supports the computation on
closed subsets illustrated in \cref{sec-motivations}). As a
consequence, we cannot avoid giving a formal definition, even if this
mostly amounts to administrative technicalities.

To simplify this task, we start by fixing the representation for
closed subsets: these will be represented as finite unions of prime
subsets as explained in \cref{sec-basics}. This provides a robust,
generic, and convenient data structure for $\Up(X)$ and $\Down(X)$
based on data structures (to be defined) for $\Fil(X)$ and
$\Idl(X)$. We do not require the decomposition to be canonical and
leave this as an implementation choice (the underlying complexity
trade-offs depend on the \wqo and the application at hand).  Moreover,
and since all filters are principal in \wqos, any data structure for
$X$ can be reused for representing $\Fil(X)$, so we will only need to
assume that $X$ and $\Idl(X)$ have an effective presentation.

This leads to the following definition.  Note that, rather than being
completely formal and talk of recursive languages or G\"odel
numberings, we will allow considering more versatile data structures
like terms, tuples, graphs, etc., that are closer to actual
implementations. All data structures considered in this paper will be
recursive sets, and in particular one can enumerate their elements.

{\leqnomode
\begin{definition}[Ideally Effective \wqos]
  \label{def:eff-wqo}
A \wqo $(X,{\leq})$ further equipped with data structures for representing $X$ and
$\Idl(X)$ is \emph{ideally effective} if:
\begin{align}
\tag{\textrm{OD}}\label{OD}
&\begin{minipage}{-10em+\textwidth}
    the ordering $\leq$ is decidable on (the representation of) $X$;
\end{minipage} \\
\tag{\textrm{ID}}\label{ID}
&\begin{minipage}{-10em+\textwidth}
   similarly, $\subseteq$ is decidable on $\Idl(X)$;
   \end{minipage} \\
\tag{\textrm{PI}}\label{PI}
&\begin{minipage}{-10em+\textwidth}
   principal ideals are computable, that is, \\
   $x\mapsto\dc x$ is
   computable;
   \end{minipage} \\
\tag{\textrm{CF}}\label{CF}
&\begin{minipage}{-10em+\textwidth}
   complementation of filters, denoted \\
   $\neg:\Fil(X)\to\Down(X)$, is
computable;
\end{minipage} \\
\tag{\textrm{IF}}\label{IF}
&\begin{minipage}{-10em+\textwidth}
   intersection of filters, denoted \\
   $\cap:\Fil(X)\times\Fil(X)\to\Up(X)$, is
   computable;
   \end{minipage} \\
\tag{\textrm{CI}}\label{CI}
&\begin{minipage}{-10em+\textwidth}
   complementation of ideals, denoted \\
   $\neg:\Idl(X)\to\Up(X)$, is
   computable;
   \end{minipage} \\
\tag{\textrm{II}}\label{II}
&\begin{minipage}{-10em+\textwidth}
   intersection of ideals, denoted \\
   $\cap:\Idl(X)\times\Idl(X)\to\Down(X)$, is
   computable.
   \end{minipage}
\end{align}
\end{definition}
}
Some immediate remarks are in order:
\begin{itemize}
\item As mentioned earlier, elements of $\Up(X)$ and $Down(X)$ are
  represented
as collections (via lists, or sets, or~\ldots) of elements of $X$ and of $\Idl(X)$
respectively. The computability of unions is thus trivial and
therefore was not required in the formal definition.
\item
Similarly, checking membership $x\in D$ for  downwards-closed sets
reduces to  deciding $\dc x\subseteq D$, hence was not required
either.
\item We said earlier that operations on $\Up$ and $\Down$ boil down
  to operations on filters and ideals. Note that there are some
  subtleties. For example, deciding inclusions over $\Up(X)$ or
  $\Down(X)$ is made possible because the decompositions only use
  prime subsets.  Explicitly, in order to check whether
  $D \subseteq D'$ for example, where $D = \bigcup_{i<m} I_i$ and
  $D' = \bigcup_{j<n} I'_j$, we check whether every $I_i$ is included
  in some $I'_j$---this is correct because every ideal $I_i$ is down
  prime.
\item
There is some asymmetry in the definition between upwards-closed and
downwards-closed sets. This should be expected since \wqos are
well-founded but the reverse orderings need not be.
\item
The astute reader may have noticed that the definition contains some
hidden redundancies. Our proposal is justified by algorithmic
efficiency concerns, see discussion in \cref{sec-axioms}.
\end{itemize}

\subsection{Some first ideally effective \wqos}

We quickly show that the simplest \wqos are ideally effective. They
will be used later as building blocks for more complex \wqos.

\subsubsection{Finite orderings.}
\label{sec-finite}

A frequently occurring quasi-ordering in computer science is the
\emph{finite alphabet with $n$ symbols}. It consists of a set with $n$
elements, usually denoted $A$, ordered by equality. This is a \wqo
since $A$ is finite. The name ``alphabet'' comes from its applications
in language theory but this very basic \wqo appears in many other
situations, e.g., as colorings of some other objects, as the set of
control states in formal models of computations such as Turing
machines, communicating automata, etc.

Let us spell out, as a warming-up exercise, why
this \wqo $(A, =)$ is ideally effective. One can for instance
represent elements of $A$ using natural numbers up to $|A|-1$. The
ordering is trivially decidable. All ideals of $(A,=)$  are principal,
that is of the form $\dwc x=\{x\}$ for $x \in A$. We thus represent ideals
as elements, exactly as we do for filters. Therefore, ideal inclusion coincides
with equality, and \PI is given by the identity function. All
other operations are trivial: intersection of filters (\emph{resp.} ideals) is
always empty except if the  two filters (\emph{resp.} ideals) are
equal, and $\neg \upc x = \neg \dwc x = A \setminus \{ x \}$.
\\

We could of course have dispensed with these explanations since, more
generally, any finite \qo is a \wqo and is ideally effective. In
particular, all operations required by \cref{def:eff-wqo} are always
computable, being operations on a finite set. Let us note that all
ideals are principal in this setting, which is no surprise since
$(X,{\ge})$ is also a \wqo, and its filters are the ideals of
$(X,{\le})$.

\subsubsection{Natural numbers.}
\label{sec-naturals}

Apart from finite orders, the simplest \wqo is $(\nat,{\le})$.  We now
restate our observations from \cref{sec-motivations} in the more
formal framework of \cref{def:eff-wqo}.

Observe that since $\le$ is linear, any downwards-closed set is
actually an ideal, except for $\emptyset$. The ideals that are bounded
from above have the form $\dwc n$ for some $n\in\Nat$, and the only
unbounded ideal is the whole set $\nat$ itself, often denoted
$\dwc \omega$ as we did in \cref{sec-dwc-nat}. Ideal inclusion is thus
decidable: principal ideals are compared as elements, and
$\dwc \omega$ is larger than all the others.  Thus
$(Idl(\nat),{\subseteq})$ is linearly ordered, which makes
intersections trivial: one has $\upc n\cap \upc m=\upc \max(n,m)$ and
$\dwc n\cap \dwc m=\dwc \min(n,m)$. Finally, complements are computed
as follows:
\begin{xalignat*}{2}
  \neg \upc (n+1) &= \dwc n \:,      & 	\neg \dwc n &= \upc (n+1) \:,
\\
  \neg \upc 0 &= \emptyset \:,       &      \neg \dwc \omega &= \emptyset \:.
\end{xalignat*}

\subsubsection{Ordinals.}
\label{sec-ordinals}

The above analysis extends to any recursive linear \wqo, i.e., any
recursive ordinal (see~\cite{sacks90} for definitions).
Given an ordinal $\alpha$, we write $\bm{\alpha}$ (in bold font) for the set of
ordinals $\{ \beta \mid \beta < \alpha \}$---the classical
set-theoretic construction of ordinals equates $\alpha$ with
$\bm{\alpha}$.

Let $(X, {\le}) = (\bm{\alpha}, {\le})$. Once again, $X$ being linearly
ordered, its ideals are its downwards-closed sets (except
$\emptyset$). Therefore, there are three types of ideals:
\begin{enumerate}
\item $I = X$,
\item $I$ has a maximal element $\beta \in X$,
  in which case $I = \dwc \beta$,
\item Or $I$ has a supremum $\beta \in X \setminus I$,
  in which case $I = \dwc_< \beta = \bm{\beta}$.
\end{enumerate}
Note that in the second case, $I = \dwc \beta = \dwc_< (\beta+1) =
\bm{\beta+1}$. Thus every ideal of $(X, {\le})$ is a $\bm{\beta}$ for
some $\beta \in \bm{\alpha+1} \setminus {0}$, and ideal inclusion
coincides with the natural ordering on $\bm{\alpha+1}$.

Now, assuming that we can represent elements of $X$ in a way that
makes $\le$ decidable,  then $(X,{\le})$ is ideally effective. Indeed,
the representation is
easily extended to $(\bm{\alpha+1}, {\le})$ and one can thus decide ideal
inclusion. Intersections are computable as the maximum for filters,
as the minimum for ideals. Finally, complements of filters and ideals are computed as follows:
\begin{xalignat*}{2}
        \begin{array}{rl}
          \neg \upc 0 &= \emptyset \:,
        \\
         \neg \bm{\alpha} & = \emptyset \:,
        \end{array}
&&
\left.
        \begin{array}{rl}  
                \neg \upc \beta &= \bm{\beta}
        \\
                \neg \bm{\beta} &= \upc \beta 
        \end{array}\right\} \text{for } \beta \in \bm{\alpha} \:.
\end{xalignat*}

While the above applies to any recursive ordinal, the applications
that we are aware of usually only need ordinals below $\epsilon_0$,
for which the \emph{Cantor Normal Form} is well known and understood,
and leads to natural data structures~\cite{manolios05}.  One can push
this at least to all ordinals below the larger ordinal $\Gamma_0$
\cite{Gallier:Gamma0}.

Note that, when $\alpha = \omega$, the representation of ideals
differs from the representation for $Idl(\nat)$ proposed in
\cref{sec-naturals}: in one case we use $\dwc_< n$ while in the other
we use $\dwc n$. Both options are equivalent, leading to very similar
algorithms.  In \cref{sec-naturals} we adopted the representation
that has long been common in Petri nets tools.

\section{Constructing ideally effective \wqos}
\label{sec-basic-eff}

We now look at more complex \wqos.  In practice these are obtained by
combining simpler \wqos via well-known operations like Cartesian
product, sequences extension, etc.  Our strategy is thus to show that
these operations produce ideally effective \wqos when they are applied
to ideally effective \wqos.

\subsection{Ideally effective \wqo constructors}
\label{sec-idl-eff-constructors}

We shall provide generic (i.e.,
uniform) algorithms that manage filters and ideals of compound \wqos
by invoking the algorithms for the filters and ideals of their
components. This is made precise in \cref{def:idl-eff-constr}, and to
this end, we have to introduce the following notion:

{\leqnomode
\begin{definition}
\label{def-presentation}
A \emph{presentation} of an ideally effective \wqo $(X,{\le})$, is a
list of:
\begin{xalignat}{1}
\notag
\bm{-} & \text{ data structures for $X$ and $\Idl(X)$},
\\
\notag
\bm{-} & \text{ algorithms for the seven computable functions required by \cref{def:eff-wqo}},
\\
\tag{\textrm{XI}}\label{XI}
\bm{-} & \text{ the ideal decomposition $X=\bigcup_{i<n}I_i$ of $X$ as a
downwards-closed set},
\\
\tag{\textrm{XF}}\label{XF}
\bm{-} & \text{ as well as its filter decomposition  $X=\bigcup_{i<n'}F_i$}.
\end{xalignat}
\end{definition}
}
Obviously, a \wqo is ideally effective if and only if it has a
presentation as defined in \cref{def-presentation}.

The notion of presentations as actual objects is needed because
they are the actual inputs of our \wqo constructions.
This explains why we added \XI and \XF in the requirements. For a
given $(X,{\leq})$, the ideal and filter decompositions of $X$ always
exist and requiring them in \cref{def:eff-wqo} would make no sense.
However, these decompositions are needed by algorithms that
work uniformly on \wqos given via their presentations.
\\

Let us informally call \emph{order-theoretic constructor} (constructor
for short) any operation $C$ that produces a quasi-ordering
$C[(X_1, {\le_1}), \dots, (X_n,{\le_n})]$ from given quasi-orderings
$(X_1, {\le_1}), \dots, (X_n,{\le_n})$. In subsequent sections, $C$
will be instantiated with very well-known constructions, such as
Cartesian product with componentwise ordering, finite sequences with
Higman's ordering, finite sets with the Hoare quasi-ordering, and so
on.  In practice, we will always have $n=1$ or $2$. We also say that
an order-theoretic constructor preserves \wqo if
$C[(X_1, {\le_1}), \dots, (X_n,{\le_n})]$ is a \wqo whenever
$(X_1, {\le_1}), \dots, (X_n,{\le_n})$ are. The constructors we just
mentioned are well-known to be \wqo-preserving. We extend this concept
to ideally effective \wqos:
\begin{definition}
  \label{def:idl-eff-constr}
  An order-theoretic \wqo-preserving constructor $C$ is said to be
  \emph{ideally effective} if:
  \begin{itemize}
  \item It preserves ideal effectiveness, that is
    $C[(X_1, {\le_1}), \dots, (X_n,{\le_n})]$ is ideally effective
    when each $(X_i, {\le_i})$ is.
  \item A presentation of $C[(X_1, {\le_1}), \dots, (X_n,{\le_n})]$ is
    uniformly computable from presentations of the \wqos
    $(X_i,{\le_i})$ ($i=1, \ldots, n$).
\end{itemize}
\end{definition}

In the following sections, we proceed to prove that some of the most
prominent \wqo-preserving constructors are also ideally effective.

\subsection{Sums of \wqos}
\label{sec-sums}

We start with two simple constructions, disjoint sums and
lexicographic sums of \wqos.  They will be our first examples of
ideally effective constructors and will set the template for later
constructions.

\subsubsection{Disjoint sum.}
The \emph{disjoint sum} $X_1 \sqcup X_2$ of two \qos $(X_1, {\le_1})$
and $(X_2, {\le_2})$ is the set
$\{1\} \times X_1 \cup \{2\} \times X_2$, quasi-ordered by:
\begin{align*}
  \tupsum{i,x} \le_{\sqcup} \tupsum{j,y} \text{ iff } i=j \text{ and } x \le_i y
\:.
\end{align*}
We use $X_{\sqcup}$ to denote $X_1\sqcup X_2$ and generally use the
$\sqcup$ subscript to identify operations associated with the
structure $(X_{\sqcup},{\le_{\sqcup}})$. This structure is obviously well
quasi-ordered when $(X_1, {\le_1})$ and $(X_2, {\le_2})$ are.

We let the reader check the following characterization.
\begin{proposition}[Ideals of $\pmb{X_1\sqcup X_2}$]
\label{prop-ideals-sqcup}
Given $(X_1, {\le_1})$ and $(X_2,{\le_2})$ two \wqos,
the ideals of $(X_1 \sqcup X_2,{\le_{\sqcup}})$
are exactly the sets of the form
$I=\{i\}\times J$ with $i\in \{1,2\}$ and  $J$ an ideal of $X_i$.
\end{proposition}
Thus $(\Idl(X_1\sqcup X_2),{\subseteq})$ is isomorphic to
$(\Idl(X_1),{\subseteq})\sqcup (Idl(X_2),{\subseteq})$.

Given data structures for $X_1$ and $X_2$, we use the natural
data structure for $X_1 \sqcup X_2$. Moreover,
\cref{prop-ideals-sqcup} shows that ideals of the \wqo
$(X_1 \sqcup X_2, {\le_{\sqcup}})$ can similarly be represented using
data structure for $\Idl(X_1)$ and $\Idl(X_2)$.
\begin{theorem}
  With the above representations of elements and ideals, disjoint
  union is an ideally effective constructor.
\end{theorem}
\begin{proof}[Sketch]
  Let $(X_1, {\le_1})$ and $(X_2, {\le_2})$ be two ideally effective
  \wqos.

  In the following, we write $\ibar$ for $3-i$ when $i\in\{1,2\}$, so
  that $\{i,\ibar\}=\{1,2\}$.  We also abuse notation and, for a
  downwards-closed subset $D=\bigcup_a I_a$ of $X_i$, we write
  $\tupsum{i,D}$ to denote $\bigcup_a \tupsum{i,I_a}$, a
  downwards-closed subset of $X_{\sqcup}$ represented via
  ideals. Similarly, for an upwards-closed subset
  $U=\bigcup_a \upc_{i} x_a$ of $X_i$, we let $\tupsum{i,U}$ denote
  $\bigcup_a \upc_{\sqcup}\tupsum{i,x_a}$.

  \begin{description}
  \item[\OD:] the definition of $\le_{\sqcup}$ is already an implementation.
  \item[\ID:] we use
    $\tupsum{i,J}\subseteq\tupsum{i',J'} \iff i=i'\land J\subseteq
    J'$.
  \item[\PI:] we use $\dwc_{\sqcup}\tupsum{i,x}=\tupsum{i,\dwc_{i}x}$
    for $i\in\{1,2\}$.
  \item[\CF:] we use
    $X_{\sqcup}\setminus\upc_{\sqcup}\tupsum{i,x}
    =\tupsum{i,X_i\setminus\upc_{i} x} \cup \tupsum{\ibar,X_{\ibar}}$.
    Note that this relies on \CF for $X_i$ (to express
    $X_i\setminus \upc_{i}x$ as a union of ideals) and on \XI for
    $X_{\ibar}$.
  \item[\II:] we rely on \II for $X_1$ and $X_2$, using
    \[
      \tupsum{i,I}\cap \tupsum{j,J} =
      \begin{cases}
        \tupsum{i,I\cap J} & \text{if $i=j$,}
        \\
        \emptyset &\text{otherwise.}
      \end{cases}
    \]
  \end{description}
  Operations \CI to complement ideals and \IF to intersect filters are
  analogous.

  Observe that the presentation of $(X_1 \sqcup X_2, {\le_{\sqcup}})$
  described above is clearly computable from presentations for
  $(X_i,{\le_i})$ ($i=1,2$). Notably, a filter (resp.\ ideal)
  decomposition of $X_1 \sqcup X_2$ is easily obtained by taking the
  union of filter (resp.\ ideal) decompositions of $X_1$ and $X_2$,
  thus establishing \XF (resp.\ \XI).
  \qed
\end{proof}

\subsubsection{Lexicographic sums.}
The \emph{lexicographic sum}  $X_1 \oplus X_2$ \footnote{A warning
about notation: the lexicographic sum should not be confused with the
natural sum of ordinals even if they are both denoted with
$\oplus$. In particular, the lexicographic sum of ordinals is their
usual addition.} of two \qos $(X_1, {\le_1})$, $(X_2, {\le_2})$ is the
\qo $(X_{\oplus}, {\le_{\oplus}})$ given by $X_{\oplus}= \{1\} \times X_1
\cup \{2\} \times X_2$ and
\begin{align*}
  \tupsum{i,x} \le_{\oplus} \tupsum{j,y} \text{ iff } i < j \text{ or } (i=j \text{ and } x \le_i y)
\:.
\end{align*}
Therefore $X_1\oplus X_2$ and $X_1\sqcup X_2$ share the same
underlying set. The ordering $\le_{\oplus}$ is an extension of
$\leq_{\sqcup}$ hence is a \wqo too, when $(X_1, {\le_1})$ and
$(X_2,{\le_2})$ are.

Again, the following characterization is easy to obtain.
\begin{proposition}[Ideals of $\pmb{X_1\oplus X_2}$]
  Given two \wqos $(X_1, {\le_1})$ and $(X_2,{\le_2})$, the ideals of
  $X_1\oplus X_2$ are exactly the sets of the form $\{1\}\times J_1$
  with $J_1\in\Idl(X_1)$, or of the form
  $\{1\}\times X_1\cup\{2\}\times J_2$ with $J_2\in\Idl(X_2)$.
\end{proposition}
Thus $(\Idl(X_1 \oplus X_2), {\subseteq})$ is isomorphic to
$(\Idl(X_1),{\subseteq}) \oplus (\Idl(X_2),{\subseteq})$, which leads to a
simple data structure for the set of ideals\footnote{
   Note that with this representation, a pair $\tupsum{i,J}$ where
   $J\in\Idl(X_i)$ denotes $\{1\}\times J$ when $i=1$, and
   $\{1\}\times X_1\cup \{2\}\times J$ ---and not $\{2\}\times J$--- when $i=2$.
}
when $X_1$ and $X_2$ are effective.

\begin{theorem}
  With the above representations, lexicographic union is an ideally
  effective constructor.
\end{theorem}

\begin{proof}[Sketch]
  We reuse the abbreviations $\tupsum{i,U}$, $\tupsum{i,D}$, $\ibar$,~\ldots, 
  introduced for disjoint sums.  Also, we only consider the
  case where both $X_1$ and $X_2$ are non-empty (the claim is trivial
  otherwise).
  
  \begin{description}
  \item[\OD:] follows from the definition.
  \item[\ID:] ideal inclusion can be tested as the lexicographic sum
    of $\Idl(X_1)$ and $\Idl(X_2)$.
  \item[\PI:] $\dwc_{\oplus}\tupsum{i,x}$ is (represented by)
    $\tupsum{i,\dwc_{i} x}$.
  \item[\CF:] the complement
    $X_{\oplus}\setminus\upc_{\oplus}\tupsum{i,x}$ is (represented by)
    $\tupsum{i,X_i\setminus\upc_i x}$ except when $i=2$ and
    $\upc_i x=X_2$, in which case
    $X_{\oplus}\setminus\upc_{\oplus}\tupsum{2,x}$ is
    $\tupsum{1,X_1}$.
  \item[\II:] intersection of two ideals considers two cases.  First
    $\tupsum{1,I}\cap \tupsum{2,J}$ is (represented by) $\tupsum{1,I}$
    for ideals issued from different components in $X_{\oplus}$. For
    $\tupsum{i,I}\cap \tupsum{i,J}$, i.e., ideals issued from the same
    component, we use $\tupsum{i,I\cap J}$ except when $i=2$ and
    $I\cap J=\emptyset$, in which case $\tupsum{2,I}\cap \tupsum{2,J}$
    is $\tupsum{1,X_1}$.
  \end{description}
  Procedures for the dual operations \CI and \IF are similar.
  Moreover, the presentation above is obviously computable from
  presentations for $(X_1, {\le_1})$ and $(X_2, {\le_2})$. Regarding
  \XI and \XF, the ideal decomposition of $X_1 \oplus X_2$ is the
  ideal decomposition of $X_2$ and the filter decomposition of
  $X_1 \oplus X_2$ is the filter decomposition of $X_1$.
  \qed
\end{proof}

 \subsection{Products of \wqos and Dickson's Lemma}
\label{sec-products}

Given two \qos $(X_1,{\leq_1})$ and $(X_2,{\leq_2})$,
we define the \emph{componentwise quasi-ordering} $\le_{\times}$ on
the Cartesian product $X_1\times X_2$ by
$\tup{x_1,x_2}\leq_{\times}\tup{y_1,y_2} \equivdef x_1\leq_1 y_1\land
x_2\leq_2 y_2$.

\begin{lemma}[Dickson's Lemma]
  If $X_1$ and $X_2$ are \wqos, so is $X_1 \times X_2$.
\end{lemma}
\begin{proof}[Idea]
  Given an infinite sequence in $X_1 \times X_2$, we extract an
  infinite sequence which is monotonic in the first component, and
  from that, an infinite sequence that is monotonic in the second
  component.
\qed
\end{proof}

The  ideals of $(X_1 \times X_2,{\le_{\times}})$ are well known.
\begin{proposition}[Ideals of $\pmb{X_1\times X_2}$]
  \label{lem-ideals-of-product}
  Let $(X_1, {\le_1})$ and $(X_2,{\le_2})$ be two \wqos.
A subset $I$ is an ideal of $X_1\times X_2$ if, and only if,
$I=I_1\times I_2$ for some ideals $I_1,I_2$ of $X_1$ and $X_2$ respectively.
\end{proposition}
\begin{proof}
$(\DaG)$: One checks that $I=I_1\times I_2$ is non-empty,
  downwards-closed, and directed, when $I_1$ and $I_2$ are. For
  directedness, we consider two elements
  $\tup{x_1,x_2},\tup{y_1,y_2}\in I$. Since $I_1$ is directed and
  contains $x_1,y_1$, it contains some $z_1$ with $x_1\leq_1 z_1$ and
  $y_1\leq_1 z_1$. Similarly $I_2$ contains some $z_2$ above $x_2$ and
  $y_2$ (wrt.\ $\leq_2$). Finally, $\tup{z_1,z_2}$ is in $I$, and
  above both $\tup{x_1,y_1}$ and $\tup{x_2,y_2}$.

  \noindent
  $(\GaD)$: Consider $I\in\Idl(X_1\times X_2)$ and write $I_1$
  and $I_2$ for its projections on $X_1$ and $X_2$. These projections
  are downwards-closed (since $I$ is), non-empty (since $I$ is) and
  directed (since $I$ is), hence they are ideals (in $X_1$ and $X_2$).
  We now show that $I_1\times I_2\subseteq I$. Consider an arbitrary
  $x_1\in I_1$: since $I_1$ is the projection of $I$, there is some
  $y_2\in X_2$ such that $\tup{x_1,y_2}\in I$. Similarly, for any
  $x_2\in I_2$, there is some $y_1\in X_1$ such that
  $\tup{y_1,x_2}\in I$. Since $I$ is directed, there is some
  $\tup{z_1,z_2}\in I$ with $\tup{x_1,y_2}\leq_{\times} \tup{z_1,z_2}$
  and $\tup{y_1,x_2}\leq_{\times} \tup{z_1,z_2}$. But then
  $x_1\leq_1 z_1$ and $x_2\leq_2 z_2$. Thus $\tup{x_1,x_2}\in I$ since
  $I$ contains $\tup{z_1,z_2}$ and is downwards-closed. Hence
  $I=I_1\times I_2$ and $I$ is a product of ideals.
\qed
\end{proof}

Thus $\Idl(X_1\times X_2,{\subseteq})$ is isomorphic to
$(\Idl(X_1),{\subseteq})\times (\Idl(X_2),{\subseteq})$.  If
$(X_1,{\le_1})$ and $(X_2,{\le_2})$ are ideally effective, we
naturally represent elements of $X_1 \times X_2$ as pairs of elements
of $X_1$ and $X_2$, and similarly ideals of
$(X_1 \times X_2, {\le_{\times}})$ as pairs of ideals of $X_1$ and
$X_2$. This is notably how we handled $\Idl(\Nat^2)$ in
\cref{sec-dwc-nat}.

\begin{theorem}
With the above representations, Cartesian product is an ideally
effective constructor.
\end{theorem}
\begin{proof}
  Let $D_1$ and $D_2$ be downwards-closed sets of $(X_1,{\le_1})$ and
  $(X_2,{\le_2})$ respectively, given by some ideal decompositions
  $D_1=\bigcup_i I_{1,i}$ and $D_2=\bigcup_j I_{2,j}$.  Then
  $D_1 \times D_2$ is downwards-closed in $X_1 \times X_2$, and it
  decomposes as $\bigcup_i\bigcup_j I_{1,i}\times I_{2,j}$ since products
  distribute over unions.  The same reasoning holds for upwards-closed
  sets and their filter decompositions and we rely on these properties
  in the following explanations.
  \begin{description}
  \item[\OD:] the ordering $\le_{\times}$ is obviously decidable.
  \item[\ID:] $I_1 \times I_2 \subseteq J_1 \times J_2$ iff
    $I_1 \subseteq J_1$ and $I_2 \subseteq J_2$ (exercise: the
    nonemptiness of ideals is required here).
  \item[\PI:] $\dwc \tup{x_1, x_2} = \dwc x_1 \times \dwc x_2$.
  \item[\II:] to compute intersections, use
    $(I_1 \times I_2) \cap (I'_1 \times I'_2) = (I_1 \cap I'_1) \times
    (I_2 \cap I'_2)$, and build the product of downwards-closed sets
    as explained above.
  \item[\CF:] to complement filters, use
    $(X_1\times X_2) \setminus \upc_{\times} \tup{x_1, x_2} =
    \bigl[(X_1 \setminus \upc x_1 ) \times X_2 \bigr] \cup \bigl[ X_1
    \times (X_2 \setminus \upc x_2) \bigr]$ and build products of
    downwards-closed sets.
  \end{description}
  Procedures for the remaining operations are obtained similarly. Note
  that here too, the presentation above is computable from
  presentations for $(X_1,{\le_1})$ and $(X_2,{\le_2})$.  Notably, a
  filter and ideal decomposition of $X_1 \times X_2$ is easily
  obtained from decompositions of $X_1$ and $X_2$, by distributing
  products over unions.
  \qed
\end{proof}

 \subsection{Sequence extensions of \wqos and Higman's Lemma}
\label{sec-sequences}

Given a \qo $(X,{\leq})$, we denote by $X^*$ the \emph{sequence
  extension} of $X$, i.e., the set of all finite sequences over $X$,
often called \emph{words} when $X$ is an alphabet.  We write
$\epsilon$ for the empty (zero-length) sequence, and denote
multiplicatively the concatenation of sequences, as $\bu \bv$ or
$\bu \cdot \bv$.  Elements of $X^*$ will be denoted in bold font, such
as $\bu, \bv,...$, while elements of $X$ are denoted $x,y,...$.  In
particular, if $x \in X$, then $\bx \in X^*$ denotes the sequence of
length one containing only the symbol $x$.

The set $X^*$ is often quasi-ordered with \emph{Higman's
  quasi-ordering} $\le_*$, also known as the \emph{sequence embedding}
quasi-ordering, defined by
$\bm{u}=x_1\cdots x_n\leq_* \bm{v}=y_1\cdots y_m$ $\equivdef$ there
are $n$ indices $1\leq p_1<p_2<\cdots <p_n\leq m$ such that
$x_i\leq y_{p_i}$ for each $i=1,\ldots,n$.  In other words, and
writing $[n]$ for the set $\{1, \cdots, n\}$, there is a strictly
increasing mapping $p$ from $[n]$ to $[m]$ such that
$x_i \le y_{p(i)}$.  Such a mapping will be called a \emph{witness} of
$\bm{u} \le_* \bm{v}$.  Equivalently, $\bm{u}\leq_*\bm{v}$ if $\bm{v}$
contains a length $n$ subsequence $\bm{v'}=y_{p_1}\cdots y_{p_n}$ such
that $\bm{u}\leq_{\times} \bm{v'}$ using the product ordering from
\cref{sec-products}.

The structure $(X^*,{\le_*})$ is sometimes called the \emph{Higman
  extension} of $(X,{\le})$. This constructor preserves \wqo: this is
Higman's Lemma \cite{Higman:Lemma}.

Showing that this constructor is ideally effective requires some work
and \cref{sec-sequences} is one of the longest in this chapter. This
is justified by the importance of this construction. Being generic,
our algorithms apply to non-trivial instances such as $(\nat^k)^*$
---used in Timed-arc nets~\cite{HSS-lics2012}, in data
nets~\cite{lazic2008}, for runs of Vector Addition
Systems~\cite{leroux2016}---, or $(\Sigma^*)^k\times(\Sigma^*)^*$
---used in Dynamic Lossy Channel Systems~\cite{AAC-fsttcs2012}---, or
to even richer settings like the Priority Channel Systems and the
Higher-Order Channel Systems of~\cite{HSS-lmcs}. The algorithms for
$(X^*,{\leq_*})$ are also invoked when showing ideal effectiveness of
many \wqos derived from $X^*$.  \\

Before we study the ideals of $X^*$, let us first lift the
concatenation of sequences to sets of sequences: the product (for
concatenation) of two sets of sequences $\bU, \bV \subseteq X^*$ is
denoted
${\bU \cdot \bV \egdef \{ \bu \cdot \bv \mid \bu \in \bU, \bv \in \bV
  \}}$.  A useful property of $(X^*,{\leq_*})$ is that the
concatenation of downwards-closed sets distributes over intersection:
\begin{lemma}
  \label{lem-concat-distrib-inter}
  Let $\bD_1,\bD_2,\bD\in\Down(X^*)$. Then
  $(\bD_1\cap\bD_2)\cdot\bD =(\bD_1\cdot\bD)\cap(\bD_2\cdot\bD)$.
\end{lemma}
\begin{proof}
  The left-to-right inclusion is obvious. For the right-to-left
  inclusion, let $\bw \in \bD_1 \cdot \bD \cap \bD_2 \cdot \bD$. Then
  $\bw = \bu_1 \bv_1$ for some $\bu_1 \in \bD_1$ and $\bv_1 \in
  \bD$. Also, $\bw = \bu_2\bv_2$ for some $\bu_2 \in \bD_2$ and
  $\bv_2 \in \bD$. One of $\bu_1$ and $\bu_2$ is a prefix of the
  other. Assume $\bu_1$ is a prefix of $\bu_2$ (the other case is
  analogous). Since $\bD_2$ is downwards-closed and
  $\bu_1 \leq_* \bu_2$, $\bu_1 \in \bD_2$. Thus,
  $\bu_1 \in \bD_1 \cap \bD_2$ and
  $\bw = \bu_1 \bv_1 \in (\bD_1 \cap \bD_2) \cdot \bD$.  \qed
\end{proof}

The structure of ideals of $(X^*,{\le_*})$ is given in~\cite{kabil92}
where the following theorem is proved. An alternate proof is presented
in \cref{subsec-proof-of-idl-higman}.
\begin{theorem}[Ideals of $\pmb{X^*}$]
  \label{thm-idl-higman}
  Given a \wqo $(X,{\le})$, the ideals of $(X^*,{\leq_*})$ are exactly
  the finite products of {atoms}, of the form
  $\bP=\bA_1\cdot \bA_2\cdots \bA_n$ where \emph{atoms} are:
  \begin{itemize}
  \item any set of the form $\bA=D^*$, for $D\in\Down(X)$,
  \item any set of the form
    $\bA=I+\epsilon \egdef \{ \bx \mid x \in I \} \cup \{ \epsilon \}$,
    for $I\in\Idl(X)$.
  \end{itemize}
\end{theorem}

\subsubsection{Ideal effectiveness.}

The elements of $X^*$ will be represented in the natural way, e.g.,
via lists of elements of $X$ (assuming a data structure for $X$).
When $(X,{\leq})$ is ideally effective, \cref{thm-idl-higman} leads to
a natural data structure for ideals of $X^*$, as lists of atoms, where
the representation of atoms is directly inherited from those for
$\Idl(X)$ and $\Down(X)$.

\begin{theorem}
\label{thm-hig-idl-eff}
With the above representations, the sequence extension is an ideally
effective constructor.
\end{theorem}
\begin{proof}
  Let $(X,{\le})$ be an ideally effective \wqo.
  \begin{description}
  \item[\OD:] deciding $\le_*$ over $X^*$ reduces to comparing
    elements of $X$, {e.g.} by looking for a \emph{leftmost
      embedding}.
  \item[\PI:] given a finite sequence $\bu = x_1 \cdots x_n$, the
    principal ideal $\dwc \bu$ is represented by the product
    $(\dwc x_1 + \epsilon) \cdots (\dwc x_n + \epsilon)$.
  \end{description}

  Procedures for the remaining operations required by
  \cref{def:eff-wqo} are more elaborate, and we therefore introduce a
  lemma for each one. This series of lemmas concludes the proof since
  the fact that a presentation of $(X^*,{\le_*})$ can be uniformly
  computed from a presentation of $(X,{\le})$ will be clear. As for
  \XF, the filter decomposition of $X^*=\upc \bm{\epsilon}$ is given
  by the empty sequence (and does not depend on $X$), while for \XI we
  note that $X^*$ is already an ideal made of a single atom.
  \qed
\end{proof}

Subsequently, $(X,{\le})$ denotes an ideally effective \wqo.  We begin
with ideal inclusion.  A similar procedure was already obtained by
Abdulla et al.\  in the case where $X$
is a finite alphabet with equality~\cite{abdulla-forward-lcs}.
\begin{lemma}[\IDsanspar]
  \label{prop:incl-seq}
  Inclusion between ideals of $(X^*,{\le_*})$ can be tested using a
  linear number of inclusion tests between downwards-closed sets of
  $X$, using a version of \emph{leftmost embedding} search. The
  following equations implicitly describe an algorithm deciding
  inclusion by induction on the length, or number of atoms, of ideals:
  \addtocounter{equation}{1}
  \begin{enumerate}
  \item Atoms are compared as follows:
    \begin{align}
      \tag{\theequation.1}
      \label{emi0-1}
      (I_1 + \epsilon) \subseteq (I_2+\epsilon) &\iff I_1 \subseteq I_2,
      \\
      \tag{\theequation.2}
      \label{emi0-2}
      (I + \epsilon) \subseteq D^* &\iff I \subseteq D, \\
      \tag{\theequation.3}
      \label{emi0-3}
      D_1^* \subseteq D_2^* &\iff D_1 \subseteq D_2, \\
      \tag{\theequation.4}
      \label{emi0-4}
      D^* \subseteq (I+\epsilon) &\iff D = \emptyset.
    \end{align}
  \item\label{emi1} For any ideal $\bP$: $\bepsilon \subseteq \bP$.
  \item\label{emi2} For any ideal $\bP$ and atom $\bA$:
    $\bA \cdot \bP \subseteq \bepsilon \iff \bA = \emptyset^* \land \bP
    \subseteq \bepsilon$.
  \item Finally, for all atoms $\bA$ and $\bB$, and ideals $\bP$ and
    $\bQ$:
    \begin{enumerate}
    \item\label{caca1} if $\bA \not\subseteq \bB$ then:
      \[ \bA \cdot \bP \subseteq \bB \cdot \bQ \iff \bA \cdot \bP
        \subseteq \bQ \:; \]
    \item\label{caca2} if $\bA \subseteq \bB$ as in \eqref{emi0-1}, i.e., $\bA = (I_1 + \epsilon)$,
      $\bB = (I_2+\epsilon)$ for some $I_1,I_2 \in \Idl(X)$, then:
      \[  \bA \cdot \bP \subseteq \bB \cdot \bQ \iff
        \bP \subseteq \bQ \:;
      \]
    \item\label{caca3} if $\bA \subseteq \bB$ as in any of
      \crefrange{emi0-2}{emi0-4}, then:
      \[  \bA \cdot \bP \subseteq \bB \cdot \bQ \iff
        \bP \subseteq \bB \cdot \bQ \:.
      \]
    \end{enumerate}
  \end{enumerate}
\end{lemma}

\begin{proof}
  The first three cases are trivial.  We concentrate on the fourth
  one.
\begin{enumerate}
\item[\ref{caca1}] Since $\bB$ contains $\epsilon$, 
  $\bA \cdot \bP \subseteq \bQ $ implies
  $\bA \cdot \bP \subseteq \bB \cdot \bQ$.  Conversely, let
  $\bu \in \bA$ and $\bv \in \bP$, so that
  $\bu \bv \in \bA \cdot \bP \subseteq \bB \cdot \bQ$.  Assuming
  $\bA \not\subseteq \bB$, there exists $\bw' \in \bA \setminus \bB$
  and by directedness, there exists a word $\bw \in \bA$ such that
  $\bw \ge_* \bw', \bu$.  In particular, $\bw$ is in
  $\bA \setminus \bB$ and $\bw \ge_* \bu$.

  If $\bA = I + \epsilon$ for some $I \in \Idl(X)$, then $\bw$ is of
  length at most one.  Since $\bw' \not\in \bB$, in particular
  $\bw' \neq \epsilon$, so $\bw$ is of length exactly one.  Also,
  since $\bw \not\in \bB$, the word $\bw \bv$, which is in
  $\bA \cdot \bP \subseteq \bB \cdot \bQ$ has to actually be in
  $\bQ$. Since $\bQ$ is downwards-closed, $\bu \bv \in \bQ$.

  Otherwise, $\bA = D^*$ for some $D \in \Down(X)$. In this case,
  $\bw \bw \in \bA$ and thus $\bw \bw \bv \in \bB \cdot \bQ$.  We
  factor $\bw \bw \bv$ as $\bv_1 \bv_2$ with $\bv_1 \in \bB$ and
  $\bv_2 \in \bQ$.  Since $\bw$ is not in $\bB$, no word of which
  $\bw$ is a prefix is in $\bB$ either, and that implies that $\bv_1$
  is a proper prefix of $\bw$, and that $\bv_2$ has $\bw \bv$ as a
  suffix.  In particular, $\bv_2 \ge_* \bw\bv$.  Recalling that
  $\bw \bv \ge_* \bu \bv$ and that $\bQ$ is downwards-closed,
  $\bu \bv$ is in $\bQ$.

\item[\ref{caca2}] Here also, the right-to-left implication is
  trivial.  Conversely, assume $\bA \cdot \bP \subseteq \bB \cdot \bQ$
  and $\bA = (I_1+\epsilon)$ and $\bB = (I_2 + \epsilon)$ for some
  $I_1 \subseteq I_2 \in \Idl(X)$.  Let $\bu \in \bP$. Pick
  $x \in I_1$: $\bx \bu \in \bA \cdot \bP$, thus
  $\bx \bu \in \bB \cdot \bQ$. Therefore, $\bu \in \bQ$ since
  sequences of $\bB$ have length at most one.

\item[\ref{caca3}] The left-to-right implication is trivial, since
  $\epsilon \in \bA$. For the other implication, we consider some
  $\bu \in \bA$ and $\bv\in\bP$, and we have to show that
  $\bu\bv\in\bB\cdot\bQ$.  Since $\bP\subseteq\bB\cdot\bQ$, we can
  factor $\bv$ as $\bv_1\bv_2$ with $\bv_1 \in\bB$ and
  $\bv_2\in\bQ$. We claim that $\bu\bv_1\in\bB$: if $\bB=D^*$ is an
  atom of the second kind, the claim follows from $\bA\subseteq \bB$;
  if $\bB=I+\epsilon$ is an atom of the first kind, then we are in
  case \eqref{emi0-4}, $\bA$ is $\emptyset^*$, and
  $\bu=\epsilon$. With $\bu\bv_1\in\bB$ we have $\bu\bv\in\bB\cdot\bQ$
  as needed.
\qed
\end{enumerate}
\end{proof}

The next lemma deals with the complementation of filters:
\begin{lemma}[\CFsanspar]
  \label{lem-res-w-as-fup}
  Given $\bm{w}\in X^*$, the downwards-closed set
  $X^*\setminus \upc\bm{w}$ can be computed inductively using the
  following equations:
  \begin{align}
    \label{eqn1}  X^* \setminus \upc \epsilon &= \emptyset \text{ (empty union), } \\
    \label{eqn2}  X^* \setminus \upc x\bv &=\begin{cases}
      (X \setminus \upc x)^* &\text{ if $\bv=\epsilon$,}\\
      (X \setminus \upc x)^* \cdot (X + \epsilon) \cdot (X^* \setminus \upc
      \bv)
      &\text{otherwise.}
    \end{cases}
  \end{align}
\end{lemma}

Note that $X$ might not be an ideal, in which case $X + \epsilon$ is
not an atom in \cref{eqn2}. In this case, one has to first get the
ideal decomposition $X = \bigcup_i I_i$ from a presentation of
$(X,{\le})$ and use distributivity of concatenation over unions.

In the commonly encountered case where $X$ is a
finite alphabet, ordered by equality, there is no need to distribute, and indeed, the
complement of a filter is always an ideal.
More precisely, if $X = \{ a_1, \dots, a_n \}$
is a finite alphabet under equality, then
one checks easily that
$(X \setminus \upc a_i)^* \cdot (X + \epsilon) = \{a_j \mid j \neq
i\}^* \cdot (a_i + \epsilon)$.  It follows that complement of filters
are ideals in this case.

\begin{remark}
  Kabil and Pouzet~\cite{kabil92} use the following (equivalent)
  expression to complement filters:
  \begin{equation}
    \label{eq-kabil92-negfil}
    X^* \setminus \upc x y\bw =
    (X \setminus \upc x)^* \cdot \left[ \dwc (\upc x \cap \upc y) +
      \epsilon\right ] \cdot (X^* \setminus \upc y\bw) \:.
  \end{equation}
  We used a different formula because, in general, our setting does
  not guarantee that the expression $\dwc U$ is computable for
  $U \in \Up(X)$, even in the particular case where
  $U=\upc x \cap \upc y$. It is fair to mention that Kabil and Pouzet
  make no claim on computability.

  Still, \cref{eq-kabil92-negfil} is interesting when $X$ is a finite
  alphabet since then the expression $\upc x \cap \upc y$ either
  denotes the empty set or $(x+\epsilon)$, depending on whether $x$
  and $y$ coincide. Therefore, using \cref{eq-kabil92-negfil}, one
  directly obtains an ideal written in canonical form (a notion
  defined below, in \cref{sec-canonical-products}).
  \qed
\end{remark}

\begin{proof}[of \cref{lem-res-w-as-fup}] \\
We only prove the second case of
  \cref{eqn2} since the other equalities are obvious.
\\
\noindent
$(\supseteq)$: Let $\bm{w}' = \bm{u y w}$ with $\bm{u} \in (X
\setminus \upc x)^*$, $\bm{y} \in X + \epsilon$ and $\bm{w} \in (X^*
\setminus \upc \bm{v})$.  Thus $\bm{v}\not\leq_*\bm{w}$. Since
$\bm{y}$ has length at most 1, we deduce $x \bm{v}\not\leq_* \bm{y
w}$.  Since all elements in $\bm{u}$ are taken from $X\setminus\upc
x$, we further have $x\bm{v}\not\leq_*\bm{u y v}$.  Therefore
$\bm{w'}\in X^*\setminus\upc x\bm{v}$.  \\
\noindent
$(\subseteq)$: Let $\bm{w}' \notin \upc x\bm{v}$. Then either $\bm{w}'
\in (X \setminus \upc x)^*$, or we can write $\bm{w}' = \bm{u}y\bm{w}$
with $\bm{u} \in (X\setminus \upc x)^*$ and $y \ge x$. Moreover,
$\bm{w} \notin \upc \bm{v}$, since otherwise $x\bm{v} \le_* y\bw \le_*
\bu y\bw = \bw'$. Therefore, $\bw' \in (X \setminus \upc x)^* \cdot X
\cdot (X^* \setminus \upc \bv)$. Joining the two cases, and since
$\epsilon\in (X^*\setminus \upc\bv)$, we obtain the required $\bw'\in(X
\setminus \upc x)^* \cdot (X+\epsilon) \cdot (X^* \setminus \upc \bv)$.
\qed
\end{proof}

We now show how to intersect ideals:
\begin{lemma}[\IIsanspar]
  \label{lem-inter-products}
  The intersection of two ideals of $(X^*, {\le_*})$ can be computed
  inductively using the following equations:
  \begin{align}
    \label{int-idl-seq-1}
    \bepsilon \cap \bQ =
    \bP \cap \bepsilon &= \bepsilon \:,
    \\
    \label{int-idl-seq-2}
    D_1^* \cdot \bP \cap D_2^* \cdot \bQ &=
                                           (D_1\cap D_2)^* \cdot \left[
                                           \begin{array}{rl}       & (D_1^* \cdot \bP) \cap \bQ\\
                                             \cup\; & \bP \cap (D_2^* \cdot \bQ)
                                           \end{array}\right],
    \\
    \label{int-idl-seq-3}
    (I_1+\epsilon) \cdot \bP \cap (I_2+\epsilon)\cdot \bQ
                       &=
                         \left[\begin{array}{rl}
                                 &  ((I_1+\epsilon) \cdot \bP) \cap \bQ              \\
                                 \cup\;   &  \bP \cap ((I_2+\epsilon) \cdot \bQ)              \\
                                 \cup\;   &  \bigl( (I_1\cap I_2)+\epsilon \bigr) \cdot (\bP \cap \bQ)
                               \end{array}
                                            \right],
    \\
    \label{int-idl-seq-4}
    D^* \cdot \bP \cap (I+\epsilon)\cdot \bQ &=
                                               \left[\begin{array}{rl}
                                                       &  \bP \cap ((I+\epsilon) \cdot \bQ)              \\
                                                       \cup\; &\bigl( (D\cap I)+\epsilon \bigr) \cdot ((D^* \cdot \bP) \cap \bQ)
                                                     \end{array}
                                                                \right].
  \end{align}
\end{lemma}

Here also, some shortcuts are used. For instance, the intersection of
two ideals need not be an ideal. Therefore, $(I_1 \cap I_2)+\epsilon$
in \cref{int-idl-seq-3} might not be an ideal.  As before, by
decomposing downwards-closed sets as union of ideals, and distributing
concatenations over unions, one can compute the actual ideal
decomposition of the intersection of two ideals of $(X^*, {\le_*})$.

\begin{proof}[of \cref{lem-inter-products}]
  \Cref{int-idl-seq-1} is obviously correct.  The other right-to-left
  inclusions are easily checked using \cref{prop:incl-seq}.  For the
  left-to-right inclusions:
  \begin{description}
  \item[\cref{int-idl-seq-2}:] Let
    $\bu \in D_1^* \cdot \bP \cap D_2^* \cdot \bQ$.  Let $\bv$ be
    the longest prefix of $\bu$ which is in $D_1^*$.  Without loss of
    generality, we assume that the longest prefix of $\bu$ which is in
    $D_2^*$ is longer than $\size{\bv}$, and thus can be written
    $\bv \bw$ for some $\bw \in D_2^*$. Moreover, there exists
    $\bt \in X^*$ so that $\bu = \bv \bw \bt$. We have
    $\bv \in (D_1 \cap D_2)^*$, $\bw\bt \in \bP$ and $\bt \in
    \bQ$. Therefore, $\bw\bt \in \bP \cap D_2^*\cdot \bQ$.

  \item[\cref{int-idl-seq-3}:] Consider any word in
    $(I_1+\epsilon) \cdot \bP \cap (I_2+\epsilon)\cdot \bQ$.  If it is
    empty, it is also in the right-hand side of \cref{int-idl-seq-3},
    so we assume that it of the form $x \bu$. Depending on whether
    $x \in I_1 \setminus I_2$, $x \in I_2 \setminus I_1$ or
    $x \in I_1 \cap I_2$, $\bu$ is easily proved to be in
    $((I_1+\epsilon) \cdot \bP) \cap \bQ$, in
    $\bP \cap ((I_2+\epsilon) \cdot \bQ)$, or in
    $((I_1\cap I_2)+\epsilon ) \cdot (\bP \cap \bQ)$.  If $x$ is
    neither in $I_1$ nor $I_2$, then $x\bu$ belongs to all three sets.

  \item[\cref{int-idl-seq-4}:] This is similar, combining arguments
    from the previous two cases. \qed
\end{description}
\end{proof}

We now turn to intersecting filters:
\begin{lemma}[\IFsanspar]
  The intersection of two filters can be computed inductively using
  the following equations:
\begin{align}
\label{eq-int-fil-seq-1}
    \upc \bv \cap \upc\epsilon &=
    \upc \epsilon \cap \upc \bv = \upc \bv \:,
\\
\label{eq-int-fil-seq-2}
    \upc  x\bv \cap \upc  y\bw &=
   \left[\begin{array}{l}
      (\upc  \bx) \cdot (\upc \bv \cap \upc  y\bw) \; \cup\;
      (\upc  \by) \cdot (\upc   x\bv \cap \upc \bw)
      \\
      \cup\;  (\upc_X x \cap \upc_X y) \cdot (\upc \bv \cap \upc \bw)
   \end{array}\right],
\end{align}
  where $\bv,\bw \in X^*$ and $x,y \in X$. The actual filter
  decomposition in the last equation is obtained using $(\upc
  \bu)\cdot(\upc \bu')=\upc(\bu\bu')$ and distributivity over unions.
\end{lemma}

\begin{proof}
  \Cref{eq-int-fil-seq-1} and the ``${\supseteq}$'' half of
  \cref{eq-int-fil-seq-2} are obvious. For the remaining
  ``${\subseteq}$'' half, we consider
  $\bu \in \upc x\bv \cap \upc y\bw$.  Let us write $\bu$ as
  $\bu=\bu_1 z \bu_2$ where $z\bu_2$ is the shortest suffix of $\bu$
  in $\upc x\bv \cap \upc y\bw$---this suffix cannot be empty since it
  contains $x\bv$ and $y\bw$ as embedded sequences.  Note that $z$
  must be above $x$ or $y$ in $X$, otherwise $\bu_2$ would be a
  shorter suffix of $\bu$ in $\upc x\bv \cap \upc y\bw$.   One now
  considers   whether $z$ is above $x$, $y$, or both, and
  picks the corresponding
  summand in the right of \Cref{eq-int-fil-seq-2}.
  \qed
\end{proof}

Finally, we focus on the complementation of ideals. This operation
requires more work, and is decomposed in several lemmas.  We first
show how to complement atoms, and then how to complement products of
atoms.
\begin{itemize}
\item If $D \subseteq X$ is downwards-closed, then
  $X^* \setminus D^*$, also written $\neg D^*$, consists of all
  sequences having at least one element not in $D$. One first computes
  $X \setminus D = \upc a_1 \cup \cdots \cup \upc a_n$, using \CI for
  $X$. Then
  $\neg D^* = \upc_{X^*} \ba_1 \cup \cdots \cup \upc_{X^*} \ba_n$.
\item If $I \subseteq X$ is an ideal, $\neg (I + \epsilon)$ consists
  of all sequences of length at least 2, as well as all sequences
  having an element not in $I$. The latter is obtained as in the
  previous case, by computing
  $X\setminus I=\upc b_1 \cup \cdots \cup \upc b_m$ in $X$. The former
  is $\upc_{X^*} (X \cdot X)$, easily computed in a similar way using
  \XF for $X$.
\end{itemize}

We now consider products $\bA_1 \cdots \bA_n$ of atoms. We know how to
compute $\bU_i = \neg \bA_i$. One has
$\neg(\bA_1\cdots \bA_n) = \neg(\neg \bU_1 \cdots \neg \bU_n)$, and
this motivates the following definition:
\begin{definition}
  Define the operator
  $\odot \colon \Up(X^*) \times \Up(X^*) \to \Up(X^*)$
  as \\
  $\bU \odot \bV := \neg ( \neg \bU \cdot \neg \bV)$.
\end{definition}

Note that $\bU \odot \bV$ is upwards-closed when $\bU$ and $\bV$ are.
The operation $\odot$ is easily shown associative using the
associativity of the product, thus
$\bU_1 \odot \cdots \odot \bU_n = \neg( \neg \bU_1 \cdot \cdots \cdot
\neg \bU_n )$.  The previous relation becomes
$\neg (\bA_1 \cdots \bA_n) = \bU_1 \odot \cdots \odot \bU_n$, and it
only remains to show that $\odot$ is computable on upwards-closed
sets.  In what follows, we will often use the following obvious
characterization: $\bw \in \bS \odot \bT$ if and only if for all
factorizations $\bw=\bw_1\bw_2$, $\bw_1 \in \bS$ or $\bw_2 \in \bT$.

We first show that $\odot$ is computable on principal filters, then we
show how to complement ideals.
\begin{lemma}
  \label{odot-of-cones}
  On principal filters, $\odot$ can be computed using the following
  equations:
  \begin{align}
    \label{eq-odot-fil-1}
    \upc \bv \odot \upc \epsilon =
    \upc \epsilon \odot \upc \bv &= X^* \:,
    \\
    \label{eq-odot-fil-2}
    \upc \bv a \odot \upc b\bw &= \upc (\bv ab\bw) \cup (\upc \bv) \cdot (\upc_X a
                                 \cap \upc_X b) \cdot (\upc \bw) \:,
  \end{align}
  where $\bv,\bw \in X^*$ and $a,b \in X$.
\end{lemma}

\begin{proof}
  \Cref{eq-odot-fil-1} is clear. We concentrate on
  \cref{eq-odot-fil-2}:
  \\
  \noindent
  $(\supseteq)$ If $\bu \geq_* \bv ab \bw$, then for every
  factorization of $\bu = \bu_1 \bu_2$, the left factor $\bu_1$ is
  above $\bv a$, or the right factor $\bu_2$ is above $b \bw$, and
  thus $\bu \in \upc \bv a \odot \upc b\bw$.  If
  $\bu \geq_* \bv c \bw$, where $c \in X$ is such that $c \geq a$ and
  $c \geq b$, then in every factorization of $\bu$ as $\bu_1 \bu_2$,
  $c$ appears either in the left factor $\bu_1$ or in the right factor
  $\bu_2$, and this suffices to show that either $\bu \ge_* \bv a$ or
  $\bu \ge_* b\bw$.
  \\
  \noindent
  $(\subseteq)$ Let $\bu \in (\upc \bv a) \odot (\upc b \bw)$. From
  the factorizations $\bu = \bu \cdot \epsilon$ and
  $\bu = \epsilon \cdot \bu$ we get $\bv a \leq_* \bu$ and
  $b\bw \leq_* \bu$. Consider the shortest prefix of $\bu$ above
  $\bv a$ and the shortest suffix above $b\bw$. These factors cannot
  have an overlap of length $\geq 2$, otherwise splitting  $\bu$ in
  the middle of the overlap would provide a shorter factor above $\bv
  a$ or one above $b\bw$, contradicting our assumption.
  If the factors do not
  overlap, we get $\bu \geq_* \bv ab\bw$. If they overlap, necessarily
  over a single letter $c\in X$, we write
   $\bu = \bu_1 c \bu_2$. Then $\bu_1 \geq_* \bv$, $c \geq a$, $c \geq b$, and
  $\bu_2 \geq_* \bw$, which proves the statement.
  \qed
\end{proof}

\begin{lemma}[\CIsanspar]
  Complementing ideals of $(X^*,{\le_*})$ is computable.
\end{lemma}

\begin{proof}
  Given an ideal $\bP = \bA_1 \cdots \bA_n$, its complement is
  $\neg \bP = \neg \bA_1 \odot \cdots \odot \neg \bA_n$.  Using the
  procedure to complement downwards-closed sets of $(X,{\le})$, we can
  write each $\neg \bA_i$ as a union of filters.  Since $\odot$
  distributes over unions of upwards-closed sets (from
  \cref{lem-concat-distrib-inter} by duality), we can write $\neg \bP$
  as a finite union of sets of the form
  $F_1 \odot F_2 \odot \cdots \odot F_n$, where the $F_i$'s are
  filters. Finally, \cref{odot-of-cones} allows us to reduce these
  expressions to a finite union of filters.
  \qed
\end{proof}

\subsubsection{A proof of \cref{thm-idl-higman}.}
\label{subsec-proof-of-idl-higman}

One direction of the theorem is easy to check: products of atoms are
indeed ideals (downwards-closed and directed) of $(X^*,{\le_*})$.  For
the other direction, consider an arbitrary ideal $I$ of
$(X^*,{\le_*})$. Its complement is upwards-closed, hence can be
written $\neg I = \bigcup_{i<n} F_i$ for some filters $F_1$,~\ldots,
$F_n$.  Therefore,
\[
I =  \neg \bigcup_{i<n} F_i =
\bigcap_{i<n} \neg F_i.
\]
Now, since any $\neg F_i$ is a finite union of products of atoms (see
\cref{lem-res-w-as-fup}), by distributing the intersection over the
unions, we are left with a finite union of finite intersections of
products of atoms. Since these intersections can be decomposed as
finite unions of products of atoms (see \cref{lem-inter-products}), we
have decomposed the ideal $I$ into a finite union of products of
atoms. Since products of atoms are ideals (cf.\ first direction), and
since ideals are prime subsets (by \cref{prop:prime-idl}), we obtain
that $I$ is actually equal to one of those products of atoms (by
\cref{lem-irreduc}).  \\

This proof highlights a general technique for identifying the ideals
of some \wqo: if we have some subclass $\Jcal$ of the ideals such that
the complement of any filter can be written as a finite union of
ideals of $\Jcal$, and the intersection of any two ideals of
$\Jcal$ can be written as a finite union of ideals of $\Jcal$, then
$\Jcal$ is the class of all ideals.

\subsubsection{Uniqueness of ideal representation.}
\label{sec-canonical-products}

Writing ideals as products of atoms can be done in several ways.  For
example $D^*\cdot D^*$ and $D^*$ coincide. They also coincide with
$D^*\cdot (I+\epsilon)$ and $D^*\cdot D'{}^*$ if $I$, resp.\ $D'$, are
subsets of $D^*$.

More generally, if $\bA$ is an atom and $D\in\Down(X)$ is such that
$\bA \subseteq D^*$, then $\bA D^* = D^* \bA = D^*$. Subsequently, we
show that these are the only causes of non-uniqueness: avoiding such
redundancies, every ideal has a unique representation as a product of
atoms.  (This was already observed for finite alphabets in
\cite{abdulla-forward-lcs}.)  This can be used to define a canonical
representation for ideals of $X^*$ (assuming one has defined a
canonical representation for the ideals of $X$) and then for the
downwards-closed sets. This representation is easy to use (moving from
an arbitrary product of atoms to the canonical representation just
requires testing inclusions between atoms) and can lead to more
efficient algorithms.

Below, we use letters such as $\mathtt{A}, \mathtt{P}$, etc., to
denote sequences of atoms (syntax), and corresponding letters such as
$\bA$, $\bP$, etc to denote the ideals obtained by taking the product
(semantics). For example if
$\mathtt{P} = (\mathtt{A}_1, \mathtt{A}_2, \cdots, \mathtt{A}_n)$,
then $\bP = \bA_1\cdot\bA_2 \cdot \cdots \cdot \bA_n$. Thus it is
possible to have $\mathtt{P} \neq \mathtt{Q}$ and $\bP=\bQ$.

\begin{definition}
  \label{def-atom-reduced}
  A sequence of atoms $\bA_1, \cdots, \bA_n$ is said to be
  \emph{reduced} if for every $i$, the following hold:
  \begin{itemize}
  \item $\bA_i \neq\emptyset^*= \{\epsilon\}$;
  \item if $i < n$ and $\bA_{i+1}$ is some $D^*$, then
    $\bA_i \not \subseteq \bA_{i+1}$;
  \item if $i > 1$ and $\bA_{i-1}$ is some $D^*$, then
    $\bA_i \not \subseteq \bA_{i-1}$.
  \end{itemize}
\end{definition}

Every ideal has a reduced decomposition into atoms, since any
decomposition can be converted to a reduced one by dropping atoms
which are redundant as per \cref{def-atom-reduced}. It remains to show
that reduced representations are unique:
\begin{theorem}
  \label{theo-ideal-uniq-dec}
  If $\mathtt{P}$ and $\mathtt{Q}$ are reduced sequences of atoms such
  that $\bP = \bQ$, then $\mathtt{P} = \mathtt{Q}$.
\end{theorem}
\begin{proof}
  Let us first observe that the claim is obvious for atoms: $\bA=\bB$
  entails $\mathtt{A}=\mathtt{B}$.  We now prove the statement in
  several steps: let $\mathtt{A} \cdot \mathtt{P}$ and $\mathtt{B}
  \cdot \mathtt{Q}$ be two reduced sequences of atoms.  \\

\noindent \textit{First claim:} $\bA \cdot \bP \neq \bP$:\\
      By induction on $\mathtt{P}$. If $\mathtt{P}$ is the empty
      sequence, then $\bP = \{\epsilon\}$ and $\bA \cdot \bP = \bA$.
      Now \cref{def-atom-reduced} guarantees $\bA \neq \{\epsilon\}$.
      Otherwise, $\mathtt{P}$ is some $\mathtt{A}' \cdot \mathtt{P}'$.
      If $\bA \cdot \bP \subseteq \bP$, the inclusion test described
      in \cref{prop:incl-seq} implies either $\bA \subseteq \bA'$,
      which contradicts reducedness, or $\bA \cdot \bA' \cdot \bP'
      \subseteq \bP'$, which entails $\bA' \cdot \bP' \subseteq \bP'$
      and contradicts the induction hypothesis. Therefore $\bA \cdot
      \bP \neq \bP$.
\\

\noindent \textit{Second claim:} $\bA \cdot \bP = \bB \cdot \bQ$ implies $\bA = \bB$:\\
      Since $\bA \cdot \bP \subseteq \bB \cdot \bQ$,
      \cref{prop:incl-seq} implies either $\bA
      \subseteq \bB$, or $\bA \cdot \bP \subseteq \bQ$. The second
      option, combined with $\bQ \subseteq \bB \cdot \bQ \subseteq \bA
      \cdot \bP$, leads to $\bQ = \bB \cdot \bQ$, which is impossible
      (first claim). Therefore, $\bA \subseteq \bB$, and the reverse
      inclusion is proved symmetrically.
\\

\noindent \textit{Third claim:} $\bA \cdot \bP = \bB \cdot \bQ$ and $\bA = \bB$ imply
    $\bP = \bQ$:\\
      If $\mathtt{Q}$ is the empty sequence, then
      $\bA\cdot\bP=\bB\cdot\bQ=\bB=\bA$, thus $\bP\subseteq\bA$. But
      if $\mathtt{P}$ is some $\mathtt{A}'\cdot\mathtt{P}'$ then by
      \cref{prop:incl-seq} either $\bA'\subseteq\bA$, which is
      impossible by reducedness of $\mathtt{A}\cdot\mathtt{P}$, or
      $\bA'\cdot\bP' \subseteq \{\epsilon\}$, requiring
      $\bA'\subseteq\{\epsilon\}$ which is also impossible. Thus
      $\mathtt{P}$ too is the empty sequence.

      If $|\mathtt{P}|=0$ the same reasoning applies so we now assume
      that both products are non-trivial, writing $\mathtt{P} =
      \mathtt{A}' \cdot \mathtt{P}'$ and $\mathtt{Q} = \mathtt{B}'
      \cdot \mathtt{Q}'$.  If now $\mathtt{A}$ is $I + \epsilon$ for
      some $I$, then so is $\mathtt{B}$ and \cref{prop:incl-seq}
      implies $\bA'\cdot \bP' \subseteq \bB'\cdot \bQ'$.  Otherwise,
      $\mathtt{A}$ is $D^*$ for some $D$, in which case
      \cref{prop:incl-seq} entails first $\bA'\cdot\bP' \subseteq \bB
      \cdot \bB' \cdot \bQ'$, then $\bA' \cdot \bP' \subseteq \bB'
      \cdot \bQ'$ (since $\bA' \not\subseteq \bA=\bB$ by reducedness
      of $\mathtt{A} \cdot \mathtt{P}$).  In other words, we deduce
      $\bP \subseteq \bQ$ and the reverse inclusion is proved
      symmetrically.  \\

\noindent \textit{Proof of the Theorem:}
      By induction on $\size{\mathtt{P}} + \size{\mathtt{Q}}$.  If
      either $\mathtt{P}$ or $\mathtt{Q}$ is the empty sequence, the
      property is trivially verified, otherwise we can write
      $\mathtt{P} = \mathtt{A} \cdot \mathtt{P}'$ and $\mathtt{Q} =
      \mathtt{B} \cdot \mathtt{Q}'$.  From $\bP = \bQ$ we deduce $\bA
      = \bB$ (second claim), which in turn implies $\bP' = \bQ'$
      (third claim), hence $\mathtt{P}' = \mathtt{Q}'$ by induction
      hypothesis. We already noted that $\bA = \bB$ implies
      $\mathtt{A} = \mathtt{B}$ and combining those gives $\mathtt{P}
      = \mathtt{Q}$.
\qed
\end{proof}

 \subsection{Finitary powersets}
\label{sec-powerset}

Given a \qo $(X,{\leq})$, we write $\Pcal(X)$ to denote its powerset,
with typical elements $S$, $T$,~\@\ldots{} A usual way of extending
the quasi-ordering between elements of $X$ into a quasi-ordering
between sets of such elements is the \emph{Hoare quasi-ordering} (also
called \emph{domination quasi-ordering}), denoted $\hoare$, and
defined by
\[
S\hoare T
\equivdef
\forall x\in S:\exists y\in T: x\leq y.
\]
A convenient characterization of this ordering is the following:
$S\hoare T$ iff $S \subseteq \dwc_X T$.
Note that
$(\Pcal(X),{\hoare})$ is in general not antisymmetric even when
$(X,{\leq})$ is. For example, and
writing
${\equiv_H}$ to denote ${\hoare} \cap {\invhoare}$,
the above characterization implies that
$S \equiv_H \dc_X S$ for any $S\subseteq X$.
In particular, this shows that the quotient $\Pcal(X) / {\equiv_H}$ is
isomorphic to $(\Down(X), {\subseteq})$. While
$(\Down(X), {\subseteq})$ is well-founded if, and only if, $(X,{\le})$
is a \wqo (cf.\ \cref{lem-ACC}), this does not guarantee that
$(\Down(X), {\subseteq})$ is a \wqo, as famously shown by
Rado~\cite{rado54}.\footnote{In fact $(\Pcal(X), {\hoare})$ is a \wqo
  iff $X$ is an $\omega^2$-\wqo~\cite{marcone2001,jancar99c}.}  In
other words, powerset is not a \wqo-preserving construction.

However, the \emph{finitary powerset} construction is \wqo-preserving.
Let $\Pf(X)$,  sometimes also written $[X]^{<\omega}$, denote the set
of all \emph{finite subsets} of $X$.
\begin{theorem}
  $(\Pf(X),{\hoare})$ is \wqo if, and only if, $(X, {\leq})$ is \wqo.
\end{theorem}
The if direction is an easy consequence of Higman's Lemma: the
function that maps each word in $X^*$ to its set of letters, in
$\Pf(X)$, is monotonic and surjective, and the image of a \wqo by any
monotonic map is \wqo.  We shall see another proof in
\cref{sec-equiv}.

\begin{proposition}[Ideals of $\pmb{\Pf(X)}$]
  \label{prop:idl-pfX}
  Given a \wqo $(X,{\le})$, the ideals of $(\Pf(X), {\hoare})$ are
  exactly the sets $\Jcal$ of the form $\Jcal=\Pf(D)$ for
  $D\in\Down(X)$.
\end{proposition}

\begin{proof}
  $(\DaG)$ : $\emptyset \in \Pf(D)$, so $\Pf(D)$ is non-empty. It is
  downwards-closed, since if $S \hoare T \in \Pf(D)$, then
  $S \subseteq \dc_X T \subseteq \dc_X D = D$. It is directed, since
  if $S, T \in \Pf(D)$, then $S \cup T \in \Pf(D)$, and
  $S, T \hoare S \cup T$.
\\
\noindent
  $(\GaD)$ : Let $\Jcal$ be an ideal of $\Pf(X)$ and let
  $D = \bigcup_{S \in \Jcal} S$, so that $\Jcal \subseteq
  \Pf(D)$. Since $\Jcal$ is downwards-closed under $\hoare$, $D$ is
  downwards-closed under $\leq$ and $\{x\} \in \Jcal$ for all
  $x \in D$. Since $\Jcal$, being an ideal, is non-empty,
  $\emptyset \in \Jcal$. Finally, if $S, T \in \Jcal$, then there is
  some $U \in \Jcal$ such that $S, T \hoare U$. Thus
  $S \cup T \hoare U$, and therefore $S \cup T \in \Jcal$. Therefore,
  $\Jcal$ contains the empty set, all the singletons included in $D$,
  is closed under finite unions, and so is equal to $\Pf(D)$.
\qed
\end{proof}

When $(X,{\leq})$ is ideally effective,
finite subsets of $X$ can be represented using any of the usual data
structures and
\cref{prop:idl-pfX} directly
leads to a data structure for $\Idl(\Pf(X))$ inherited from the
representation of $X$'s ideals and downwards-closed sets.

\begin{theorem}
With the above representations, the finitary powerset with Hoare's
ordering is an ideally effective constructor.
\end{theorem}

\begin{proof}
  Let $(X,{\le})$ be an ideally effective \wqo.  In the following we
  use shorthand notations such as $\dwc_H$ for $\dwc_{\Pf(X)}$, etc.,
  with the obvious meaning.
  \begin{description}
  \item[\OD:] The sets we consider being finite, the definition of
    $\hoare$ leads to an obvious implementation.
  \item[\ID:] Testing inclusion in $\Idl(\Pf(X))$ reduces to testing
    inclusion in $\Down(X)$ by
    $\Jcal_1=\Pf(D_1) \subseteq \Jcal_2=\Pf(D_2) \iff D_1 \subseteq
    D_2$.

  \item[\PI:] Given $S$ a finite subset of $X$, the principal ideal
    $\dwc_H S$ is $\Pf(\dwc_X S)$ so we just need to compute the
    downwards-closed $\dwc_X S= \bigcup_{x \in S} \dwc x$ in $X$'s
    representation.
  \item[\CF:] Given $S \in \Pf(X)$, the complement of $\upc_H S$ can
    be given an ideal decomposition via
    \begin{align*}
      \Pf(X) \setminus \upc_H S
      &= \bigcup_{x \in S} \Pf(X) \setminus \upc_H \{x\}
        = \bigcup_{x \in S} \Pf(X \setminus \upc x)
        \:.
    \end{align*}
    This can now be computed using \CF for $X$.
  \item[\II:] We have $\Pf(D_1) \cap \Pf(D_2) = \Pf(D_1 \cap D_2)$.
  \item[\IF:] Filters may be intersected using
    $\upc S \cap \upc T = \upc (S \cup T)$.
  \item[\CI:] Given an ideal $\Jcal=\Pf(D)$, $\Pf(X) \setminus \Jcal$
    consists of the sets that contain at least one element not in
    $D$. That is:
    \begin{align*}
      \neg\Jcal =  \upc_H \{ x_1 \}\cup \cdots
      \cup\upc_H \{x_n\} \text{ if } X\setminus D=\upc_X x_1\cup\cdots\cup
      \upc_X x_n,
    \end{align*}
    which is computable using \CI for $X$.
    
  \end{description}

  The above proves that the finite powerset constructor is an ideally
  effective constructor. Once again, the computability of the
  presentation described above from a presentation of $(X,{\le})$ is
  clear.  For \XI, observe that $\Pf(X)$ is its own ideal
  decomposition since $X\in\Down(X)$. For \XF, use
  $\Pf(X) = \upc_H\emptyset$.
  \qed
\end{proof}

\section{More constructions on ideally effective \wqos}
\label{sec-more-constr}

In this section we describe more constructions that yield new ideally
effective \wqos from previously defined ones.  By contrast with the
constructors of \cref{sec-idl-eff-constructors}, these constructions
take some extra parameters that are not \wqos---for example, an
equivalence relation in order to build quotient \wqos (see
\cref{sec-quotient}). Showing that the quotient \wqo is ideally
effective will need some effectiveness assumptions on the equivalence
at hand, in the spirit of what we did with the one-sorted
constructors.

\subsection{Order extension}
\label{sec-extension}

Let $(X, {\leq})$ be a \wqo and let $\le'$ be an extension of $\le$
(i.e., $\le \:\subseteq\: \le'$). Then $(X, {\le'})$ is also a
\wqo. In this subsection, we investigate the ideals of $(X,{\le'})$
and present some sufficient condition for $(X,{\le'})$ to be ideally
effective, assuming $(X, {\leq})$ is.  In the next subsections, we
will present natural applications of this framework.

\begin{proposition}
  \label{ideals-under-extension}
  Given a \wqo $(X,{\leq})$ and an extension $\le'$ of $\le$, the
  ideals of $(X, {\leq'})$ are exactly the downward closures under
  $\le'$ of the ideals of $(X, {\leq})$. That is,
  \begin{align*}
    \Idl(X, {\le'}) = \{\dwcp I \mid I \in \Idl(X, {\leq}) \}\:.
  \end{align*}
\end{proposition}

\begin{proof}
  $(\supseteq)$ Let $I\in\Idl(X,{\le})$.  Even though $I$ may not be
  downwards-closed in $(X, {\le'})$, it is still directed.  It is easy
  to see that $\dwcp I$ is directed, non-empty, and downwards-closed
  for $\le'$. Thus it is an ideal of $(X, {\le'})$.
  \\
  \noindent
  $(\subseteq)$ Let $J$ be an ideal of $(X,{\le'})$. $J$ may not be
  directed in $(X, {\leq})$, but it is still downwards-closed under
  $\le$. As a consequence, it can be decomposed as a finite union of
  ideals of $(X, {\leq})$: $J = I_1 \cup \dots \cup I_n$. Then
  $J = \dwcp J = \dwcp I_1 \cup \dots \cup \dwcp I_n$. Now applying
  \cref{lem-irreduc} to $(X,{\leq'})$, we have $J = \dwcp I_i$ for
  some $i$.  \qed
\end{proof}

Assume that $(X,{\leq})$ is an ideally effective \wqo for which we
have a presentation at hand, in particular data structures for $X$ and
$\Idl(X)$.  Let $\le'$ be an extension of $\le$.  To represent
elements of $(X,{\le'})$, it is natural to use the same data structure
for $X$ as the one used for $(X,{\leq})$.  For ideals,
\cref{ideals-under-extension} suggests to also use the same data
structure as the one for ideals of $X$. That is, an ideal
$J \in \Idl(X,{\le'})$ will actually be represented by any
$I \in \Idl(X)$ such that $J = \dwc_{\le'} I$.

Using these representations for $(X,{\le'})$ does not always lend
itself to algorithms that would witness ideal effectiveness, even
under the assumptions that $(X,{\leq})$ is ideally effective and that
$\le'$ is decidable. There is even is a ``natural'' counter example:
the lexicographic ordering over $X\times X$ (see
\cref{sec-alternatives}).  This fact justifies that we make further
assumptions.  More precisely, we show that $(X, {\le'})$ is ideally
effective if we can compute downward closures under $\le'$:
\begin{theorem}
  \label{thm-ext-idl-eff}
  Let $(X,{\leq})$ be an ideally effective \wqo and $\le'$ an extension
  of $\le$. Then, $(X,{\le'})$ is ideally effective for the
  aforementioned data structures of $X$ and $\Idl(X,{\le'})$,
  whenever the following functions are computable:
\begin{xalignat*}{2}
  \ci:&
  \begin{array}{rl}
   \Idl(X, {\leq}) &\to \Down(X, {\leq}) \\
   I &\mapsto \dwcp I
  \end{array}
&
\cf:&  
  \begin{array}{rl}
   \Fil(X, {\leq}) &\to \Up(X, {\leq}) \\
   \upc x &\mapsto \upc_{\le'} (\upc x) = \upc_{\le'} x
  \end{array}
\end{xalignat*}
Moreover, under these assumptions, a presentation of $(X,{\le'})$ can
be computed uniformly from a presentation of $(X,{\leq})$ and
algorithms realizing $\ci$ and $\cf$.
\end{theorem}
Note that if $I \in \Idl(X, {\leq})$, then $\dwcp I$ is also
downwards-closed for $\le$ and thus can be represented as a
downwards-closed set of $(X,{\leq})$. This is precisely this
representation that the function $\ci$ outputs. Same goes for
$\cf$. Note that using functions $\ci$ and $\cf$, it is possible to
compute the downward and upward closure under $\le'$ of arbitrary
downwards- and upwards-closed sets for $\le$ using the canonical
decompositions:
$\dwcp (I_1 \cup \dots \cup I_n) = (\dwcp I_1) \cup
\dots \cup (\dwcp I_n)$ and
$\upc_{\le'} (\upc x_1 \cup \dots \cup
\upc x_n) = \upc_{\le'} x \cup \dots \cup \upc_{\le'} x_n$.
\begin{proof}
We proceed to show that $(X, {\le'})$ is ideally effective.

  \begin{description}
\item[\OD:] One can tests $x \le' y$, since this is equivalent to
$y \in\cf(\upc_{\le} x)$.
\item[\ID:] Ideal inclusion can be decided using $\ci$ and the inclusion test
  for downwards-closed sets of $(X,{\leq})$:
  $\dwcp{I_1} \subseteq \dwcp{I_2} \lequiv I_1 \subseteq \ci{I_2}$.

\item[\PI:] The principal ideal $\dwcp x$ of $(X, {\le'})$ is
  represented by $\dwc_{\le} x$, since $\dwcp (\dwc_{\le} x) = \dwcp x$.

\item[\CF:] For $x \in X$, the filter complement $X \setminus
  \upc_{\le'} x$ is $X \setminus \cf(\upc_{\le} x)$ which can be
  computed, using \CF and \II for $(X, {\leq})$, as a downwards-closed
  set in $(X,{\leq})$.  This is represented by an ideal decomposition
  $D=\bigcup_{i<n}I_i$ which is canonical in $(X,{\le})$ but not
  necessarily in $(X,{\le'})$ since one may have $\dwcp I_i\subseteq
  \dwcp I_j$ for $i\neq j$. However, extracting the canonical ideal
  decomposition wrt.\ $\le'$ can be done using \ID for $(X, {\le'})$.

\item[\II:] Intersection of ideals is computed with
  $\dwcp{I_1} \cap \dwcp{I_2} = \ci(I_1) \cap \ci(I_2)$. Here again,
  this result in an ideal decomposition that is canonical for $\le$
  but not for $\le'$ until we process it
   as done for \CF.
\item[\CI, \IF:] these operations are obtained similarly.
\end{description}
With algorithms for the closure functions $\ci$ and $\cf$, the presentation above is
computable from a presentation of $(X,{\le})$. Regarding \XF and
\XI, we note that filter
and ideal decompositions of $X$ for $\le$ are also valid decompositions
for $\le'$. However, these decompositions might not be canonical for
$\le'$ even if they are for $\le$, in which case the canonical
decompositions can be obtained using $\OD$ and $\ID$, as usual.
\qed
\end{proof}

 \subsubsection{Sequences under stuttering.}
\label{sec-stutter}

In this subsection, we apply \cref{thm-ext-idl-eff} to an extension
of Higman's ordering $\leq_*$ on
finite sequences (from \cref{sec-sequences}).

Given a \qo $(X, {\leq})$, we define the \emph{stuttering ordering}
$\leqst$ over $X^*$ by
$\bm{x}=x_1\cdots x_n\leqst \bm{y}=y_1\cdots y_m$ $\equivdef$ there
are $n$ indices $1\leq p_1 \leq p_2 \leq \cdots \leq p_n\leq m$ such
that $x_i\leq y_{p_i}$ for all $i=1,\ldots,n$. Compared with Higman's
ordering, the sequence of positions $(p_i)_{i=1,\ldots,n}$ in $\bm{y}$
need not be \emph{strictly} increasing: repetitions are allowed.  For
instance, if $X = \{ a,b \}$ is a finite alphabet, then
$aabbaa \leqst aba \leqst aabbaa$ but $aabbaa \not\leqst ab$. Or with
$X = \nat$, $(1, 1, 1) \leqst (2)$. Note that even when $(X, {\leq})$ is
antisymmetric, $(X^*, {\leqst})$ need not be.

\begin{remark}
There is another way to define the stuttering ordering: define the
stuttering equivalence relation $\simst$ on $X^*$ as the smallest
equivalence relation such that for all $\bm{x}, \bm{y} \in X^*$ and $a \in X$,
$\bm{x}a\bm{y} \simst \bm{x}aa\bm{y}$. Informally, this equivalence
does not distinguish between a single and several consecutive
occurrences of a same element.
Then, $\leqst \:=\: \leq_* \circ \simst$, where $\circ$ denotes
the composition of relations.
Observe that $\simst$ is \emph{not} the same as the equivalence
relation $\equivst \:=\: \leqst \cap \gest$ induced by the ordering, even
if $(X, {\leq})$ is a partial order $\leq$. For instance, if $a \leq b$
in $X$, then $ab \equivst b$ in $X^*$, but $ab \simst b$ does not
hold. However the inclusion ${\simst} \subseteq {\equivst}$ is always
valid.
\qed
\end{remark}

Obviously, $\leqst$ is an order extension of $\leq_*$, thus $(X^*, {\leqst})$
is a \wqo when $(X,{\leq})$ is, and we can apply \cref{thm-ext-idl-eff}.

\begin{theorem}
\label{thm-stutter-idl-eff}
  The stuttering extension of a \wqo $(X,{\leq})$ is an ideally
  effective constructor.
\end{theorem}

\begin{proof}
  In the light of \cref{thm-ext-idl-eff,thm-hig-idl-eff}, it suffices to show that the following
  closure functions are computable:
\begin{xalignat*}{2}
  \ci:&
  \begin{array}{rl}
 \Idl(X^*, {\leq_*}) &\to \Down(X^*, {\leq_*}) \\
  I &\mapsto \dwc_{\mathrm{st}} I
  \end{array}
&
  \cf:&
  \begin{array}{rl}
 \Fil(X^*, {\leq_*}) &\to \Up(X^*, {\leq_*}) \\
  \upc \bu &\mapsto \upc_{\mathrm{st}} (\upc \bu) = \upc_{\mathrm{st}} \bu
  \end{array}
\end{xalignat*}
Recall from \cref{sec-extension} that the ideals of
$(X^*,{\le_*})$ are the (concatenation) products of atoms, where atoms are either of the
form $D^*$ for some $D \in \Down(X)$ or $I + \epsilon$ for some
$I \in \Idl(X)$. It is quite immediate to see that
$\ci(D^*) = D^*$ and $\ci(I+\epsilon) = I^*$, and that given two
products of atoms $\bP_1, \bP_2$,
$\ci(\bP_1 \cdot \bP_2) = \ci(\bP_1) \cdot \ci(\bP_2)$.
From these equations, it is simple to write an inductive algorithm
computing $\ci$.

Function $\cf$ is computable as well, although less
straightforward. We provide an expression for $\cf$ in \cref{stutter-cf}
which is clearly computable.
\qed
\end{proof}

\begin{lemma}
  \label{stutter-cf}
  
  Given $\bu = x_1 \cdots x_n \in X^*$ a non-empty sequence, 
  \begin{multline*}
    \cf(\bu) = \upcst \bu =
    \upc_* \left\{ y_1 \cdots y_k \Biggm\vert
      \begin{array}{l} 0 < k \le n \\
        0 = \ell_0 < \ell_1 < \dots < \ell_k = n \\ \forall j=1,\ldots,k:
        y_j \in \min (\bigcap_{\ell_{j-1} < \ell \le \ell_{j}} \upc_X x_{\ell})
      \end{array} \right\},
  \end{multline*}
  where $\min (A)$ denotes a finite basis of the upwards-closed subset
  $A$. The remaining case is trivial:
  $\cf(\epsilon) = \upc_* \epsilon$.

  (Intuitively, the set ranges over all ways to cut $\bu$ in $k$
  consecutive pieces, and embeds all elements of the $j$-th piece into
  the same element $y_j$. It has long sequences, the longest being
  $\bu$, and shorter ones with potentially larger elements.)
\end{lemma}
This is the fully generic formula to describe the function $\cf$ for any
$X$. However, in simple cases, $\cf(\bw)$ takes a much simpler form.
For instance, for $X = \nat$, we have $\cf(x_1 \cdots x_n) = \upc_*
\max(x_1,\cdots,x_n)$, and for $X = \Sigma$ a finite alphabet,
$\cf(\bw) = \upc_* \bv$ where $\bv$ is the shortest
member of  the equivalence class  $[\bw]_{\sim_{st}}$ (that is, $\bv$
is obtained from $\bw$ by fusing consecutive equal letters).

\begin{proof}[Of \cref{stutter-cf}]
  The ``$\supseteq$'' direction is obvious.
  \\
  \noindent
  $(\subseteq)$ Given $\bw \gest x_1 \cdots x_n$, there exists a
  decomposition
  $\bw = \bw_0 y_1 \bw_1 y_2 \cdots y_k \bw_k$ for some $k \le n$,
  $y_1, \dots, y_k \in X$
  and $\bw_0, \dots, \bw_k \in X^*$, and there exists a monotonic mapping
  $p: [n] \rightarrow [k]$ such that $x_i \le y_{p(i)}$.
  For $j \in [k]$, define
  $i_j$ to be the largest $i$ such that $p(i) = j$ (i.e., the index of
  the right-most symbol of $x_1 \cdots x_n$ to be mapped to $y_j$),
  and let $i_0 = 0$. It
  follows that $0 = i_0 < i_1 < \dots < i_k = n$, and for all
  $\ell \in [n]$ and $j \in [k]$,
  $i_{j-1} < \ell \le i_{j} \implies x_{\ell} \le y_j$. Then
  $\bw \ge_* y_1 \cdots y_k$ which is indeed an element of the set
  described in the proposition.
  \qed
\end{proof}

 \subsection{Quotienting under a compatible equivalence}
\label{sec-quotient}
\label{sec-equiv}

In this subsection, we apply the results of \cref{sec-extension} to the most commonly encountered case of order-extension: quotient under an equivalence relation.

Let $(X, {\leq})$ be a \wqo and let $\eqE$ be an equivalence relation on
$X$ which is \emph{compatible} with $\leq$ in the sense that
$\leq \circ \eqE \:=\: \eqE \circ \leq$, where $\circ$ denotes the composition
of relations. Define the relation $\leq_E$ on $X$ to be
$\leq \circ \eqE$. Then $\leq_E$ is clearly reflexive, and is transitive
since
\[\leq_E \circ \leq_E \:=\: (\leq \circ \eqE) \circ (\leq \circ \eqE) \:=\: \leq \circ (\eqE \circ \leq) \circ \eqE \:=\: \leq \circ (\leq \circ \eqE) \circ \eqE \:=\: \leq \circ \eqE \:=\: \leq_E .\]
In this subsection, we give sufficient conditions for $(X,{\le_E})$ to be
ideally effective, provided $(X,{\leq})$ is.

\begin{remark}
Note that stuttering from \cref{sec-stutter} is \emph{not} an example:
Although ${\leqst} = {\leq_*} \circ {\simst}$,
the other condition does not hold: $\leqst \:\neq\: \simst \circ \leq_*$.
For instance, consider $X = \nat^2$ where $\tup{1,2}\tup{2,1} \leqst \tup{2,2}$.
However, if $X$ is a finite alphabet, the equality
$\leqst \:=\: \simst \circ \leq_*$ holds and
$(X^*, {\leqst})$ can be treated as a quotient.
As another example, the finitary powerset $\Pf(X)$ from \cref{sec-powerset} can be obtained as a quotient of
$(X^*, {\le_*})$, and could be shown ideally effective using
\cref{thm-quo-idl-eff} below. However, because operations in $\Pf(X)$ are quite simple,
and because powerset is a fundamental constructor, we decided to provide a
direct, more concrete, construction.
\qed
\end{remark}

Observe that $\leq_E$ is an extension of $\leq$, and thus results on
quotients can be seen as an application of \cref{sec-extension}.
However, since quotients are of such importance in computer science
(and used more often than mere extensions), we reformulate
\cref{thm-ext-idl-eff} in this specific context: functions $\ci$ and
$\cf$ take an interesting form.  As in the case of extensions,
elements and ideals of $(X,{\leq_E})$ will be represented using the
data structures coming from a presentation of $(X,{\leq})$.
\begin{theorem}
  \label{thm-quo-idl-eff}
  Let $(X,{\leq})$ be an ideally effective \wqo and $\eqE$ be an
  equivalence relation on $X$ compatible with $\le$.  Then,
  $(X,{\le_E})$ is ideally effective for the aforementioned data
  structures of $X$ and $\Idl(X,{\le_E})$, whenever the following
  functions are computable:
  \begin{xalignat*}{2}
     \ci:&
     \begin{array}{rl}
      \Idl(X, {\leq}) &\to \Down(X, {\leq}) \\
          I &\mapsto \overline{I}
    \end{array}
    &
   \cf:&
    \begin{array}{rl}
       \Fil(X, {\leq}) &\to \Up(X, {\leq}) \\
          \upc x &\mapsto \upc \overline{x}
    \end{array}
  \end{xalignat*}
  where, given $S \subseteq X$, $\overline S$ denotes the closure
  under $\eqE$ of $S$, i.e.,
  $\overline{S} \egdef \{ y \mid \exists x \in S: x\eqE y \}$, and
  $\overline{x}$ is a shortcut for $\overline{\{ x \}}$ which
  is the equivalence class of $x$.

Moreover, under these assumptions, we can compute a presentation of
$(X,{\le_E})$ from a presentation of $(X,{\leq})$.
\end{theorem}

\begin{proof}
  In the light of \cref{thm-ext-idl-eff}, it suffices to show
  $\upc_{\le_E} F = \overline{F}$ and $\dwc_{\le_E} I = \overline{I}$
  for any filter $F$ and any ideal $I$ of $(X,{\leq})$.  The first
  equality follows from $\le_E \:=\: \le \circ \eqE$ while the second comes
  from $\le_E \:=\: \eqE \circ \le$. This is why we introduced the
  compatibility condition $\le \circ \eqE \:=\: \eqE \circ \le$.
\qed
\end{proof}
In particular, we see that
the ideals of $(X, {\leq_E})$ are exactly the closures under $\eqE$ of the ideals of $(X, {\leq})$. That is,
$\Idl(X, {\leq_E}) = \bigl\{\overline{I} : I \in \Idl(X, {\leq}) \bigr\}$.

We conclude this section with two results that are specific to \wqos
obtained by quotienting, and which lead to simplifications in several
algorithms.
\begin{proposition}
  Let $J$ be an ideal under $\leq_E$, and let
  $J = I_1 \cup \cdots \cup I_k$ be the canonical ideal decomposition
  of $J$ under $\leq$. Then $J = \overline{I_i}$ for every $i$.
\end{proposition}
\begin{proof}
  Recall from the proof of \cref{ideals-under-extension} that $J =
  \overline{I_i}$ for some $i$. Without loss of generality, we can assume
  $i=1$.  For the sake of contradiction, suppose that there exists some $i$
  such that $J \neq \overline{I_i}$. Again without loss of generality, we
  can assume $i = k$. From $\overline{I_1} \neq \overline{I_k}$, we deduce
  that there exists $x \in I_1$ which has no $\eqE$-equivalent in $I_k$.

  We will now show that $J \subseteq I_1 \cup \cdots \cup I_{k-1}$,
  which will be a contradiction since we assumed that we started from
  a canonical ideal decomposition. Let $y \in J$. Then there exists
  a $y' \in I_1$ such that $y \eqE y'$. Since $I_1$ is an ideal under
  $\leq$, there is a $z \in I_1$ such that $x \leq z$ and $y' \leq
  z$. We have $y \eqE y' \leq z$, thus there exists $z'$ such that
  $y \leq z' \eqE z$. Since $J$ is closed under $\eqE$-equivalence,
  $z' \in J$, hence $z' \in I_i$ for some $i$. However, $z'$ cannot
  belong to $I_k$, since $x \leq_E z'$ and the $\eqE$-equivalence class
  of $x$ is disjoint from $I_k$. So
  $z' \in I_1 \cup \cdots \cup I_{k-1}$, and hence
  $y \in I_1 \cup \cdots \cup I_{k-1}$. Thus
  $J = I_1 \cup \cdots \cup I_{k-1}$, and we have a contradiction.
\qed
\end{proof}

\begin{proposition}
\label{prop-intersection-in-modE}
  For any two ideals $I_1,I_2\in\Idl(X,{\le})$,
  $\overline{I_1} \cap \overline{I_2} = \overline{I_1 \cap
    \overline{I_2}} =
  \overline{\overline{I_1} \cap I_2}$.\\
  For any two filters $F_1,F_2\in \Fil(X,{\le})$,
  $\overline{F_1} \cap \overline{F_2} = \overline{F_1 \cap
    \overline{F_2}} = \overline{\overline{F_1} \cap F_2}$.
\end{proposition}
\begin{proof}
We show $\overline{I_1} \cap \overline{I_2} = \overline{I_1 \cap
\overline{I_2}}$, the other equality is symmetric. For the right-to-left
inclusion, we have $I_1 \cap \overline{I_2} \subseteq \overline{I_1} \cap
\overline{I_2}$, and closing both sides under $\eqE$ gives the required
result. For the left-to-right inclusion, let $x \in \overline{I_1} \cap
\overline{I_2}$. Then there exist $x_1 \in I_1$ and $x_2 \in I_2$ such that
$x_1 \eqE x \eqE x_2$. Then $x_1 \in I_1 \cap \overline{I_2}$, and thus $x
\in \overline{I_1 \cap \overline{I_2}}$. 

The same proof applies to filters.
\qed
\end{proof}
Thanks to \cref{prop-intersection-in-modE}, we can compute intersections of
filters (resp., ideals) with only one invocation of $\cf$ (resp., $\ci$)
instead of the two invocations required by the algorithm described in the proof
of \cref{thm-ext-idl-eff}.

 \subsubsection{Sequences under conjugacy.}
\label{sec-conjugacy}

Consider a \wqo $(X, {\leq})$, and define an equivalence relation
$\simcj$ on $X^*$ as follows: $\bw \simcj \bv$ iff there exist
$\bm{x},\bm{y}$ such that $\bw=\bm{x}\bm{y}$ and $\bv=\bm{y}\bm{x}$.
One can imagine an equivalence class of $\simcj$ as a sequence written
on an (oriented) circle instead of a line.
We can now define a notion of \emph{subwords under conjugacy} via
$\leqcj \: \egdef \: \simcj \circ \leq_*$, which is exactly the
relation denoted $\preceq_c$ in~\cite[p.~49]{abdulla-forward-lcs}.

Since $\simcj$ is compatible with $\leq_*$, that is
${\leq_*} \circ {\simcj} \: = \: {\simcj} \circ {\leq_*}$, our results
over quotients apply to $(X^*,{\leqcj})$.
\begin{theorem}
  Sequence extension with conjugacy is an ideally effective constructor.
\end{theorem}

\begin{proof}
  Note that the data structures used for elements and ideals of
  $(X^*, {\leqcj})$ are obtained from data structures for $(X^*,{\le_*})$
  as done with \cref{thm-quo-idl-eff}.

  In the light of \cref{thm-quo-idl-eff}, it suffices to show that we
  can compute closures under $\simcj$ of elements and ideals of
  $(X^*,{\le_*})$.  Given $\bw \in X^*$, the equivalence class of
  $\bw$ under $\simcj$ is equal to
  $ \overline{\bw} = \{ c^{(i)}(\bw) \mid 0 \le i < \max(1,\size{\bw})
  \}$, where $c^{(i)}$ denotes the $i$-th iterate of the cycle
  operator $c(w_1\cdots w_n) = w_2 \cdots w_n w_1$, which corresponds
  to rotating the sequence $i$ times.  This expression is obviously
  computable.

  Computing the closure under $\simcj$ of ideals is quite similar.
  Remember that ideals of $(X^*, {\le_*})$ are products of atoms, where
  atoms are either of the form $D^*$ for some
  $D \in \Down(X)$, or of the form $I+\epsilon$, for
  some $I \in \Idl(X)$. Then,
  given $\bP = \bA_0 \cdots \bA_{k-1}$ an ideal of $(X^*, {\le_*})$:
\begin{align*}
  \overline{\bP} = \bigcup_{i=0}^{k-1} c^{(i)}(\bP) \cdot e(\bA_i) \:,
\end{align*}
where $e(D^*) = D^*$ and $e(I+\epsilon) = \epsilon$.
The presence of the extra $e(A_i)$ in the above expression might
become clearer when considering a simple example as $\bP = \{a\}^*\{b\}^*$ where
$\overline{\bP}= \{a\}^*\{b\}^*\{a\}^*\cup \{b\}^*\{a\}^*\{b\}^*$.
Indeed, $abba\simcj baab\simcj aabb\in\bP$.
\qed
\end{proof}

 \subsubsection{Multisets under the embedding ordering.}
 \label{sec-multisets}

Given a \wqo $(X, {\leq})$, we consider the set $\multX$ of finite
multisets over $X$. Intuitively, multisets are sets where an element
might occur multiple times.
Formally, a multiset $M \in \multX$ is a function from $X$ to
$\nat$: $M(x)$ denotes the number of occurrences of $x$ in $M$.
The support of a multiset $M$ denoted $Supp(M)$ is the set
$\{x \in X \mid M(x) \neq 0 \}$. A multiset is said to be finite if its
support is.

A natural algorithmic representation for these objects are lists of
elements of $X$, keeping in mind that a permutation of a list
represents the same multiset.
Formally, this means that $\multX$ is the quotient of $X^*$ by the
equivalence relation $\sim$ defined by 
\begin{align*}
  \bu=u_1 \cdots u_n \sim \bv = v_1 \cdots v_m
 \equivdef n = m \land
  \exists \sigma \in S_n : u_i = v_{\sigma(i)}\text{ for all $i=1,\ldots,n$,}
\end{align*}
where $S_n$ denotes the group of permutations over $\{1,\cdots,n\}$.

Once again, the equivalence relation $\sim$ is compatible with
$\le_*$. We denote by $\lemb$ the composition
${\sim} \circ {\le_*} = {\le_*} \circ {\sim}$, often called multiset
embedding.
(There exist other classical quasi-orderings on finite multisets, such
as the domination quasi-ordering, aka the Dershowitz-Manna
quasi-ordering~\cite{dershowitz79}: see~\cite[Theorem 7.2.3]{halfon-thesis} for
 a proof that it is an ideally  effective constructor.)
For this section, we focus on $(\multX, {\lemb})$,
which is an application of our results on quotients.

\begin{theorem}
  \label{thm-mult-idl-eff}
  Finite multisets with multiset embedding is an ideally effective
  constructor.
\end{theorem}

\begin{proof}[Sketch]
    Note that the data structures used for elements and ideals of
  $(\multX, {\lemb})$ are obtained from data structures for $(X^*,{\le_*})$
  as done with \cref{thm-quo-idl-eff}.

  In the light of \cref{thm-quo-idl-eff}, it suffices to show that we
  can compute closures under $\sim$ of elements and ideals of $(X^*,{\le_*})$.
  Given $\bw = x_1 \cdots x_n \in X^*$, the equivalence class of $\bw$ under $\sim$
  simply consists of all the possible permutations of the word $\bw$:
  \[
  \overline{\bw} = \bigcup_{\sigma \in S_n} x_{\sigma(1)} \cdots
  x_{\sigma(n)} \:.
  \]

  Closures of ideals are a little more complex. Let
  $\bP = \bA_1 \cdots \bA_n$ be an ideal of $(X^*, {\le_*})$.
  Define:
  \[
  D \egdef \bigcup \bigl\{ E \in \Down(X) \bigm\vert \exists i \in \{1,\cdots,n\}: E^* = \bA_i \bigr\}.
  \]
  In other words, $D \in \Down(X)$ is obtained from $\bP$ by picking
  the atoms $\bA_i$ that are of the second kind, $\bA_i=E^*$, and
  taking the union of their generators.  Similarly, let
  $I_1, \dots, I_p$ be the ideals of $(X,{\leq})$ that appear as
  $I_i + \epsilon$ in $\bP$, with repetitions, and in order of
  occurrence.  Then:
  \[
  \overline{\bP} = \bigcup_{\sigma \in S_p} D^* I_{\sigma(1)} D^* \cdots
  D^* I_{\sigma(p)} D^*
\:.
  \]  \qed
\end{proof}

 \subsection{Induced \wqos}
\label{sec-induced}

Let $(X,{\leq})$ be a \wqo. A subset $Y$ of $X$ (not necessarily finite) induces a
quasi-ordering $(Y,{\leq})$ which is also \wqo.

Any subset $S\subseteq X$ induces a subset $Y\cap S$ in $Y$.
Obviously, if $S$ is upwards-closed
(or downwards-closed) in $X$, then it induces an upwards-closed
(resp., downwards-closed) subset in $Y$.
However an ideal $I$ or a filter $F$ in $X$ does not always induce an
ideal or a filter in $Y$.
In the other direction though, if $J\in\Idl(Y)$, the downward closure
$\dc_X J$ is an ideal of $X$. Therefore, to describe the ideals of
$Y$, we need to identify those ideals of $X$ that are of the form
$\dwc_X J$ for some ideal $J$ of $Y$. This is captured by the
following notion:
\begin{definition}
  Given a \wqo $(X,{\leq})$ and a subset $Y$ of $X$, we say that an
  ideal $I \in \Idl(X)$ is \emph{in the adherence} of $Y$ if
  $I=\dc_X(I\cap Y)$.
\end{definition}

In particular this implies that $I\subseteq\dc_X Y$ (we say that $I$
is ``below $Y$'') and $I\cap Y\not=\emptyset$ (we say that $I$ is
``crossing $Y$''). The converse implication does not hold, as
witnessed by $X = \nat$, $Y = [1,3] \cup [5,7]$ and $I = \dwc 4$.

We now show that the ideals of $Y$ are exactly the subsets induced by
ideals of $X$ that are in the adherence of $Y$.
\begin{theorem}
  \label{lem-compat-ideals}
  Let $(X,{\leq})$ be a \wqo and $Y$ be a subset of $X$.  A subset $J$
  of $Y$ is an ideal of $Y$ if and only if $J=I\cap Y$ for some
  $I\in\Idl(X)$ in the adherence of $Y$.  In this case,
  $I = \dwc_X J$, and is thus uniquely determined from $J$.
\end{theorem}

\begin{proof}
  $(\GaD)$ : If $J\in\Idl(Y)$ then $I\egdef\dc_X J$ is directed hence
  is an ideal of $X$.  Clearly, $J = I \cap Y$, so $I$ is in the
  adherence of $Y$.
  \\
  $(\DaG)$: If $I\in\Idl(X)$ is in the adherence of $Y$ then
  $J\egdef I\cap Y$ is non-empty (since $I$ is crossing $Y$) and it is
  directed since for any $x,y\in J$ there is $z\in I$ above $x$ and
  $y$, and $z\leq z'$ for some $z'\in J$ since $I$ is below $Y$.

Uniqueness is clear since the compatibility assumption ``$I=\dc_X(I\cap Y)$''
completely determines $I$ from the ideal $J=I\cap Y$ it induces.
\qed
\end{proof}

An earlier definition of \emph{adherence} can be found in the
literature:
an ideal $I \in \Idl(X)$ is in the adherence of $Y$ if and only if
there exists a directed subset $\Delta \subseteq Y$ such that
$I = \dwc_X \Delta$~\cite{leroux2015}.  The two definitions are
equivalent~\cite[Lemma~14]{goubault2016}, so that, notably,
\cref{lem-compat-ideals} extends Lemma~4.6
from~\cite{leroux2015}.
\begin{proof}[that the two notions of adherence coincide]~\\
  $(\GaD)$ : Assume $I = \dwc_X (I \cap Y)$. We show that $\Delta = I
\cap Y$ is directed: let $x,y \in \Delta \subseteq I$, since $I$ is
directed, there exists $z \in I$ such that $z \ge x,y$. But since $I =
\dwc_X \Delta$, there exists $z' \in \Delta$ such that $z' \ge z \ge
x,y$, which proves that $\Delta$ is directed.
\\
\noindent
  $(\DaG)$: Assume that there exists a directed subset $\Delta
\subseteq Y$ such that $I = \dwc_X \Delta$. Then $\dwc_X (I \cap Y) =
\dwc_X (\dwc_X \Delta \cap Y) = \dwc_X (\Delta \cap Y) = \dwc_X \Delta
= I$.
\qed
\end{proof}

Similarly, we can define a notion of adherence for filters. However,
in this case, the condition $F =\upc_X(F\cap Y)$ simplifies: writing
$F$ as $\upc_X x$, this means that $x' \equiv_X x$ for some
$x' \in Y$, in which case $F = \upc_X x'$.
This is not surprising: $(Y,{\leq})$ is a \wqo, hence all its filters
are principal.

Assuming that $(X,{\leq})$ is an ideally effective \wqo, and given
$Y \subseteq X$, we can simply represent elements of $Y$ by
restricting the data structure for $X$ to $Y$. This requires that $Y$
be a recursive set. Alternatively,
\cref{lem-compat-ideals} suggests that we represent ideals of $Y$
as ideals of $X$ that are in the adherence of $Y$. This requires that
we can decide membership in the adherence of $Y$.
As in the case of extensions, the ideal effectiveness of $(Y,{\leq})$
does not always follow from the ideal effectiveness of $(X,{\leq})$
(see~\cite[Section 8.4]{halfon-thesis} for an example).
We therefore have to introduce extra assumptions.
\begin{theorem}
  \label{thm-sub-idl-eff}
  Let $(X,{\leq})$ be a \wqo and $Y \subseteq X$.
  Then $(Y,{\leq})$ is ideally effective (for the aforementioned
  representations) provided:
  \begin{itemize}
    \item membership in $Y$  is decidable over (the representation
      for) $X$,
    \item the following functions are computable:
\begin{xalignat*}{2}
  \si:&
  \begin{array}{rl}
 \Idl(X, {\leq}) &\to \Down(X, {\leq}) \\
  I &\mapsto \dwc_X  (I \cap Y)
  \end{array}
&
 \sff:&
 \begin{array}{rl}
 \Fil(X, {\leq}) &\to \Up(X, {\leq}) \\
  F &\mapsto \upc_X (F \cap Y)
  \end{array}
\end{xalignat*}
  \end{itemize}
Moreover, in this case, a presentation of $(Y,{\leq})$ can be computed
from a presentation of $(X,{\leq})$.
  \end{theorem}
The rest of this subsection  is dedicated to the proof of this
theorem.

First, let us mention that our first assumption implies that we have a
data structure for elements of $Y$ and that thanks to function $\si$,
we can decide whether an ideal $I$ of $X$ is in the adherence of $Y$: it
suffices to check that $\si(I) = I$.

Let us prove that $(Y,{\leq})$ is ideally effective.
  \begin{description}
    \item[\OD:] since $\le$ is decidable on $X$, its restriction to
      $Y$ is still decidable.
    \item[\ID:] Given two ideals $I_1, I_2$ that are in the adherence
      of $Y$,
      $I_1 \cap Y \subseteq I_2 \cap  Y \iff I_1 \subseteq I_2$.
      The left-to-right implication uses that
      $I_i = \dwc_X (I_i \cap Y)$.
      Therefore, inclusion for ideals of $Y$ can be implemented by relying
on \ID for $X$.
    \item[\PI:] if $y \in Y$, then $\dwc_X y$ is adherent to $Y$ and
      one relies on
      $\dwc_Y y = \dwc_X y \cap Y $.
  \end{description}

  For the four remaining operations, we need to be able to compute a
  representation of $D \cap Y$ and $U \cap Y$ for $D \in \Down(X)$ and
  $U \in \Up(X)$.

\begin{lemma}
  \label{lem:induced-dwc}
  Let $D \in\Down(X)$. The canonical representation of $D \cap Y$ (as
  a downwards-closed set of $Y$) is exactly the canonical
  representation of $\dwc_X (D \cap Y)$ (as a downwards-closed set of
  $X$).
\end{lemma}

\begin{proof}
  Let $\bigcup_i I_i$ be the \emph{canonical} decomposition of
  $\dwc_X (D \cap Y)$. Remember that an ideal $J$ of $Y$ is
  represented by the unique ideal $I$ of $X$ which is in the adherence
  of $Y$ such that $J = I \cap Y$. Thus, stating that $\bigcup_i I_i$
  is the canonical representation of $D \cap Y$ means that:
  \begin{enumerate}
  \item $D\cap Y = \bigcup_{i} (I_i\cap Y)$;
  \item for every $i$, $I_i \cap Y$ is an ideal of $Y$;
  \item $I_i \cap Y$ and $I_j \cap Y$ are incomparable for inclusion, for $i \neq j$.
  \end{enumerate}

  For the first point, $\bigcup_i (I_i \cap Y) =
  (\bigcup_i I_i) \cap Y = (\dwc_X (D \cap Y)) \cap Y = D \cap  Y$.

  We now argue that each $I_i \cap Y$ is indeed an ideal of $Y$, i.e.,
  all $I_i$'s are in the adherence of $Y$. One inclusion being
  trivial, we need to show that $I_i \subseteq \dwc_X (I_i \cap Y)$,
  for any $i$. Let $x_i \in I_i$.  Since the ideals $I_j$ are
  incomparable for inclusion, there exists $x'_i \in I_i$ such that
  $x_i \le x'_i$ and for any $j \neq i$, $x'_i \notin I_j$ ($I_i$ is
  directed). Besides, $x'_i \in I_i \subseteq \dwc_X (D \cap Y)$ and
  thus there is an element $x''_i$ such that
  $x'_i \le x''_i \in D \cap Y$.  As the sets $I_j$ are
  downwards-closed, $x''_i$ cannot belong to any $I_j$ with
  $j \neq 0$, hence $x''_i$ is in $I_i \cap Y$.  Therefore,
  $x_i \in \dwc_X (I_i \cap Y)$.

  Finally, the ideal decomposition $D \cap Y = \bigcup_j (I_j \cap Y)$
  is canonical since the $I_j$'s are incomparable in $X$ (recall the
  above criterion for inclusion of ideals of $Y$).
  \qed
\end{proof}

Observe that if $D = \bigcup_i I_i$ then
$\dwc_X (D \cap Y) = \bigcup_i \dwc_X (I_i \cap Y) = \bigcup_i
\si(I)$. Thus the canonical representation of $D \cap Y$ is indeed
computable from $D \in \Down(X)$.

We now present the dual of the previous lemma:
\begin{lemma}
  \label{lem:induced-upc}
  Given $U \in \Up(X)$, a canonical representation of $U \cap Y$ (as
  an upwards-closed set of $Y$) can be computed from a canonical
  representation of $\upc_X (U \cap Y)$ (as an upwards-closed set of
  $X$).
\end{lemma}

\begin{proof}
  Let $\bigcup_i \upc x_i$ be a canonical filter decomposition (in
  $X$) of the upwards-closed set $\upc_X (U \cap Y)$. We first prove
  that for every $i$, $x_i$ is equivalent to some element of
  $Y$. Indeed, since $\upc_X x_i \subseteq \upc_X (U \cap Y)$, there
  exists $y \in U \cap Y$ with $y \le x_i$. But then, $y$ must be in
  some $\upc_X x_j$. Since the decomposition is canonical, the $x_j$'s
  are incomparable, hence we cannot have $x_j \le y \le x_i$ for
  $j \neq i$. Thus, $x_i \equiv y \in Y$.

  Moreover, we can compute a canonical filter decomposition of
  $\upc_X (U \cap Y)$ using only elements in $Y$: for each $x_i$, it
  is decidable whether $x_i \in Y$ (our first assumption on $Y$). If
  not, we can enumerate elements of $Y$ until we find some
  $y_i \equiv x_i$. Such an element exists, and thus the enumeration
  terminates.

  We thus obtain a canonical filter decomposition $\bigcup_i \upc y_i$
  of $\upc_X (U \cap Y)$ with $y_i \in Y$. The rest of the proof is
  similar to the proof of \cref{lem:induced-dwc}.
\qed
\end{proof}

Here also, a canonical representation of $\upc_X (U \cap Y)$ is
computable from $U$, using the function $\sff$.

We can now describe procedures for the four remaining operations:
\begin{description}
\item[\CF:] Given $y \in Y$, the complement of $\upc_Y y$ is computed
  by using
  $Y \setminus \upc_Y y = (X \setminus \upc_X y) \cap Y$. Here the
  downwards-closed set $(X \setminus \upc_X y)$ is computable using
  \CF for $X$, and its intersection with $Y$ is computable using
  \cref{lem:induced-dwc}.

\item[\II:] Given two ideals $I$ and $I'$ in the adherence of $Y$, the
  intersection of the ideals they induce is
  $(I\cap Y)\cap (I'\cap Y) = (I \cap I') \cap Y$, which is computed
  using \II for $X$ and \cref{lem:induced-dwc}.

\item[\IF:] Computing the intersection of filters is similar to
  computing the intersection of ideals: given $y_1, y_2 \in Y$,
  $(\upc_Y y_1) \cap (\upc_Y y_2) = (\upc_X y_1 \cap \upc_X y_2) \cap
  Y$, which is computed using \IF for $X$ and
  \cref{lem:induced-upc}.

\item[\CI:] Given an ideal $I$ in the adherence of $Y$,
  $Y \setminus (I \cap Y) = (X \setminus I) \cap Y$, which is
  computed using \CI for $X$ and \cref{lem:induced-upc}.
\end{description}

Finally, and as always, the above presentation can be computed from a
presentation of $(X,{\leq})$, thanks to the functions $\si$ and
$\sff$. Notably, the ideal decomposition of $Y$ can be computed with
\cref{lem:induced-dwc} as the set induced by $X$, seen as a
downwards-closed subset, while the filter decomposition of $Y$ can be
computed using \cref{lem:induced-upc}, again as the set induced by $X$
seen this time as an upwards-closed subset.

\begin{remark}
  If $Y$ is a downwards-closed subset of $X$, then $I$ is adherent to $Y$
  if and only if $I \subseteq Y$, and therefore
  $\Idl(Y) = \Idl(X) \cap \Pcal(Y)$. Moreover, $\si$ is computable
  thanks to \II, and $\sff(\upc x) = \upc x$ if $x \in Y$,
  $\sff(\upc x) = \emptyset$ otherwise. Indeed, if $x \notin Y$, then
  $\upc x \cap Y = \emptyset$.

  Similarly, if $Y$ is upwards-closed, $\sff$ can be computed with \II, and
  $\si(I) = I$ if $Y \cap I \neq \emptyset$, $\si(I) = \emptyset$
  otherwise. Again, $Y \cap I \neq \emptyset$ if and only if
  $\exists x \in \min(Y):\ x \in I$. Given such an $x$, then
  $\forall y \in I: \exists z \in I: z \ge x,y$ by
  directedness. Therefore,
  $I \subseteq \dwc (I \cap \upc x) \subseteq \dwc (I \cap Y)$.
\qed
\end{remark}

\section{Towards a richer theory of ideally effective \wqos}
\label{sec-richer?}
\subsection{A minimal definition}
\label{sec-axioms}

As we mentioned in the remarks following \cref{def:eff-wqo}, our
definition contains redundancies: some of the requirements are implied
by the others. Here is the same definition in which we removed
redundancies:
\begin{definition}[Simply effective \wqos]
A \wqo $(X,{\leq})$ further equipped with data structures for $X$ and
$\Idl(X)$ is \emph{simply effective} if:
\begin{description}
\item[(\textrm{ID})]
  ideal inclusion $\subseteq$ is decidable on $\Idl(X)$;
\item[(\textrm{PI})]
  principal ideals are computable, that is, $x\mapsto\dc x$ is computable;
\item[(\textrm{CF})]
  complementation of filters, denoted $\neg:\Fil(X)\to\Down(X)$, is
computable;
\item[(\textrm{II})]
intersection of ideals, denoted $\cap:\Idl(X)\times\Idl(X)\to\Down(X)$, is
computable.
\end{description}
A \emph{short presentation} of $(X,{\leq})$ is a list of: data structures
for $X$ and $\Idl(X)$, procedures for the above operations, the ideal
decomposition of $X$.
\end{definition}

Note that a short presentation of $(X,{\leq})$ is obtained from a
presentation of $(X,{\leq})$
by dropping procedures for $\OD, \CI, \IF$ and by dropping \XF.
Surprisingly, short presentations carry enough
information:
\begin{theorem}
  \label{thm-vj}
  There exists an algorithm that given a short presentation of
  $(X,{\leq})$ outputs a presentation of $(X,{\leq})$.
\end{theorem}

\begin{corollary}
  A \wqo $(X,{\leq})$ (with data structures for $X$ and $\Idl(X)$)
  is ideally effective if and only if it is simply effective.
\end{corollary}

Before we proceed to proving \cref{thm-vj}, why did we bother to
display full presentations of \wqos in previous sections?  Our proofs
of ideal effectiveness would indeed have been shorter.

Our choice is motivated by practical reasons: the algorithms we have
given until now are much more efficient than the ones deduced from
\cref{thm-vj}, which is simply impractical.  (Indeed, most of these
algorithms have been implemented, at the highest level of generality,
by the second author.)  \cref{thm-vj} is more conceptual, and if one
only needs computability results, then \cref{thm-vj} provides a
simpler path to this goal.

As practice goes, we will refine the notion of ideally effective \wqos
to ``efficient'' ideally effective \wqos in \cref{sec-complexity}.
Most of the \wqos we have seen earlier are efficient in that sense.
By contrast, the presentation of $(X,{\leq})$ built from
\cref{sec-complexity} is not \emph{polynomial-time} (see
\cref{sec-complexity} for a definition).

\begin{proof}[of \cref{thm-vj}]
  We explain how to obtain the missing procedures:
  \begin{description}

  \item[\OD:] Given $x,y \in X$,
    $x \le y \iff \dwc x \subseteq \dwc y$. The latter can be tested
    using \PI and \ID.

\item[\CI:] We show a stronger statement, denoted \CD, that complementing an arbitrary downwards-closed set is computable. This strengthening is necessary for \IF.

  Let $D$ be an arbitrary downwards-closed set. We compute $\neg D$ as follows: \\
  \begin{enumerate}
  \item\label{step:init} Initialize $U := \emptyset$;
  \item\label{step:while} While $\neg U \not\subseteq D$ do
    \begin{enumerate}
    \item\label{step1} pick some $x \in \neg U \cap \neg D$;
    \item\label{step2} set $U := U \cup \upc x$.
    \end{enumerate}
  \end{enumerate}
  Every step of this high-level algorithm is effective. The
  complement $\neg U$ is computed using the description above: $\neg
  \bigcup_{i=1}^n \upc x_i = \bigcap_{i=1}^n \neg \upc x_i$ which is
  computed with \CF and \II (or with \XI in case $n=0$, i.e., for $U =
  \emptyset$). Then, inclusion $\neg U \subseteq D$ is tested with
  \ID. If this test fails, then we know $\neg U \cap \neg D$ is not
  empty, and thus we can enumerate elements $x \in X$ by brute force,
  and test membership in $U$ and in $D$. Eventually, we will find some
  $x \in \neg U \cap \neg D$.

  To prove partial correctness we use the following loop invariant:
  $U$ is upwards-closed and $U \subseteq \neg D$. The invariant holds
  at initialization and is preserved by the loop's body since if $\upc
  x$ is upwards-closed and since $x \notin D$ and $D$ downwards-closed
  imply $\upc X\subseteq\neg D$.  Thus when/if the loop terminates,
  one has both $\neg U \subseteq D$ and the invariant $U \subseteq
  \neg D$, i.e., $U = \neg D$.

  Finally, the algorithm terminates since it builds a strictly
  increasing sequence of upwards-closed sets, which must be finite by
  \cref{lem-no-chains}.

\item[\IF:] This follows from \CF and \CD, by expressing intersection in terms of complement and union.
  \end{description}

Lastly, we need to show that we can retrieve the filter decomposition
of $X$. It suffices to use \CD to compute $X = \neg \emptyset$.
  \qed
\end{proof}

The algorithm for \CD computes an upwards-closed set $U$ from an
oracle answering queries of the form ``Is $U \cap I$ empty?'' for
ideals $I$.  It is an instance of the generalized Valk-Jantzen Lemma
\cite{VJGL09}, an important tool for showing that some upwards-closed
sets are computable. This algorithm was originally developed by Valk
and Jantzen~\cite{valk85} in the specific case of
$(\nat^k, {\le_{\times}})$.

As seen in the above proof, the fact that \ID, \CF, \II and \PI entail
\CI is non-trivial. The existence of such a non-trivial redundancy in
our definition raises
the question of whether there are other hidden redundancies. The
following theorem answers the question in the negative.
\begin{theorem}
  \label{thm-min-ax}
  For each operation $A$ among \ID, \CF, \II and \PI, there exists a
  \wqo $(X_A, {\le_A})$ equipped with data structures for $X$ and
  $\Idl(X)$ for which operation $A$ is not computable, while the
  other three are.
\end{theorem}

This theorem means that short presentations are the shortest possible
to capture the information we want. Technically, we should also argue
that the ideal decomposition of $X$ cannot be retrieved from
procedures for operations \ID, \CF, \II, \PI.

For a full proof of \Cref{thm-min-ax}, we refer the interested reader
to~\cite[Proposition 8.1.4]{halfon-thesis}. Here we only illustrate the techniques at
hand by dealing with one case.

\begin{example}
\label{ex:X-CF}
For $n\in\nat$ we write $T_n$ for the halting time of $M_n$, the $n$-th Turing
machine (in some fixed recursive enumeration), letting $T_n=\infty$ if $T_n$ does not
halt.

Let now $X_{\CFbare}=\nat^2$ and define an equivalence relation $E$ over $X_{\CFbare}$ by
\[
\tup{n,m} E \tup{n',m'}
\equivdef
n=n' \text{ and $\bigl( T_n <\min(m,m')$ or
$T_n\geq\max(m,m') \bigr)$}
\]

One easily checks that $E$ is compatible with the lexicographic ordering on
$\nat^2$ in the sense of \Cref{sec-quotient}, and we consider
the \wqo $(X_{\CFbare},{\leq_{\CFbare}})$ with $\leq_{\CFbare}\egdef E\circ {\lex}$.
Regarding implementation, we use pairs of natural
numbers to represent elements of $X_{\CFbare}$, as well as the
corresponding principal ideals. We also use a special symbol to
represent the only non-principal ideal: $X_{\CFbare}$ itself.

With this representation, $(X_{\CFbare},\leq_{\CFbare})$ is almost ideally
effective:
deciding whether $\tup{n,m}\leq_{\CFbare} \tup{n,m'}$ only requires  simulating $M_n$ for $\max(m,m')$
steps \OD;
ideal inclusion reduces to comparing elements \ID;
creating $\dc x$ from $x$ is trivial \PI; as is
representing $X_{\CFbare}$ itself as a sum of ideals \XI.

However, $X_{\CFbare}$ with the chosen representation does not admit
an effective way of computing the complement of filters \CF: indeed
the complement of some $\upc_X\tup{n+1,0}$ must be some $\dc \tup{n,m}$ with
$m>T_n$ if $M_n$ halts (any $m$ is correct if $M_n$ does not
halt). Thus a procedure for \CF could be used to decide the halting
problem, which is impossible.
\qed
\end{example}

\begin{remark}[On ideally effective extensions]
$(X_{\CFbare},\leq_{\CFbare})$ is obtained as an extension of
$(\nat^2,{\lex})$, an ideally effective \wqo.
This proves that extensions of ideally effective \wqos are not always
ideally effective, even in the special case of a quotient by an
effective compatible equivalence, and justifies the two extra
assumptions we used in \cref{thm-ext-idl-eff}. More precisely, it
justifies that at least one of these assumptions is necessary, and
indeed, one can always compute the closure function $\cf$ from the
closure function $\ci$ (but not the converse!), and this in a uniform
manner. The latter result relies on an algorithm that is very similar
to the generalized Valk and Jantzen Lemma.
\qed
\end{remark}

 \subsection{On alternative effectiveness assumptions}
\label{sec-alternatives}

The set of effectiveness assumptions collected in \cref{def:eff-wqo}
or \cref{def-presentation} is motivated by the need to perform Boolean
operations on (downwards-, upwards-) closed subsets, as illustrated in
our motivating examples from \cref{sec-motivations}.  Other choices
are possible, and we illustrate a possible variant here.

\subsubsection{A natural but not ideally effective constructor.}

Given two \qos $(X,{\le_X})$ and $(Y,{\le_Y})$, we can define the
lexicographic quasi-ordering $\lex$ on $X \times Y$ by:
\[ \tup{x_1, y_1} \lex \tup{x_2, y_2} \equivdef x_1 <_X x_2 \vee (x_1 \equiv_X x_2
\wedge y_1 \le_Y y_2),
\]
where classically, ${\equiv_X}$ denotes the equivalence relation
${\le_X} \cap {\ge_X}$ and ${<_X}$ denotes the strict ordering
associated to $X$, defined as ${\le_X} \setminus {\equiv_X}$.

Since $\lex$ is coarser than the product ordering $\le_{\times}$ from
\cref{sec-products}, $(X\times Y,{\lex})$
is a \wqo as soon as $\le_X$ and $\le_Y$ are. Besides, when $(X,{\le_X})$
and $(Y,{\le_Y})$ are ordinals, the lexicographic product corresponds to
the ordinal multiplication $Y\cdot X$.

This \wqo is simple and natural, but it is not always ideally
effective in the sense of \cref{def:eff-wqo} (at least for the natural
representation of elements of $X \times Y$). The fact that our
definition misses such a simple \wqo is disturbing and will be
discussed in the next subsection.  For now, let us show why
lexicographic product is not an ideally effective constructor.

\begin{proposition}
  \label{prop:lex-not-idl-eff}
Lexicographic product is not an ideally effective constructor. In
particular, there exists an ideally effective \wqo $X_\PP$ such
that $(X_\PP \times A_2,{\lex})$ is not
ideally effective for any useful representation.
\end{proposition}

\begin{proof}
  Recall from \cref{sec-finite} that $A_2=\{a,b\}$ is the two-letter
  alphabet, where $a$ and $b$ are incomparable. We use the following
  property: Let $(X,\leq)$ be some \wqo and $I\in\Idl(X)$ one of its
  ideals. Then $I$ is principal if, and only if, $I\times A_2$ is not
  an ideal in the lexicographic product $(X \times A_2, {\lex})$.
  Indeed, if $I = \dwc x$ for some $x \in X_\PP$, then $\tup{x,a}$ and
  $\tup{x,b}$ do not have a common upper bound in $I \times A_2$ with
  respect to ${\lex}$, hence $I \times A_2$ is not
  directed. Conversely, if $I$ is not principal, then for any two
  elements $\tup{x,c}, \tup{y,d} \in I \times A_2$, there is some $z
  \in I$ such that $z > x$ and $z > y$. The element $\tup{z,a}$ is a
  suitable common upper bound, showing that $I \times A_2$ is directed.

  Regarding $X_\PP$, we refer to~\cite[Section 8.3]{halfon-thesis} and
  do not describe it here: it is an ideally effective \wqo, similar to
  $X_\CFbare$ from \cref{ex:X-CF}, and for which it is undecidable
  whether an ideal $I$ is principal.
  This is enough to prove that $(X_\PP \times A_2, {\lex})$ is not ideally
  effective. Assume, by way of contradiction, that it is ideally
  effective.  Then for any $I \in \Idl(X)$, one can compute the ideal
  decomposition of $D = I\times A_2$ and then see whether this
  downwards-closed set is an ideal. But deciding whether $D$ is an
  ideal amounts to deciding whether $I$ is not principal, which is
  impossible in $X_\PP$. 

  Note: the only representation assumption that the proof makes on
  $X_\PP\times A_2$ is that the pairing function $x,c\mapsto\tup{x,c}$
  is effective. With this assumption $I\times A_2$ can be built in the
  following manner: (1) compute $X_\PP\setminus I=\upc x_1+\cdots+\upc
  x_n$ in $X_\PP$; (2) derive $(X_\PP\setminus I)\times A_2 = U =
  \upc\tup{x_1,a}+\upc\tup{x_1,b}+\cdots+\upc\tup{x_n,b}$ using
  pairings; (3) obtain $I\times A_2$ by complementing $U$ in
  $X_\PP\times A_2$, assumed to be ideally effective.
\qed
\end{proof}

\subsubsection{Deciding principality.}

In the previous subsection, we have shown that a very natural
constructor, the lexicographic product, is not ideally effective.
However, in practice $(X \times Y, {\lex})$ is usually ideally effective, that
is, the lexicographic product of two ``actually used'' \wqos
$(X,{\le_X})$ and $(Y,{\le_Y})$ is ideally effective.

Thus, the problem seems to come from the fact that our definition
allows too many exotic \wqos. Indeed, we can show 
that the lexicographic product of two ideally
effective \wqos for which we can decide whether an ideal is principal,
is ideally effective~\cite[Theorem 5.4.2]{halfon-thesis}.
All \wqos used in practice trivially meet this
extra condition, to the point that we could argue that we should not
accept as ideally effective any \wqo that would not meet this requirement.

If, in the definition of ideally effective \wqos, one now adds the
condition that principality of ideals be decidable, then lexicographic
product becomes an ideally effective constructor, most of the
constructors described in this chapter remain ideally effective, to
the notable exception of extensions and quotients:
\cref{thm-ext-idl-eff} and \cref{thm-quo-idl-eff} fail with the new
definition (see~\cite[Section 8.3]{halfon-thesis} for details).

\subsubsection{Directions for future work.}

We would like to mention three directions in which our work can be
extended.

The first one was carried out in~\cite{FGL:partI}, relying on the
topological notion of notion called \emph{Noetherian space} to
generalize \wqos, in the following sense.  Given a quasi-ordered set
$(X,{\leq})$, the Alexandroff topology has as open sets exactly the
upwards-closed sets for the quasi-ordering $\le$.  It turns out that
the Alexandroff topology associated to $\le$ is Noetherian if and only
if $\le$ is a \wqo on $X$. There are also Noetherian topologies that
do not arise as Alexandroff topologies, for example the cofinite
topology on an infinite set, or the Zariski topology on a Noetherian
ring. One advantage of Noetherian spaces is that they are preserved
under more constructors than \wqos, e.g., the full powerset of a
Noetherian space (with the so-called lower Vietoris topology) is again
Noetherian. In~\cite{FGL:partI}, the authors define a notion of
effectiveness very similar to ours for Noetherian spaces, which
however excludes complements and filters, which do not make sense
there.  Similarly, this notion of effectiveness is preserved under
many constructors.

A second extension of this work was carried out
in~\cite[Chapter 9]{halfon-thesis}. The motivation is close to the one above:
handling more constructors.  As mentioned in \cref{sec-powerset}, the
infinite powerset $\Pcal(X)$ of a \wqo, ordered with the Hoare
ordering is not a \wqo in general. However, the class of \wqos for
which $(\Pcal(X), {\hoare})$ is a \wqo is well-known: these \wqos are
called $\omega^2$-\wqo (e.g.,
see~\cite{marcone2001,jancar99c}). The second
author~\cite{halfon-thesis} proposes a generalization of our notion of
ideal effective \wqos which he calls ideal effective $\omega^2$-\wqos
(also $\Idl^2$-effective \wqos). He then shows that the constructors
presented in this chapter also preserve this stronger notion of
$\Idl^2$-effectiveness, and also prove that, e.g., the powerset
of an $\Idl^2$-effective \wqo, ordered with the Hoare quasi-ordering, is
an ideally effective \wqo.  The notion of $\omega^2$-\wqo can be
generalized to the notion of $\alpha$-\wqo for any indecomposable
ordinal $\alpha$, eventually leading to the notion of \emph{better
  quasi-ordering} ($\alpha$-\wqo for every countable
$\alpha$). In~\cite{halfon-thesis}, the author raises the question on
how to generalize ideal effectiveness to these classes of
quasi-orderings.

Finally, one might challenge our own decision of representing upwards-
and downwards-closed sets as their filter/ideal decompositions. Its
main advantage is genericity: as proved in \cref{sec-basics}, this
decomposition is possible in any \wqo.  It is also very convenient.
In the simple cases of $(\nat^k,{\leq_{\times}})$ and
$(A^*,{\leq_*})$, the representations and algorithms we illustrated in
\cref{sec-motivations} have been used for years by researchers who
were not aware that they were manipulating ideals.  This suggests that
the idea is somehow natural.

This does not rule out the existence of better ad-hoc solutions when
considering a specific \wqo, notably in terms of efficiency.  As will
be seen in \cref{sec-complexity}, the procedures we have presented in
\cref{sec-sequences} have an exponential-time worst-case
complexity. This exponential blow-up essentially occurs when one has
to distribute the unions over the products in order to retrieve an
actual filter/ideal decomposition. We are not sure this can be
averted, but when one only needs to represent certain particular closed
subsets of $(X^*,{\le_*})$, better representations do exist: see for
instance~\cite{geffroy2017}.

 \subsection{On computational complexity}
\label{sec-complexity}

In~\cite{halfon-thesis}, the second author provides a complexity
analysis of the algorithms we have described in this chapter. Let us
briefly summarize the complexity of the \wqo constructors we have
considered.

Formally, let us define a \emph{polynomial-time} ideally effective
\wqo to be an ideally effective \wqo for which there exist
\emph{polynomial-time} procedures for \OD, \ID, \CF, \IF, \CI, \II,
\PI. A presentation of an ideally effective \wqo is said to be
\emph{polynomial-time} if all the procedures it is composed of run in
polynomial time. For instance, $\nat$ is a polynomial-time ideally
effective \wqo, and the presentation we gave for it is
polynomial-time. However, a \wqo as simple as $(A^*, {\le_*})$, where
$A = \{ a,b \}$, is not polynomial-time, at least for our choice of
data structure for $A^*$ and $\Idl(A^*)$. Indeed, observe that the
upwards-closed set $U_n = \upc a^n \cap \upc b^n$ has at least
exponentially many (in $n$) minimal elements: any word with $n$ $a$'s
and $n$ $b$'s is a minimal element of $U_n$. Therefore, the filter
decomposition of $U_n$ is of exponential size in $n$, and thus
requires exponential-time to compute.

However, for instance, the Cartesian product
$(X \times Y, {\le_{\times}})$ of polynomial-time ideally effective
\wqos is polynomial-time.  (That would fail if $X$ or $Y$ were not
polynomial-time: for instance, if $(X,{\le_X}) = (A^*, {\le_*})$, then
the upwards-closed set $\upc (a^n, y_1) \cap \upc (b^n, y_2)$ has at
least exponentially many minimal elements, independently of the filter
decomposition of $\upc y_1 \cap \upc y_2$.)  Furthermore, from
polynomial-time presentations of $(X, {\le_X})$ and $(Y, {\le_Y})$,
the presentation of $(X \times Y, {\le_{\times}})$ we compute in
\cref{sec-products} is polynomial-time as well. This motivates the
following definition: an ideally effective constructor $C$ is
\emph{polynomial-time} if it is possible to compute a polynomial-time
presentation for $C[(X_1, {\le_1}), \dots, (X_n,{\le_n})]$ given
polynomial-time presentations of
$(X_1, {\le_1}), \dots, (X_n,{\le_n})$.  Note that we require that the
procedures of the presentation for
$C[(X_1, {\le_1}), \dots, (X_n,{\le_n})]$ are polynomial-time, but we
do not make any assumption on the complexity of the procedure that
builds the new presentation from presentations for each
$(X_i, {\le_i})$.

With this definition in mind, here is a summary of the complexity results
from~\cite{halfon-thesis}:
\begin{itemize}
\item Both disjoint sum and lexicographic sum are polynomial-time
  ideally effective constructors---this is a trivial analysis of the
  presentation of \cref{sec-sums}.
\item Cartesian product is a polynomial-time ideally effective
  constructor; that again follows easily from an analysis of \cref{sec-products}.
\item Higman's sequence extension \qo is \emph{not} a polynomial-time
  ideally effective constructor. As we have seen above, already in the
  simple case of finite sequences over a finite alphabet, some
  operations require exponential time. It is not difficult to
  see that the presentation we gave in \cref{sec-sequences} consists
  of exponential time procedures.
\item The finite powerset constructor (under the Hoare quasi-ordering)
  is a polynomial-time ideally effective constructor.  This again
  follows from an easy analysis of \cref{sec-powerset}. This justifies
  implementing $\Pf(X)$ directly, and not as a quotient of $X^*$.
\item The finite multiset constructor, under multiset embedding, is an
  exponential-time ideally effective constructor, and already
  $(\mult{{\nat^{2}}}, {\lemb})$ is not a polynomial-time ideally
  effective \wqo. However, $(\mult{A}, {\lemb})$ and
  $(\mult{\nat}, {\lemb})$ are polynomial-time effective \wqos when
  $A$ is a finite alphabet under equality.
\end{itemize}

\section{Concluding remarks}
\label{sec-concl}

We have proposed a set of effectiveness assumptions that allow one to
compute with upwards-closed and downwards-closed subsets of \wqos,
represented as their canonical filter and ideal decompositions
respectively.  These effectiveness assumptions are fulfilled in the
main \wqos that appear in practical computer applications, which are
built using constructors that we have shown to be ideally
effective. Our algorithms unify and generalize some algorithms that
have been used for many years in simple settings, such as $\nat^k$ or
the set of finite words ordered by embedding.

We have not considered any \wqo constructor more complex than sequence
extension, and this is an obvious direction for extending this
work. How does one compute with closed subsets of finite labeled trees
ordered by Kruskal's homeomorphic embedding? Or of some class of
finite graphs well-quasi-ordered by some notion of embedding?  The
case of finite trees has already been partially tackled by the first
author, see~\cite{FGL:partI,goubault2016}.  The technicalities are
daunting, well beyond the ambitions of this chapter,
however.

\section*{Acknowledgments}
\label{sec-ack}

We thank Sylvain Schmitz for his many
suggestions and his constant support during the several years that it
took us to prepare this chapter. The final text also owes to
suggestions made at various times by D.\ Lugiez, A.\ Finkel,
D.\ Kuske and, finally, an anonymous reviewer.

\bibliographystyle{abbrv}
\bibliography{wqo}

\begin{thebibliography}{10}

\bibitem{AAC-fsttcs2012}
P.~A. Abdulla, M.~F. Atig, and J.~Cederberg.
\newblock Timed lossy channel systems.
\newblock In {\em FST\&TCS 2012, LIPIcs 18}, pages 374--386. Leibniz-Zentrum f{\"u}r Informatik, 2012.

\bibitem{abdulla2000c}
P.~A. Abdulla, K.~{\v{C}}er{\=a}ns, B.~Jonsson, and Y.-K. Tsay.
\newblock Algorithmic analysis of programs with well quasi-ordered domains.
\newblock {\em Information and Computation}, 160(1/2):109--127, 2000.

\bibitem{abdulla-forward-lcs}
P.~A. Abdulla, A.~{Collomb-Annichini}, A.~Bouajjani, and B.~Jonsson.
\newblock Using forward reachability analysis for verification of lossy channel systems.
\newblock {\em Formal Methods in System Design}, 25(1):39--65, 2004.

\bibitem{akama2011}
Y.~Akama.
\newblock Set systems: Order types, continuous nondeterministic deformations, and quasi-orders.
\newblock {\em Theoretical Comput. Sci.}, 412(45):6235--6251, 2011.

\bibitem{bachmeier2015}
G.~Bachmeier, M.~Luttenberger, and M.~Schlund.
\newblock Finite automata for the sub- and superword closure of {CFL}s: Descriptional and computational complexity.
\newblock In {\em LATA 2015, LNCS 8977}, pages 473--485. Springer, 2015.

\bibitem{BS-fmsd2013}
N.~Bertrand and {\relax Ph}.~Schnoebelen.
\newblock Computable fixpoints in well-structured symbolic model checking.
\newblock {\em Formal Methods in System Design}, 43(2):233--267, 2013.

\bibitem{Birkhoff:latt}
G.~Birkhoff.
\newblock {\em Lattice Theory}, volume XXV of {\em American Mathematical Society Colloquium Publications}.
\newblock American Mathematical Society, Providence, R.I., corrected reprint of the third edition (1967) edition, 1979.

\bibitem{blondin2017b}
M.~Blondin, A.~Finkel, and P.~McKenzie.
\newblock Well behaved transition systems.
\newblock {\em Logical Methods in Comp.\ Science}, 13(3):1--19, 2017.

\bibitem{blondin2018}
M.~Blondin, A.~Finkel, and P.~McKenzie.
\newblock Handling infinitely branching well-structured transition systems.
\newblock {\em Information and Computation}, 258:28--49, 2018.

\bibitem{bonnet75}
R.~Bonnet.
\newblock On the cardinality of the set of initial intervals of a partially ordered set.
\newblock In {\em Infinite and Finite Sets, Colloquia Math.\ Soc.\ J{\'a}nos Bolyai, Keszthely, Hungary, 1973}, pages 189--198. North-Holland, 1975.

\bibitem{chen58}
K.~T. Chen, R.~H. Fox, and R.~C. Lyndon.
\newblock Free differential calculus, {IV}. {T}he quotient groups of the lower central series.
\newblock {\em Annals of Mathematics}, 68(1):81--95, 1958.

\bibitem{dalessandro2008}
F.~D'Alessandro and S.~Varricchio.
\newblock Well quasi-orders in formal language theory.
\newblock In {\em DLT 2008, LNCS 5257}, pages 84--95. Springer, 2008.

\bibitem{dershowitz87}
N.~Dershowitz.
\newblock Termination of rewriting.
\newblock {\em Journal of Symbolic Computation}, 3(1--2):69--115, 1987.

\bibitem{dershowitz79}
N.~Dershowitz and Z.~Manna.
\newblock Proving termination with multiset orderings.
\newblock {\em Commun. ACM}, 22(8):465--476, 1979.

\bibitem{Deutsch:ipv}
L.~P. Deutsch.
\newblock An interactive program verifier.
\newblock Technical Report CSL-73-1, Xerox Palo Alto Research Center, 1973.

\bibitem{finkel2016}
A.~Finkel.
\newblock The ideal theory for {WSTS}.
\newblock In {\em RP 2016, LNCS 9899}, pages 1--22. Springer, 2016.

\bibitem{FWD-WSTS-1}
A.~Finkel and J.~Goubault{-}Larrecq.
\newblock Forward analysis for~{WSTS}, part~{I}: Completions.
\newblock In {\em STACS 2009, LIPIcs 3}, pages 433--444. Leibniz-Zentrum f{\"u}r Informatik, 2009.

\bibitem{FWD-WSTS-2}
A.~Finkel and J.~Goubault{-}Larrecq.
\newblock Forward analysis for~{WSTS}, part~{II}: Complete {WSTS}.
\newblock In {\em ICALP 2009, LNCS 5556}, pages 188--199. Springer, 2009.

\bibitem{FGL:partI}
A.~Finkel and J.~Goubault-Larrecq.
\newblock Forward analysis for {WSTS}, {P}art~{I}: Completions.
\newblock {\em Mathematical Structures in Computer Science}, 2019.
\newblock Submitted. Journal version of \cite{FWD-WSTS-1}.

\bibitem{finkel98b}
A.~Finkel and {\relax Ph}.~Schnoebelen.
\newblock Well-structured transition systems everywhere!
\newblock {\em Theoretical Comput. Sci.}, 256(1--2):63--92, 2001.

\bibitem{Gallier:Gamma0}
J.~H. Gallier.
\newblock What's so special about {K}ruskal's theorem and the ordinal {$\Gamma_0$}? {A} survey of some results in proof theory.
\newblock {\em Annals of Pure and Applied Logic}, 53(3):199--260, 1991.

\bibitem{geeraerts06}
G.~Geeraerts, J.-F. Raskin, and L.~Van{ }Begin.
\newblock Expand, enlarge and check: New algorithms for the coverability problem of {WSTS}.
\newblock {\em J.~Comput. Syst. Sci.}, 72(1):180--203, 2006.

\bibitem{geffroy2017}
T.~Geffroy, J.~Leroux, and G.~Sutre.
\newblock Backward coverability with pruning for lossy channel systems.
\newblock In {\em SPIN 2017}, pages 132--141. ACM Press, 2017.

\bibitem{Goto:HLISP}
E.~Goto.
\newblock Monocopy and associative algorithms in an extended {Lisp}.
\newblock Technical report, University of Tokyo, Japan, May 1974.

\bibitem{JG:ML}
J.~Goubault.
\newblock {HimML}: Standard {ML} with fast sets and maps.
\newblock In {\em 5th ACM SIGPLAN Workshop on ML and its Applications}, 1994.
\newblock doi:10.1.1.40.4967.

\bibitem{VJGL09}
J.~Goubault{-}Larrecq.
\newblock On a generalization of a result by {Valk} and {Jantzen}.
\newblock Research Report LSV-09-09, Laboratoire Sp{\'e}cification et V{\'e}rification, ENS Cachan, France, May 2009.

\bibitem{goubault2016}
J.~Goubault{-}Larrecq and S.~Schmitz.
\newblock Deciding piecewise testable separability for regular tree languages.
\newblock In {\em ICALP 2016, LIPIcs 55}, pages 97:1--97:15. Leibniz-Zentrum f{\"u}r Informatik, 2016.

\bibitem{HSS-lmcs}
{\relax Ch}.~Haase, S.~Schmitz, and {\relax Ph}.~Schnoebelen.
\newblock The power of priority channel systems.
\newblock {\em Logical Methods in Comp.\ Science}, 10(4:4), 2014.

\bibitem{habermehl2010}
P.~Habermehl, R.~Meyer, and H.~Wimmel.
\newblock The downward-closure of {Petri} net languages.
\newblock In {\em ICALP 2010, LNCS 6199}, pages 466--477. Springer, 2010.

\bibitem{HSS-lics2012}
S.~Haddad, S.~Schmitz, and {\relax Ph}.~Schnoebelen.
\newblock The ordinal-recursive complexity of timed-arc {Petri} nets, data nets, and other enriched nets.
\newblock In {\em LICS 2012}, pages 355--364. IEEE Comp.\ Soc.\ Press, 2012.

\bibitem{haines69}
L.~H. Haines.
\newblock On free monoids partially ordered by embedding.
\newblock {\em Journal of Combinatorial Theory}, 6(1):94--98, 1969.

\bibitem{halfon-thesis}
S.~Halfon.
\newblock {\em On Effective Representations of Well Quasi-Orderings}.
\newblock Th{\`e}se de doctorat, {\'E}cole Normale Sup{\'e}rieure Paris-Saclay, France, June 2018.

\bibitem{Higman:Lemma}
G.~Higman.
\newblock Ordering by divisibility in abstract algebras.
\newblock {\em Proceedings of the London Mathematical Society}, 2(7):326--336, 1952.

\bibitem{jancar99c}
P.~Jan{\v{c}}ar.
\newblock A note on well quasi-orderings for powersets.
\newblock {\em Information Processing Letters}, 72(5--6):155--161, 1999.

\bibitem{kabil92}
M.~Kabil and M.~Pouzet.
\newblock Une extension d'un th{\'e}or{\`e}me de {P}. {Jullien} sur les {\^a}ges de mots.
\newblock {\em RAIRO Theoretical Informatics and Applications}, 26(5):449--482, 1992.

\bibitem{karp69}
R.~M. Karp and R.~E. Miller.
\newblock Parallel program schemata.
\newblock {\em J.~Comput. Syst. Sci.}, 3(2):147--195, 1969.

\bibitem{kriz90}
I.~K{\v{r}}{\'{\i}}{\v{z}} and R.~Thomas.
\newblock On well-quasi-ordering finite structures with labels.
\newblock {\em Graphs and Combinatorics}, 6(1):41--49, 1990.

\bibitem{Kruskal:tree}
J.~B. Kruskal.
\newblock Well-quasi-ordering, the tree theorem, and {Vazsonyi}'s conjecture.
\newblock {\em Transactions of the American Mathematical Society}, 95(2):210--225, 1960.

\bibitem{kruskal72}
J.~B. Kruskal.
\newblock The theory of well-quasi-ordering: A frequently discovered concept.
\newblock {\em Journal of Combinatorial Theory, Series A}, 13(3):297--305, 1972.

\bibitem{lazic2008}
R.~Lazi{\'c}, T.~Newcomb, J.~Ouaknine, A.~W. Roscoe, and J.~Worrell.
\newblock Nets with tokens which carry data.
\newblock {\em Fundamenta Informaticae}, 88(3):251--274, 2008.

\bibitem{lazic2015}
R.~Lazi{\'c} and S.~Schmitz.
\newblock The ideal view on {Rackoff}'s coverability technique.
\newblock In {\em RP 2015, LNCS 9328}, pages 76--88. Springer, 2015.

\bibitem{lazic2016b}
R.~Lazi{\'c} and S.~Schmitz.
\newblock The complexity of coverability in {\(\nu\)}-{Petri} nets.
\newblock In {\em LICS 2016}, pages 467--476. ACM Press, 2016.

\bibitem{leroux2015}
J.~Leroux and S.~Schmitz.
\newblock Demystifying reachability in vector addition systems.
\newblock In {\em LICS 2015}, pages 56--67. IEEE Comp.\ Soc.\ Press, 2015.

\bibitem{leroux2016}
J.~Leroux and S.~Schmitz.
\newblock Ideal decompositions for vector addition systems.
\newblock In {\em STACS 2016, LIPIcs 47}, pages 1:1--1:13. Leibniz-Zentrum f{\"u}r Informatik, 2016.

\bibitem{leuschel2002}
M.~Leuschel.
\newblock Homeomorphic embedding for online termination of symbolic methods.
\newblock In {\em The Essence of Computation, Complexity, Analysis, Transformation, LNCS 2566}, pages 379--403. Springer, 2002.

\bibitem{lovasz2006}
L.~Lov\'asz.
\newblock Graph minor theory.
\newblock {\em Bull.\ Amer.\ Math.\ Soc.}, 43(1):75--86, 2006.

\bibitem{manolios05}
P.~Manolios and D.~Vroon.
\newblock Ordinal arithmetic: Algorithms and mechanization.
\newblock {\em Journal of Automated Reasoning}, 34(4):387--423, 2005.

\bibitem{marcone2001}
A.~Marcone.
\newblock Fine analysis of the quasi-orderings on the power set.
\newblock {\em Order}, 18(4):339--347, 2001.

\bibitem{pin96}
J.-{\'E}. Pin.
\newblock Logic, semigroups and automata on words.
\newblock {\em Annals of Mathematics and Artificial Intelligence}, 16(1):343--384, 1996.

\bibitem{rado54}
R.~Rado.
\newblock Partial well-ordering of sets of vectors.
\newblock {\em Mathematika}, 1(2):89--95, 1954.

\bibitem{RS:minor}
N.~Robertson and P.~D. Seymour.
\newblock Graph minors. {XX}. {W}agner's conjecture.
\newblock {\em Journal of Combinatorial Theory, Series B}, 92(2):325--357, 2004.

\bibitem{sacks90}
G.~E. Sacks.
\newblock Constructive ordinals and {$\Pi^1_1$} sets.
\newblock In {\em Higher Recursion Theory}, volume~2 of {\em Perspectives in Mathematical Logic}, chapter~1, pages 3--21. Springer, 1990.

\bibitem{sakarovitch83}
J.~Sakarovitch and I.~Simon.
\newblock Subwords.
\newblock In M.~Lothaire, editor, {\em Combinatorics on Words}, volume~17 of {\em Encyclopedia of Mathematics and Its Applications}, chapter~6, pages 105--142. Cambridge Univ.\ Press, 1983.

\bibitem{schmitz2014}
S.~Schmitz.
\newblock Complexity bounds for ordinal-based termination.
\newblock In {\em RP 2014, LNCS 8762}, pages 1--19. Springer, 2014.

\bibitem{schmitz-toct2016}
S.~Schmitz.
\newblock Complexity hierarchies beyond {Elementary}.
\newblock {\em ACM Trans.\ Computation Theory}, 8(1), 2016.

\bibitem{SS-icalp11}
S.~Schmitz and {\relax Ph}.~Schnoebelen.
\newblock Multiply-recursive upper bounds with {Higman}'s lemma.
\newblock In {\em ICALP 2011, LNCS 6756}, pages 441--452. Springer, 2011.

\bibitem{SS-esslli2012}
S.~Schmitz and {\relax Ph}.~Schnoebelen.
\newblock Algorithmic aspects of {WQO} theory.
\newblock Lecture notes, 2012.

\bibitem{SS-concur13}
S.~Schmitz and {\relax Ph}.~Schnoebelen.
\newblock The power of well-structured systems.
\newblock In {\em CONCUR 2013, LNCS 8052}, pages 5--24. Springer, 2013.

\bibitem{valk85}
R.~Valk and M.~Jantzen.
\newblock The residue of vector sets with applications to decidability problems in {Petri} nets.
\newblock {\em Acta Informatica}, 21(6):643--674, 1985.

\bibitem{zetzsche2015}
G.~Zetzsche.
\newblock Computing downward closures for stacked counter automata.
\newblock In {\em STACS 2015, LIPIcs 30}, pages 743--756. Leibniz-Zentrum f{\"u}r Informatik, 2015.

\bibitem{zetzsche2018}
G.~Zetzsche.
\newblock Separability by piecewise testable languages and downward closures beyond subwords.
\newblock In {\em LICS 2018}, pages 929--938. ACM Press, 2018.

\end{thebibliography}

\end{document}